\documentclass[11pt]{ociamthesis}  

\usepackage{amssymb}

\usepackage{mymacros}

\title{Universal Algebra\\[1ex]     
        And Effectful Computation}   

\author{Nayan Rajesh}             

\degree{MSc in Mathematics and Foundations of Computer Science}     
\degreedate{Trinity 2023}         

\begin{document}

\baselineskip=15pt plus1pt

\setcounter{secnumdepth}{3}
\setcounter{tocdepth}{3}

\maketitle                  
\begin{abstract}
    Abstract clones serve as an algebraic presentation of the syntax of a simple type theory. From the perspective of universal algebra,  they define algebraic theories like those of groups, monoids and rings. An algebra for a clone is a realization of the theory; for instance,  an algebra for the algebraic theory of groups is an actual group. These links allow one to study the language of simple type theory from the viewpoint of universal algebra.

Programming languages, however, are much more complicated than simple type theory. Many useful features like reading, writing, and exception handling involve interacting with the environment; these are called side-effects. Algebraic presentations for languages with the appropriate syntax for handling effects are given by premulticategories and effectful multicategories \cite{staton-levy}. We study these structures with the aim of defining a suitable notion of an algebra.

To achieve this goal, we proceed in two steps. First, we define a tensor on $[\to,\category{Set}]$, and  show that this tensor along with the cartesian product gives the category a duoidal structure. Secondly, we introduce the novel notion of a multicategory enriched in a duoidal category which generalize the traditional notion of a multicategory. Further, we prove that an effectful multicategory is the same as a multicategory enriched in the duoidal category $[\to,\category{Set}]$. This result places multicategories and effectful multicategories on a similar footing, and provides a mechanism for transporting concepts from the theory of multicategories (which model pure computation) to the theory of effectful multicategories (which model effectful computation). As an example of this, we generalize the definition of a 2-morphism for multicategories to the duoidally enriched case. Our equivalence result then gives a natural definition of a 2-morphism for effectful multicategories, which we then use to define the notion of an algebra.
\end{abstract}

\begin{acknowledgements}
    	I begin by thanking my supervisor, Philip Saville. The idea of exploring the connections between universal algebra and effectful computation was his, and I am  very grateful for the time he spent in discussing and sharing his intuition on these subjects with me. 	
	Checking if $[\to,\mcal{V}]$ has a  funny-like tensor was his suggestion, and this formed the basis for most of the later parts of this dissertation. I cannot thank him enough for quickly providing me with his  comments and feedback on different drafts of chapters even though I got them to him quite late. 

I also thank Sam Staton for directing me to Philip when I approached him in March, and for agreeing to be my co-supervisor for this dissertation.

My time at Oxford has been a huge learning experience, and I thank everyone at Somerville and at the Mathematical Institute who was part of it. Many thanks go to Dan Ciubotaru and Carmen Constatin who influenced my intuition towards category theory with their excellent lectures. 	Special thanks go to Riju Dey and Erik Rydow who made  my day-to-day life memorable.

The final days of this dissertation were spent at home, and I thank my mother for supporting me, hearing me rant, and more importantly for not waking me up before 7am.  I thank my friends for keeping me sane: Darwin, Gautham, Jash, and Reyna  helped me take much-needed breaks. Lastly, I thank Akshaya who has always been my source of comfort.

\end{acknowledgements} 

\begin{romanpages}          
\tableofcontents            
\end{romanpages}            

\include{introduction}
\chapter{Introduction}
\section{Background}

Semantic models for  a formal language help us make sense of the language.  For example, given  the polynomial, $x^2 +5x - 6 \in \mathbb{R}[x]$, there are two ways to solve for a solution. The first is a syntactic approach, which involves using the rules of a ring (the language) and rewriting the polynomial as $(x-6)(x+1)$. Then, by the rules of an integral domain, we conclude that solutions are at $6$ and $-1$. The other, more semantic approach is to interpret the polynomial as specifying a graph on the plane, and reading off all the $x$-coordinates for which the $y$-coordinate is $0$. This leads to the general principle of using semantics to determine properties of the syntax. This approach has been used in mathematical logic - that a statement is derivable from a set of axioms can be proved by checking that it holds true in all models of those axioms, or can be disproved by showing that  it fails for one  model. Such an approach was used to show the independence of the parallel postulate of Euclidean geometry.
 In computer science, there are three major classes of semantic approaches that provide ways to determine whether two programs are equivalent:  operational, denotational and logical.

Universal algebra provides a framework for defining languages and their models. A signature for an algebraic theory consists of a set of operations, $O$ and a set of equations, $E$. Terms can be built up using the operations, and equality of terms is the smallest congruence that satisfies the set of equations. Noteworthy is the fact that two distinct pairs $(O,E)$ and $(O',E')$ of operations and equations may specify the same algebraic theory. An example of this is the algebraic theory of groups, and \cite{manes} gives three equivalent presentations of this theory. A presentation invariant approach is made possible by packaging the language into either (1) abstract clones, or (2) Lawvere theories. Each of these structures have a corresponding notion of  `algebra'. An algebra for an abstract clone, or for a  Lawvere theory is a model of the language, i.e, a domain which has the operations of the language, and satisfies the required equations. In summary, clones/Lawvere theories provide the grammar, and algebras for them are models for the grammar.

Simple type theories, which are very elementary  programming languages can be thought of as a (multi-sorted) abstract clone. Hence, abstract clones can be used to represent and study programming language syntax. This idea was taken as a starting point to define a type theory for cartesian closed bicategories in \cite{psaville}. Algebras for abstract clones, then,  give us models of this language, and this approach was used in \cite{hyland2012} to define  models for the untyped $\lambda$-calculus.

This correspondence between programming languages and abstract clones only holds when the language is assumed to be pure, and uses a call-by-name evaluation strategy. For effectful, call-by-value programming languages, it becomes necessary to weaken the definition of a clone, and to consider cartesian premulticategories and cartesian Freyd multicategories as defined in \cite{staton-levy}. These objects have not been studied from the perspective of universal algebra. Specifically, these structures do not have an appropriate notion of an algebra. 

The central preoccupation of this dissertation is to study  effectful, call-by-value structures, and to define  suitable notions of  `algebra' for them. Though cartesian premulticategories and cartesian Freyd multicategories do not make an appearance due to their complexity, we study and define algebras for premulticategories and effectful multicategories which are linear versions of these. 

\section{Outline}

\begin{itemize}
	\item \text{Chapter 2} provides a brief introduction to enriched category theory to serve as motivation for some of the enriched constructions in later chapters. The chapter motivates enrichment of categories using enrichment over graphs, and recalls some results proved from \cite{wolff}. 
	
	\item \text{Chapter 3} motivates two multi-ary structures for which the notion of an algebra is well-established: clones and multicategories. The first half of the chapter introduces the theory of clones and describes their connection to programming languages. Relationships between clones and categories, and between clones and cartesian categories are explored in the form of adjunctions between the respective categories. The notion of an algebra for a clone is defined, and its equivalence with standard definitions is shown. Multicategories are then motivated as linear versions of clones, and analogous ideas are presented: relationships with categories, and with monoidal categories, and  the notion of an algebra is introduced. The chapter ends with a general recipe to define algebras for multi-ary structures.
	
	\item \text{Chapter 4} introduces premulticategories and effectful multicategories. A description of their connection to effectful, call-by-value programming languages is given, followed by a brief exploration of their relationships to categories, premonoidal categories and effectful categories. A notion of algebra is put forward for premulticategories.
	
	\item \text{Chapter 5} explores the category of arrows and commutative squares, $[\to,\mcal{V}]$ for a cosmos. A symmetric monoidal structure is defined on this category, and some properties of this tensor are explored. 
 The chapter ends with the observation that this category has a duoidal structure. This tensor forms the basis for the equivalence result proved in the next chapter.
	
	\item \text{Chapter 6} generalizes the notion of a multicategory by introducing multicategories enriched in a duoidal category. Adjunctions with the category of enriched categories are explored, following which it is shown that multicategories enriched in the duoidal category from \text{Chapter 5} are the same as the effectful multicategories of \text{Chapter 4}. The chapter ends by showing that these structures can be collected into a 2-category and that a category of algebras for these can be defined.
	
	\item \text{Chapter 7} serves as a concluding chapter in which lines of further work are indicated.
\end{itemize}

\section{Contributions}

We highlight  the main contributions of this dissertation chapter-wise. Chapter 2 is skipped since it is just a brief run through of the ideas of enriched category theory.  

\begin{itemize}
	\item In Chapter 3, we explore the theory of clones. \begin{itemize}
		\item An explicit  construction of the free clone on a category is given. Though the existence of the free clone was shown in \cite{psaville}, this construction shows what the homsets should look like. This tells how how to freely construct a programming language from a signature provided by a category. 
		
		\item A construction of the free cartesian category on a clone is given.  This shows how product types can be freely added to a programming language.	The construction is a straightforward generalization of the construction of a Lawvere theory from a single-sorted clone, and is probably folklore.

		\item We introduce the  notion of a transformation between morphisms of clones, which does not seem to appear in the literature. We use this to define the category of algebras for a multi-sorted clone, and show its equivalence to the more familiar notion of an algebra for an algebraic theory as put forward in \cite{vitale}. Again, this equivalence is a straightforward generalization of the equivalence between the category of algebras for Lawvere theories and for single-sorted clones \cite{gould}. 
			\end{itemize}
			
			\item Chapter 4 is a quick run through of the theory of premulticategories as put forward in \cite{staton-levy}. Our contribution to this theory is the notion of a transformation between premulticategory morphisms. We illustrate that this is a sensible definition by showing that there is a category of premulticategory morphisms between fixed premulticategories and their transformations.
			
			\item In Chapter 5, we investigate the category of arrows and commutative squares, $[\to,\cal{V}]$, and define the funny tensor on it. Funny tensors formalize the notion of `sesquiness', and this one provides the appropriate notion to discuss the concept  of sesqui-substitution, which is an integral part of call-by-value multi-ary structures.   Finally, we prove  that  this tensor preserves colimits and induces a duoidal structure on the category.

			\item In Chapter 6, we introduce the novel notion of a multicategory enriched in a duoidal category. 
			\begin{itemize}
				\item Premonoidal and effectful categories are a generalization of monoidal categories, and premulticategories, and effectful multicategories are a generalization of multicategories. However, effectful categories and monoidal categories are also on the same footing as they are instances of  pseudomonoids in an appropriate bicategory \cite{hefford-roman}. 
    Hence, there is a general theory that can be applied to both these structures. We provide a multi-ary analogue by showing that effectful multicategories and multicategories are instances of a multicategory enriched in a duoidal category. As a result, it becomes easier to see how results from the well-established area of multicategories should be passed to the theory of effectful categories. 
				
				\item As an example of how structure can be passed, we show that the notion of a transformation of multicategories can be generalized to yield a 2-category of duoidally enriched multicategories. Then, this notion specializes to yield a valid definition of a 2-morphism of effectful multicategories. 

                \item (Single-sorted) clones and multicategories are equivalent to the category of  monoids in a certain monoidal category \cite{coend}. This equivalence makes heavy use of the fact that substitution in these structures is simultaneous. A similar characterization for premulticategories and effectful categories was blocked by the fact that subsitution in these structures cannot be made simultaneous. Our equivalence result seems to eliminate this roadblock, and provides a framework for investigating this idea.
				
				\item Finally, we use the notion of 2-morphism between effectful categories to define a category of algebras for these structures.
			\end{itemize}
\end{itemize}

\chapter{What are Enriched Categories?}

	As observed by Power in \cite{premon-alg}, effectful, call-by-value computation is best considered in the context of categories enriched over the cartesian monoidal category $[\to,\category{Set}]$, the category of arrows and commutative squares in $\category{Set}$. We consider the enriched viewpoint in later chapters, and we briefly recollect some concepts that will be alluded to later on.

	\section{Enriched Graphs and Cats}
	
Recall that a (directed) graph, $\mcal{C}$ consists of the following information:
\begin{itemize}
	\item A set of vertices, $A,B,\ldots, $
	\item For every pair of vertices, a set of arrows, $\Hom{\mcal{C}}{A}{B} \in \category{Set}$.
\end{itemize}
A (small) category, then, is a \textit{graph with extra structure}, i.e, a category is a graph that has identity arrows,
\[
id_A: 1\to \Hom{\mcal{C}}{A}{A}
\]
and the ability to compose arrows:
\begin{align*}
	\circ_{A,B,C}: \Hom{\mcal{C}}{B}{C}\times \Hom{\mcal{C}}{A}{B}&\to \Hom{\mcal{C}}{A}{C}\\
f,g& \mapsto f\circ g
\end{align*}
subject to the the unitality and associativity axioms. Then, the category of small categories, $\category{Cat}$ has a forgetful functor, $U$ to the category of graphs and their morphisms $\category{Graph}$. Given a graph, $\mcal{C}$, it is possible to construct the \textit{free category} on it, with the same objects, and morphisms from $A$ to $B$ corresponding to \textit{paths} from $A$ to $B$ in the graph \cite{awodey}. This leads to the well-known (monadic) adjunction, 
\[\begin{tikzcd}[cramped]
	{\category{Cat}} && {\category{Graph}}
	\arrow[""{name=0, anchor=center, inner sep=0}, "U"', shift right=3, from=1-1, to=1-3]
	\arrow[""{name=1, anchor=center, inner sep=0}, "F"', shift right=3, from=1-3, to=1-1]
	\arrow["\dashv"{anchor=center, rotate=-90}, draw=none, from=1, to=0]
\end{tikzcd}\]

The notion of an enriched graph comes by  considering an \textit{object} of arrows, $\Hom{\mcal{C}}{A}{B}$ in some fixed category, as opposed to in \category{Set}. 
\begin{definition}
	Let $\mcal{V}$ be a category, which we call the enriching category. Then, a $\mcal{V}$-graph, $\mcal{C}$ consists of the following information:
	\begin{itemize}
		\item A set of objects, $ \ob{\mcal{C}}$
		\item For every pair of objects $A,B$, an object $\Hom{\mcal{C}}{A}{B} \in \cal{V}$.
	\end{itemize}
	A morphism 	of $\mcal{V}$-graphs, $f: \mcal{C}\to \mcal{D}$ consists of the following data:
	\begin{enumerate}
		\item A function, $f: \ob{\mcal{C}}\to \ob{\mcal{D}}$
		\item A family of maps, $f_{A,B}: \Hom{\mcal{C}}{A}{B} \to \Hom{\mcal{D}}{fA}{fB} $
	\end{enumerate}
	Thus, we obtain a category, $\mcal{V}$-\category{Graph} of $\mcal{V}$-graphs and their morphisms.
\end{definition}

Our usual notion of a graph is a graph enriched in $\category{Set}$, a \category{Set}-graph. Notice that the definition of a category used the cartesian monoidal structure of \category{Set}, a structure  not guaranteed to exist in arbitrary categories. Hence, to define an enriched category, we ask  that the enriching category has some monoidal structure, $(\mcal{V},\otimes, I, \lambda, \rho, \alpha)$.

\begin{definition}
	A $\mcal{V}$-category is a $\mcal{V}$-graph, $\mcal{C}$ with the following additional information:
	\begin{enumerate}
		\item For every object, $A\in \ob{\mcal{C}}$, an identity arrow:
$		id_{A}: I \to \Hom{\mcal{C}}{A}{A}$ 
		
		\item For every triple of objects, $A,B,C\in \ob{\mcal{C}}$, a composition operation:
		$\circ_{A,B,C}: \Hom{\mcal{C}}{B}{C}\otimes \Hom{\mcal{C}}{A}{B}\to \Hom{\mcal{C}}{A}{C}
		$
	\end{enumerate}
	subject to the following axioms
	\begin{enumerate}
		\item Right unitality:
	\[\begin{tikzcd}[cramped]
		{\Hom{\mcal{C}}{B}{B}\otimes \Hom{\mcal{C}}{A}{B}} && {\Hom{\mcal{C}}{A}{B}} \\
		\\
		{I\otimes \Hom{\mcal{C}}{A}{B}}
		\arrow["{\circ_{A,B,B}}", from=1-1, to=1-3]
		\arrow["{id_{B}\otimes 1}", from=3-1, to=1-1]
		\arrow["\lambda"', from=3-1, to=1-3]
	\end{tikzcd}\]
		\item Left unitality: 
		\[\begin{tikzcd}[cramped]
			{\Hom{\mcal{C}}{A}{B}\otimes \Hom{\mcal{C}}{A}{A}} && {\Hom{\mcal{C}}{A}{B}} \\
			\\
			{\Hom{\mcal{C}}{A}{B}\otimes I}
			\arrow["{\circ_{A,A,B}}", from=1-1, to=1-3]
			\arrow["{1\otimes  id_{A}}", from=3-1, to=1-1]
			\arrow["\rho"', from=3-1, to=1-3]
		\end{tikzcd}\]
		\item Associativity:
		\[\begin{tikzcd}[cramped]
			{(\Hom{\mcal{C}}{C}{D}\otimes \Hom{\mcal{C}}{B}{C})\otimes \Hom{\mcal{C}}{A}{B}} && {\Hom{\mcal{C}}{C}{D}\otimes (\Hom{\mcal{C}}{B}{C}\otimes \Hom{\mcal{C}}{A}{B})} \\
			{\Hom{\mcal{C}}{B}{D}\otimes \Hom{\mcal{C}}{A}{B}} && {\Hom{\mcal{C}}{C}{D}\otimes \Hom{\mcal{C}}{B}{C}} \\
			& {\Hom{\mcal{C}}{A}{D}}
			\arrow["\alpha", from=1-1, to=1-3]
			\arrow["{\circ_{B,C,D}\otimes 1}", from=1-1, to=2-1]
			\arrow["{\circ_{A,B,D}}"', from=2-1, to=3-2]
			\arrow["{1\otimes \circ_{A,B,C}}", from=1-3, to=2-3]
			\arrow["{\circ_{B,C,D}}", from=2-3, to=3-2]
		\end{tikzcd}\]
	\end{enumerate}
If $\mcal{C}$ and $\mcal{D}$ are $\mcal{V}$-categories, a $\mcal{V}$-functor between them is a $\mcal{V}$-graph morphism, $f: \mcal{C}\to \mcal{D}$ satisfying the following conditions:
\begin{enumerate}
	\item Identities are preserved: 
\[\begin{tikzcd}[cramped]
	{\Hom{\mcal{C}}{A}{A}} && {\Hom{\mcal{D}}{fA}{fA}} \\
	\\
	I
	\arrow["{id_A^{\mcal{C}}}", from=3-1, to=1-1]
	\arrow["{f_{A,A}}", from=1-1, to=1-3]
	\arrow["{id_{fA}^{\mcal{D}}}"', from=3-1, to=1-3]
\end{tikzcd}\]	
	\item Composition is preserved:
	\[\begin{tikzcd}[cramped]
		{\Hom{\mcal{C}}{B}{C}\otimes \Hom{\mcal{C}}{A}{B}} && {\Hom{\mcal{D}}{fB}{fC}\otimes \Hom{\mcal{D}}{fA}{fB}} \\
		{\Hom{\mcal{C}}{A}{C}} && {\Hom{\mcal{D}}{fA}{fC}}
		\arrow["{f_{B,C}\otimes f_{A,C}}", from=1-1, to=1-3]
		\arrow["{\circ^{\mcal{C}}_{A,B,C}}", from=1-1, to=2-1]
		\arrow["{f_{A,C}}", from=2-1, to=2-3]
		\arrow["{\circ^{\mcal{D}}_{fA,fB,fC}}", from=1-3, to=2-3]
	\end{tikzcd}\]
\end{enumerate}
Then, $\mcal{V}$-categories and $\mcal{V}$-functors assemble into a category, $\mcal{V}$-\category{Cat}. 
\end{definition}

A (small) category is a $\mcal{V}$-category where $\mcal{V} = \category{Set}$, with monoidal structure  given by the cartesian structure on \category{Set}. To replicate the adjunction $F\dashv U$ between \category{Set} and \category{Graph} for $\mcal{V}$-categories, we require further structure on $\mcal{V}$ - the ability to take disjoint unions. 

As mentioned before, if $\mcal{C}$ is a graph, the category $F\mcal{C}$ has paths from $A$ to $ B$ as morphisms $A\to B$. The set of paths $A\to A_1\to A_2\to\ldots\to A_{n-1}\to B$ can be written as a product, 
\[
\Hom{F\mcal{C}}{A}{B}_{\Gamma}: = \Hom{\mcal{C}}{A}{A_1}\times \Hom{\mcal{C}}{A_1}{A_2}\times\ldots \times \Hom{\mcal{C}}{A_{n-1}}{B}
\] 
where $\Gamma = A_1,\ldots, A_{n-1}$.
The set of all paths from $A$ to $B$ is then obtained by taking the disjoint union of paths of arbitrary length:
\[
\Hom{F\mcal{C}}{A}{B} = \sum_{\Gamma \in \ob{\mcal{C}}^*} \Hom{F\mcal{C}}{A}{B}_\Gamma
\]
When $A=B$, we add in a special element, $1=\{*\}$ to stand for the identity. If the enriching category $\mcal{V}$ has coproducts, is is possible to carry out the exact same construction. However, to define composition, we require that the coproducts are preserved by $-\otimes B$ for every $B\in \mcal{C}$, i.e, the canonical map
\[
(\sum_{i\in I} (A_i\otimes B))\to (\sum_{i\in I} A_i)\otimes B
\]
for any $I$-indexed collection, $\{A_i\}_{i\in I}$, is an isomorphism.

\begin{notation}
	Let $\mcal{C}$ be a $\mcal{V}$-graph, and let $\Gamma = A_1,\ldots, A_n$ be a sequence of objects in it. Then, define:
	\[
	\Hom{\mcal{C}}{A,\Gamma}{B} : = \Hom{\mcal{C}}{A}{A_1}\otimes \Hom{\mcal{C}}{A_1}{A_2}\otimes\ldots \Hom{\mcal{C}}{A_n}{B}
	\]
\end{notation}

\begin{definition}\label{free-cat}
	Assume that $\mcal{V}$ is a monoidal category with coproducts that distribute over the tensor. Let $ \mcal{C}$ be a $\mcal{V}$-graph, then define the $\mcal{V}$-category $F\mcal{C}$ as follows:
	\begin{enumerate}
		\item Objects: $\ob{F\mcal{C}} = \ob{\mcal{C}}$
		\item Hom-Objects: 
		\[
		\Hom{F\mcal{C}}{A}{B} : = \begin{cases}
			\sum_{\Gamma\in \ob{\mcal{C}^*}} \Hom{\mcal{C}}{A,\Gamma}{B}, & \text{if $A\neq B$}\\
					\sum_{\Gamma\in \ob{\mcal{C}^*}} \Hom{\mcal{C}}{A,\Gamma}{B}+I &\text{if $A=B$}
		\end{cases}
		\]
		\item Identities: Define the identity arrow, $id_A: I\to \Hom{\mcal{C}}{A}{A}$ to be the canonical inclusion, $I\to 	\sum_{\Gamma\in \ob{\mcal{C}^*}} \Hom{\mcal{C}}{A,\Gamma}{B}+I $.
		
		\item Composition: Composition is just a matter of `concatenating', for example when $A\neq B$  and $B\neq C$,  
\[\begin{tikzcd}[cramped]
	{(\sum_{\Gamma} \Hom{\mcal{C}}{A,\Gamma}{B} ) \otimes (\sum_{\Gamma'} \Hom{\mcal{C}}{B,\Gamma'}{C})} && {\sum_{\Gamma} \Hom{\mcal{C}}{A,\Gamma}{C}} \\
	{\sum_{\Gamma,\Gamma'} \Hom{\mcal{C}}{A, (\Gamma,B,\Gamma')}{ C}}
	\arrow["\cong", from=1-1, to=2-1]
	\arrow["\circ", dashed, from=1-1, to=1-3]
	\arrow["\iota"', hook, from=2-1, to=1-3]
\end{tikzcd}\]
where $\cong$ is the map that is obtained by the preservation of coproducts by $-\otimes B$, and $\iota$ is the canonical inclusion into the coproduct. 
	When $A=B$ or $B = C$ or both, similar composites that account for the identities can be defined.   
	\end{enumerate} 
	Finally, define $\eta_{\mcal{C}}: \mcal{C}\to F\mcal{C}$ to be the identity-on-objects $\mcal{V}$-graph morphism where
	\[
	\Hom{\mcal{C}}{A}{B} = \Hom{\mcal{C}}{A,\emptyset}{B}\hookrightarrow \sum_{\Gamma} \Hom{\mcal{C}}{A,\Gamma}{B} = \Hom{F\mcal{C}}{A}{B}
	\]
	is the canonoical inclusion.
\end{definition}

This construction is a free construction, and defines a monadic adjunction between $\mcal{V}$-\category{Cat} and $\mcal{V}$-\category{Graph}, and this was proved by Wolff \cite{wolff}.

\begin{proposition}
	Assume that $\mcal{V}$ is symmetric monoidal closed, and cocomplete. The construction $F$ in Definition \ref{free-cat} extends to a functor, and is left adjoint to the forgetful functor $U$ which sends a $\mcal{V}$-category to its underlying $\mcal{V}$-graph. 
	\[\begin{tikzcd}[cramped]
		{\mcal{V}-\category{Cat}} && {\mcal{V}-\category{Graph}}
		\arrow[""{name=0, anchor=center, inner sep=0}, "U"', shift right=3, from=1-1, to=1-3]
		\arrow[""{name=1, anchor=center, inner sep=0}, "F"', shift right=3, from=1-3, to=1-1]
		\arrow["\dashv"{anchor=center, rotate=-90}, draw=none, from=1, to=0]
	\end{tikzcd}\]
	Further, the adjunction is monadic, and this exhibits $\mcal{V}$-\category{Cat} as an algebraic structure over $\mcal{V}$-\category{Graph}.
\end{proposition}

Given functors, $F,G: \mcal{C}\to \mcal{D}$ between categories, there is a notion of a natural transformation, $\eta: F\Rightarrow G$ between them. This idea can also be lifted to $\mcal{V}$-categories
	
\begin{definition}\cite{kelly}
	Let $f,g: \mcal{C}\to \mcal{D}$ be $\mcal{V}$-functors. A $\mcal{V}$-natural transformation between them consists of a family of arrows, $\{\eta_A: I\to \Hom{\mcal{D}}{fA}{gA}\}_{A\in \ob{\mcal{C}}}$ such that the following diagram commutes: 
	\[\begin{tikzcd}[cramped]
		{\Hom{\mcal{C}}{A}{B}\otimes I} && {\Mhom{\mcal{D}}{gA}{gB}\otimes \Hom{\mcal{D}}{fA}{gA}} \\
		{\Hom{\mcal{C}}{A}{B}} && {\Hom{\mcal{D}}{fA}{gB}} \\
		{I\otimes \Hom{\mcal{C}}{A}{B}} && {\Hom{\mcal{D}}{fB}{gB}\otimes \Hom{\mcal{D}}{fA}{fB}}
		\arrow["\cong"', from=2-1, to=1-1]
		\arrow["\cong", from=2-1, to=3-1]
		\arrow["{g\otimes \eta_A}", from=1-1, to=1-3]
		\arrow["{\eta_B\otimes f}"', from=3-1, to=3-3]
		\arrow["\circ"', from=3-3, to=2-3]
		\arrow["\circ", from=1-3, to=2-3]
	\end{tikzcd}\]
\end{definition}	

Clearly, when $\mcal{V} = \category{Set}$, we recover the usual notion of a natural transformation. It is possible to vertically and horizontally compose $\mcal{V}$-natural transformations. If $f,g,h:\mcal{C}\to \mcal{D}$ are $\mcal{V}$-functors, and $\eta: f\Rightarrow g$ and $\epsilon: g\Rightarrow h$ are $\mcal{V}$-natural transformations, define their vertical composite, $(\epsilon\cdot \eta):f\Rightarrow h$ to have components:
\[
(\epsilon\ast\eta)_A: I\xrightarrow{\cong} I\otimes I \xrightarrow{\epsilon_A\otimes \eta_A} \Hom{\mcal{D}}{gA}{hA} \otimes \Hom{\mcal{D}}{fA}{gA} \xrightarrow{\circ} \Mhom{\mcal{D}}{fA}{hA}
\]
This enables one to construct the enriched functor category in $\mcal{V}\text{-}\category{Cat}(\mcal{C},\mcal{D})$, and shall serve as a starting point for our definition of a transformation in Chapter \ref{chapter6}.

\chapter{Multi-ary Structures for Pure Programming}\label{chapter3}
	\section{Abstract Clones and their Algebras}
	
	\subsection{Clones Model Programming Languages}

	The simply typed $\lambda$-calculus is the prototypical functional programming language - a deduction, $x_1:A_1,\ldots, x_n:A_n\vdash M:B$ in the $\lambda$-calculus corresponds to a program $M$ containing free variables among $x_1,\ldots, x_n$. On the other hand, a deduction in the $\lambda$-calculus can also be given `meaning' as a morphism in a cartesian closed category \cite{ls86}. These correspondences highlight the principle of computational trinitarianism which considers type theories, programs and categories as three faces of the same phenomenon \cite{harper}.

	A category with no extra structure models what Moggi refers to as `many-sorted monadic equational logic' \cite{moggi91}. These are  programming languages which allow for the construction of programs $x:A\vdash M:B$ with exactly one free variable. Such a language comes equipped with a set of base types, such as ${A},{nat},{bool},\ldots $
	and a set of unary function symbols/programs, $f: {A}\to {B}, \quad {succ}:{nat}\to {nat}$. It is always possible to construct the `do nothing' program: 
	\begin{prooftree}
		\AxiomC{$\vdash A: \text{Type}$}
		\UnaryInfC{$x:A\vdash x:A$}
	\end{prooftree}
	The output of one program can be fed into the input of another,
	\begin{prooftree}
		\AxiomC{$x:A \vdash M:B$}
		\AxiomC{$y:B\vdash N:C$}
		\BinaryInfC{$x:A\vdash N[M]: C$}
	\end{prooftree}		
	These are subject to the rules of a category, which mimic some intuitive properties of a programming language: feeding the output of a program into the `do nothing' program returns the first program, etc.

	This language has minimal expressive power, and can only model programs with one free variable. However, programs typically have more than one free variable, and hence we seek to model a language that can model deductions containing finitely many free variables:
	\[
	x_1:A_1,\ldots, x_n:A_n\vdash M:B
	\]
	It is always possible to `project' free variables, i.e, there is a program in $n$ free variables which returns the $j$th free variable:
	\begin{prooftree}
		\AxiomC{}
		\UnaryInfC{$x_1:A_1, \ldots, x_n: A_n\vdash x_j: A_j$}
	\end{prooftree}
	The language also allows for substitution - if there is a program in $m$ free variables  and if there are $m$ other programs that output values in the correct types, it is possible to obtain a new program:
	\begin{prooftree}
		\AxiomC{$x_1:B_1,\ldots, x_m: B_m \vdash M:C $}
		\AxiomC{$\{y_1:A_1,\ldots, y_n: A_n\vdash N_i: B_i\}_{1\leq i\leq m} $}
		\BinaryInfC{$y_1:A_1,\ldots, y_n: A_n\vdash M[ N_1, \ldots, N_m]: C$}
	\end{prooftree}

	To model a deduction with finitely many free variables  in a category, we need to be able to interpret contexts, $x_1:A_1,\ldots, x_n:A_n$ as  objects. This is done by requiring the category  to have finite products, in which case, a deduction 
$	x_1: A_1,\ldots, x_n: A_n\vdash M:B$
	is interpreted as a morphism, 	
$	A_1\times \ldots\times A_n \to B$.

	However, in doing this, the language is forced to have product types. So we ask the following question: is there a structure that is strictly stronger than a category, in that it can interpret programs with finitely many free variables, but strictly weaker than a category with finite products, so that the language can do without having product types? This middle ground is provided by the notion of an abstract clone.

	An abstract clone can be thought of as an algebraic re-interpretation of the language described above. Just like how a category has morphisms, $f: A\to B$, which correspond to deductions with a unary context, $x:A \vdash M:B$, an abstract clone has \textit{multi-morphisms}, $f: A_1,\ldots, A_n\to B$, which correspond to deductions with a context of finite length, $x_1:A_1,\ldots, x_n:A_n\vdash M:B$. Instead of identities and composition, a clone has projections, $A_1,\ldots, A_n\to A_j$, and substitution, $t[u_1,\ldots, u_n]$.

	\begin{definition}\cite{cohn}
		An abstract clone $\bb{C}$ consists of the following data: 
		\begin{itemize}
			\item A set of objects, $\ob{\bb{C}}$,
			\item For every sequence of objects, $A_1,\ldots, A_n, B$, a set of multimaps, $\Mhom{\bb{C}}{A_1,\ldots, A_n}{B}$,
			A multimap $t\in \Mhom{\bb{C}}{A_1,\ldots, A_n}{B}$ will be denoted $t: A_1,\ldots, A_n\to B$,
			\item Given $n>0$, and a sequence of objects, $A_1,\ldots, A_n$ and some index $j \in \{1,\ldots, n\}$, a projection map,
			$pr^n_j:	1\to \Mhom{\bb{C}}{A_1,\ldots, A_n}{A_j}
			$,
			\item A substitution operation,
			
			\begin{align*}
				\Mhom{\bb{C}}{B_1,\ldots, B_m}{C}\times \prod_{i=1}^{m} \Mhom{\bb{C}}{A_1,\ldots, A_n}{B_i}&\to \Mhom{\bb{C}}{A_1,\ldots, A_n}{C}\\
				t, u_1,\ldots, u_n&\mapsto t[u_1,\ldots, u_n]
			\end{align*}			
		\end{itemize}
		subject to the following axioms: 
		\begin{enumerate}
			\item Left unitality: For a collection of terms, $\{u_i: A_1,\ldots, A_n\to B_i\}_{1\leq i\leq m}$,
			\[
			pr_{B_\bullet}^j [u_1,\ldots, u_n] = u_j
			\]
			
			\item Right unitality: 
			For any term $t: B_1,\ldots, B_m\to C$:
			\[
			t[pr_{B_\bullet}^1,\ldots, pr_{B_\bullet}^n] = t
			\]

			\item Associativity: For $t: C_1,\ldots, C_l\to D$, $\{u_i: B_1,\ldots, B_m\to C_i\}_{1\leq i\leq l}$ and $\{v_i: A_1,\ldots, A_n\to B_i\}_{1\leq i\leq m}$, the following equality holds: 
			\[
			t[u_1,\ldots, u_l][v_1,\ldots, v_m] = t[u_1[v_1,\ldots, v_m],\ldots, u_l[v_1,\ldots, v_m]]
			\]	
		\end{enumerate}
	\end{definition}
	
	It is straightforward to see how the language described can be interpreted in an abstract clone. Clones make easier to prove derived rules in the language. For example, the language described above allows for free variables in programs to be renamed. Formally, a \textit{renaming of contexts}, $\rho: (x_1:A_1,\ldots, x_n:A_n)\to (y_1:B_1,\ldots, y_m:B_m)$ is a mapping, $\rho: \{1,\ldots, n\}\to \{1,\ldots, m\}$ such that $A_{i} = B_{f(i)}$ for all $i\in \{1,\ldots,n\}$.
	The language described above (and the $\lambda$-calculus) has the following derived rule, if there exists a renaming $\rho: (x_1:A_1,\ldots, x_n:A_n)\to (y_1:B_1,\ldots, y_m:B_m)$:
	\begin{prooftree}
		\AxiomC{$x_1:A_1,\ldots, x_n:A_n \vdash M:B$}
		\UnaryInfC{$y_1:B_1,\ldots, y_m:B_m\vdash M^\rho:B$}
	\end{prooftree}
	Intuitively, this is saying that if a program can be constructed using the free variables $x_1,\ldots, x_n$, and these variables are among $y_1,\ldots, y_m$, then the program can also be constructed using $y_1,\ldots, y_m$.

	To prove this rule with only the deduction rules in language would involve a tedious induction proof. However, with clones, this corresponds to saying that if there is a multimap, $t: A_1,\ldots, A_n\to B$, and there is a renaming, $\rho: A_1,\ldots, A_n\to B_1,\ldots, B_m$, then, it is possible to construct a term, $t^\rho: B_1,\ldots, B_m \to B$. To do this, note that for all $i\in \{1,\ldots n\}$, there is a projection, 
	\[
	pr^{\rho(i)}_{B_\bullet}: B_1,\ldots, B_m \to B_{\rho(i)} = A_i
	\]
	So, we can form the following term:
	\[
	t[pr^{\rho(1)}_{B_\bullet},\ldots, pr^{\rho(n)}_{B_\bullet}]: B_1,\ldots, B_m\to B
	\]
	which gives us the required `renamed program'.

	Every clone determines a category. The essence of this mapping is that a programming language that can construct programs with finitely many free variables can also construct programs with exactly one free variable. 
	For a clone, $\bb{C}$, define  $ \overline{\bb{C}}$ as the category with the same objects, $\ob{\overline{\bb{C}}} := \ob{\bb{C}}$, and hom-sets given by $\Hom{\overline{\bb{C}}}{A}{B} := \Mhom{\bb{C}}{A}{B}$. Identities, $id_A\in \Hom{\overline{\bb{C}}}{A}{A}$ are given by the projection map, $pr^1_{A}\in \Mhom{\bb{C}}{A}{A}$, and composition, $f\circ g$ is defined as $f[g]$. We require the notion of a clone morphism to show that this mapping extends to a functor between the respective categories.

	\begin{definition}
		A clone morphism, $f: \bb{C}\to \bb{D}$ consists of the following data:
		\begin{itemize}
			\item A function, $f_{ob}: \ob{\bb{C}}\to \ob{\bb{D}}$. We just write $f(A)$ instead of $f_{ob}(A)$.
			\item A family of morphisms, $f_{\Gamma;A}: \Mhom{\bb{C}}{A_1,\ldots, A_n}{B}\to \Mhom{\bb{D}}{fA_1,\ldots, fA_n}{fB}$.
		\end{itemize}
		such that 
		\begin{enumerate}
			\item Projections are preserved: For $pr^j_{A_\bullet}: A_1,\ldots, A_n\to A_j$,
			\[
			f(pr^j_{A_\bullet}) = pr^j_{FA_\bullet}
			\]
			\item Substitution is preserved: Given $t: B_1,\ldots, B_m\to C$ and $\{u_i: A_1,\ldots, A_n\to B_i\}_{1\leq i\leq m}$,
			\[
			f(t[u_1,\ldots, u_m])  = f(t)[f(u_1),\ldots, f(u_m)]
			\]
		\end{enumerate}
	\end{definition}
	
	Clearly, there is a category of clones and their morphisms, $\category{Clone}$. A clone morphism, $f: \bb{C}\to \bb{D}$, restricts to a functor, $\overline{f}: \overline{\bb{C}}\to \overline{\bb{D}}$ just by virtue of definitions, and this yields a functor, $\overline{(-)}: \category{Clone}\to \category{Cat}$ to the category of small categories.

	\begin{lemma}
		The functor $\overline{(-)}: \category{Clone}\to \category{Cat}$ which restricts a clone to its unary contexts has a left adjoint, $\mcal{F}$, such that for $\overline{\mcal{F}\mcal{C}} \cong \mcal{C}$ for all categories $\mcal{C}$. 
	\end{lemma}
	\begin{proof}
		For any category $\mcal{C}$, define the clone $\mcal{F}\mcal{C}$ as 
		\begin{enumerate}
			\item Objects: $\ob{\mcal{F}\mcal{C}} = \ob{\mcal{C}}$. 
			\item Hom-sets: 		
			\begin{align*}
				\Mhom{\mcal{F}\mcal{C}}{A_1,\ldots, A_n}{B} : &= \sum_{i=1}^{n} \Hom{\mcal{C}}{A_i}{B}\\
				& = \bigcup_{i=1}^n \{(i,f) \mid f\in \Hom{\mcal{C}}{A_i}{B}\}
			\end{align*}
			\item Projections: Given $A_1,\ldots, A_n$, the projection $pr_{A_\bullet}^j$ is defined to be the element $(j, id_{A_j}) \in \Mhom{\mcal{F}\mcal{C}}{A_1,\ldots, A_n}{A_j}$.

			\item Substitution: Substitution is defined as follows:
			\begin{align*}
				\Mhom{\mcal{F}\mcal{C}}{B_1,\ldots, B_m}{C}\times \prod_{i=1}^{m} \Mhom{\mcal{F}\mcal{C}}{A_1,\ldots, A_n}{B_i} &\to \Mhom{\mcal{F}\mcal{C}}{A_1,\ldots, A_n}{C}\\
				(j,t),   (k_1, u_1),\ldots, (k_m, u_m) &\mapsto (  k_j , t\circ u_j    )
			\end{align*}	
			Note that this operation is well-defined: $t$ is a morphism in $\Hom{\mcal{C}}{B_j}{C}$ where $j\in \{1,\ldots, m\}$, and $u_j$ is a morphism in $\Hom{\mcal{C}}{A_{k_j}}{B_j}$, so $t\circ u_j$ is a map in $\Hom{\mcal{C}}{A_{k_j}}{C}$. 
		\end{enumerate}
		
		We verify the axioms of a clone: 
		\begin{enumerate}
			\item Left unitality: Let $\{(k_i,u_i) \in \Mhom{\mcal{F}\mcal{C}}{A_1,\ldots, A_n}{B_i}\}_{1\leq i\leq m} $ be a collection of multimorphisms. Then,
			\begin{align*}
				pr^j_{B_\bullet} [ (k_1,u_1),\ldots, (k_m, u_m)] & = (j,id_{A_j})[(k_1,u_1),\ldots, (k_m,u_m)]\\
				& = (k_j,  id_{A_j}\circ u_j) = (k_j, u_j)
			\end{align*}
			
			\item Right unitality: Let $(j,t) \in \Mhom{\mcal{F}\mcal{C}}{B_1,\ldots, B_m}{C}$. Then, 
			\begin{align*}
				(j,t)[pr^1_{B_\bullet},\ldots, pr^m_{B_\bullet}]& = (j,t)[(1,id_{A_1}),\ldots, (m,id_{A_m})]\\
				& = (j, t\circ id_{A_j})\\
				& = (j, t)
			\end{align*}
			
			\item Associativity: Let $(j,t)\in \Mhom{\mcal{F}\mcal{C}}{C_1,\ldots, C_l}{D}$, $\{(k_i,u_i) \in \Mhom{\mcal{F}\mcal{C}}{B_1,\ldots, B_m}{C_i}\}_{1\leq i\leq l}$ and $\{(p_i, v_i) \in \Mhom{\mcal{F}\mcal{C}}{A_1,\ldots,A_n}{B_i}\}_{1\leq i\leq m}$. Then, 
			\begin{align*}
				&(j,t)[(k_1,u_1),\ldots, (k_l,u_l)][(p_1,v_1),\ldots, (p_m,v_m)]\\
				& = (k_j, t\circ u_l) [(p_1,v_1),\ldots, (p_m), (v_m)]\\
				& = (p_{k_j}, (t\circ u_l)\circ v_{k_j})\\
				& = (p_{k_j}, t\circ (u_l\circ v_{k_j}))\\
				& = (j,t)[(p_{k_1}, u_1 \circ v_{k_1} ), \ldots, (p_{k_l}, u_l\circ v_{k_l}) ]\\
				& = (j,t)[(k_1,u_1)[(p_1,v_1),\ldots, (p_m,v_m)],\ldots, (k_l,u_l)[(p_1,v_1),\ldots, (p_m,v_m)]   ]
			\end{align*}
		\end{enumerate}
		
		The clone $\mcal{F}\mcal{C}$ has a special property: Let $A_1,\ldots, A_n,B$ be given, and assume $t\in \Hom{\mcal{C}}{A_j,B}$. Then, there are two ways this morphism can be thought of  as a multimorphism in $\mcal{F}\mcal{C}$:
		\begin{align*}
			(1,t)& \in \Hom{\mcal{F}\mcal{C}}{A_j}{B}\\
			(j,t) &\in \Hom{\mcal{F}\mcal{C}}{A_1,\ldots, A_n}{B}
		\end{align*}
		These are related as follows:
		\begin{align*}
			(1,t)[pr^j_{A_1,\ldots, A_n}] & = (1,t)[(j,id_{A_j})] = (j,t\circ id_{A_j}) = (j,t)
		\end{align*}
		This is the property of freeness. As observed before, if $f: A_1,\ldots, A_n\to B$ is  a morphism, and if $B_1,\ldots, B_m$ is a renaming of $A_1,\ldots, A_n$, it is possible to construct a renaming, $f'$. Since $(1,t): A_j\to B$ was a morphism, it was possible to construct a renaming, $(1,t)[pr^j_{A_1,\ldots, A_n}]: A_1,\ldots, A_n\to B$. The above property is saying that all multimaps $A_1,\ldots, A_n\to B$ are renamings of some multimap, $A_j\to B$. In other words, the set $\Mhom{\mcal{F}\mcal{C}}{A_1,\ldots, A_n}{B}$ does not have any more terms than it has to have.

		The category $\overline{\mcal{F}\mcal{C}}$ has:
		\begin{enumerate}
			\item Objects: $\ob{\overline{\mcal{F}\mcal{C}}} = \ob{\mcal{F}\mcal{C}} = \ob{\mcal{C}}$
			\item Hom-sets: 
			$		\Hom{\overline{\mcal{F}\mcal{C}}}{A}{B} = \Mhom{\mcal{F}\mcal{C}}{A}{B} = \sum_{i=1}^1\Hom{\mcal{C}}{A}{B} = \{(1,f) \mid f\in \Hom{\mcal{C}}{A}{B}\}$
		\end{enumerate}	
		So, we recover $\mcal{C}$ as $\overline{\mcal{F}\mcal{C}}$. Define $\eta_{\mcal{C}}: \mcal{C}\to \overline{\mcal{F}\mcal{C}}$ to be the obvious identity-on-objects isomorphism that sends   $f\mapsto (1,f)$.
		
		For universality, assume that there is a clone, $\bb{C}$, with a functor, $F: \mcal{C}\to \overline{\bb{C}} $. Then, define the clone morphism, $F^*: \mcal{F}\mcal{C}\to \bb{C}$ as:
		\begin{enumerate}
			\item On objects: $F^*(A) = F(A)$
			\item On Homs:
			\begin{align*}
				\Mhom{\mcal{F}\mcal{C}}{A_1,\ldots, A_n}{B} &\to \Mhom{\bb{C}}{FA_1,\ldots, FA_n}{FB}\\
				(j,f) & \mapsto F(f) [pr^j_{FA_\bullet}]
			\end{align*}
			The expressions type check: The map $f$ is in $\Hom{\mcal{C}}{A_j}{B}$, so $F(f) \in \Hom{\overline{\bb
					C}}{FA_j}{FB}$. The map, $pr^j_{FA_\bullet} $ belongs to $\Mhom{\bb{C}}{FA_1,\ldots, FA_n }{FA_j}$, so it is possible to substitute $pr^j_{FA_\bullet}$ into $F(f)$. 
		\end{enumerate}	Then, $\overline{F^*} \circ \eta_{\mcal{C}}= F$, and hence  makes the following diagram commute: 
		\[\begin{tikzcd}[cramped]
			{\overline{F\mcal{C}}} && {\overline{\bb{C}}} \\
			\\
			{\mcal{C}}
			\arrow["{\eta_{\mcal{C}}}", from=3-1, to=1-1]
			\arrow["F"', from=3-1, to=1-3]
			\arrow["{\overline{F^*}}", from=1-1, to=1-3]
		\end{tikzcd}\]
		If $G: \mcal{F}\mcal{C}\to \bb{C}$ satisfied $\overline{G}\circ \eta_{\mcal{C}} = F$, then, it would have to map objects in exactly the same way that $F^*$ does. 
		Since $\overline{G}\circ \eta_{\mcal{C}} = F$, for any $f: A\to B$, we have
		\[
		G((1,f)) = \overline{G}((1,f)) = F(f)
		\]
		Consider the projections, $(1,id_{A_j}): A_1,\ldots, A_n\to A_j$ in $\mcal{F}\mcal{C}$. The clone morphism $G$ has to preserve these projections, so:
		\[
		G((j, id_{A_j})) = pr^j_{FA_\bullet}
		\]
		Finally, using the freeness property discussed above, 
		\begin{align*}
			&G((j,t))\\ 
			&=\{\text{Using Freeness}\}\\
			&	 G((1,t) [pr^j_{A_1,\ldots, A_n}])\\
			& = \{\text{$G$ preserves substitution}\}\\
			& G((1,t))[G((pr^j_{A_1,\ldots, A_n}))] \\
			&= \{\text{Definition of $pr^j$}\}\\
			& G(1,t)[G(j,id_{A_j})]\\
			& = \{\text{$G$ preserves projections}\}\\
			& G(1,t)[pr^j_{FA_\bullet}] \\
			& = \{\text{Using the equality above}\}\\
			& F(t) [pr^j_{FA_\bullet}] = F^*(j,t)
		\end{align*}
		which proves that $F^*$ is unique.
	\end{proof}

	The construction is essentially saying this: a programming language that allows for the construction of programs with exactly one free variable can be thought of as a programming language that allows for the construction of programs with $n$ free variables by considering $n-1$ of them as dummies.

	Since $\overline{\mcal{F}\mcal{C}} \cong \mcal{C}$, the free functor $\mcal{F}$ is fully faithful and allows us to think of \category{Cat} as sitting inside \category{Clone}.

	\subsection{Clones and Cartesian Categories}
	
	We now deal with the other extreme: cartesian categories. As observed above, these are structures that can interpret languages with product types, and can also interpret  programs with finitely many free variables. 
	This suggests that there is a `forgetful' functor to $\category{Clone}$. We call the category of categories with chosen finite products, and strict product preserving functors (preserve products and terminal objects strictly) \category{CartCat}.

	\begin{example}
		Let $\mcal{C}$ be a cartesian category. Then, define its underlying clone, $\Cl{\mcal{C}}$ as:
		\begin{itemize}
			\item Objects: $\ob{\Cl{\mcal{C}}} = \ob{\mcal{C}}$,
			\item Multimaps: 
			$\Mhom{\Cl{\mcal{C}}}{A_1,\ldots, A_n}{B} : = \Hom{\mcal{C}}{A_1\times \ldots \times A_n}{B}
			$,
			\item Projections: The projection $pr^j_{A_\bullet}: A_1,\ldots, A_n\to A_j$ is just the projection $\pi^j: A_1\times \ldots\times A_n\to A_j$ that comes with the product structure,
			\item Substitution: This is defined using the cartesian structure: if $t: B_1\times\ldots \times B_m\to C$ and $\{u_i: A_1\times\ldots\times A_n\to B_i\}_{1\leq i\leq m}$ are given, define 
			\[
			t[u_1,\ldots, u_m] := t\circ \langle u_1,\ldots, u_m\rangle
			\]
			where $\langle u_1,\ldots, u_m\rangle: A_1\times\ldots\times A_n\to B_1\times\ldots\times B_m$ is the unique map in $\mcal{C}$ that satisfies $\pi^j\circ \langle u_1,\ldots, u_m\rangle = u_j$ for all $j\in \{1,\ldots, m\}$. 
		\end{itemize}
		It is clear that this defines a clone, just by virtue of how  projections and pairing work in cartesian categories. If $F: \mcal{C}\to \mcal{D}$ is a strict cartesian functor, define $\Cl{F}: \Cl{\mcal{C}}\to \Cl{\mcal{D}}$ to act on objects as in the same way as $F$, i.e,  $\Cl{F}(A) = F(A)$, and on hom-sets as
		\begin{align*}
			\Mhom{\Cl{\mcal{C}}}{A_1,\ldots, A_n}{B} &\to \Mhom{\Cl{\mcal{D}}}{FA_1,\ldots, FA_n}{FB}\\
			t &\mapsto Ft
		\end{align*}
		This type checks since if $t\in \Hom{\mcal{C}}{A_1\times\ldots\times A_n}{B}$, then, $Ft$ is a map in $\Hom{\mcal{D}}{F(A_1\times\ldots\times A_n)}{FB} = \Hom{\mcal{D}}{FA_1\times\ldots\times FA_n}{FB}$.

	\end{example}

	\begin{lemma}\label{free-cart}
		The functor, $\text{Cl}: \category{CartCat}\to \category{Clone}$ has a left adjoint. 
	\end{lemma}
	\begin{proof}
		Let $\bb{C}$ be a clone, then define $\Cart{\bb{C}}$ as:
		\begin{enumerate}
			\item Objects: The free monoid on $\ob{\bb{C}}$, i.e, lists in $\ob{\bb{C}}$
			\item Hom-sets: 
			
			\begin{align*}
				\Hom{\Cart{\bb{C}}}{[A_1,\ldots, A_n]}{[B_1,\ldots, B_m]}: &= \prod_{i=1}^{m} \Mhom{\bb{C}}{A_1,\ldots,A_n}{B_i}\\
				& = \{(f_1,\ldots, f_m) \mid f_i \in \Mhom{\bb{C}}{A_1,\ldots, A_n}{B_i} \}
			\end{align*}
			\item Identities: For a list $[A_1,\ldots, A_n]$, define its identity, $id_{[A_1,\ldots, A_n]}$ to be the $n$-tuple, $(pr_{A_\bullet}^1,\ldots, pr_{A_\bullet}^n)$. 
			
			\item Composition: Given $g=(g_1,\ldots, g_l): [B_1,\ldots, B_m] \to [C_1,\ldots, C_l]$ and $f=(f_1,\ldots, f_m): [A_1,\ldots, A_n]\to [B_1,\ldots, B_m]$, define their composite as:
			\[
			(g_1,\ldots, g_l)\circ (f_1,\ldots, f_m) = (g_1[f_1,\ldots, f_m], \ldots, g_l[f_1,\ldots, f_m])
			\]
			where $g_i[f_1,\ldots, f_m]: A_1,\ldots, A_n\to C_i$ is obtained by substituting the $f_j$ in $g_i$. 
		\end{enumerate}
		It is easily verified that this defines a category. The unitality and associativity axioms of a clone directly imply the unitality and associativity axioms of this category. 
		
		Further, $\Cart{\bb{C}}$ has the empty list $[\;\;]$ as its terminal object: For every object $[A_1,\ldots, A_n]$, 
		\[
		\Hom{\Cart{\bb{C}}}{[A_1,\ldots, A_n]}{[\;\;]} = \prod_{i=1}^{0}  = \{*\}
		\]
		Hence, there is a unique morphism, $[A_1,\ldots, A_n] \to [\;\;]$ for all $[A_1,\ldots, A_n]$. 	For the product structure, define $[A_1,\ldots, A_n]\times [A_{n+1},\ldots, A_m] := [A_1,\ldots, A_n, A_{n+1},\ldots, A_m]$. Then, the projections are given by:
		\[
		\pi = (pr^1_{A_\bullet},\ldots, pr^n_{A_\bullet}), \qquad\qquad \pi' = (pr^{n+1}_{A_\bullet},\ldots, pr^m_{A_\bullet})
		\]
		If $(f_1,\ldots, f_n): [B_1,\ldots, B_l]\to [A_1,\ldots, A_n]$ and $(f_{n+1},\ldots, f_m): [B_1,\ldots, B_l]\to [A_{n+1},\ldots, A_m]$ are morphisms, then, define 
		\[
		\langle (f_1,\ldots, f_n), (f_{n+1},\ldots, f_m)\rangle := (f_1,\ldots, f_n, f_{n+1},\ldots, f_m)
		\]
		It is verified using the unitality laws in $\bb{C}$ that $\langle (f_1,\ldots, f_n), (f_{n+1},\ldots, f_m)\rangle$ is the unique map that satisfies the universal property of a product. Thus, $\Cart{\bb{C}}$ is a cartesian category.

		Define the clone morphism, $\eta_{\bb{C}}: \bb{C}\to \Cl{\Cart{\bb{C}}}$ as:
		\begin{enumerate}
			\item On objects: $A\mapsto [A]$
			\item On homs: 
			\begin{align*}
				\eta_{\bb{C}_{A_1,\ldots, A_n;B}}:	\Hom{\bb{C}}{A_1,\ldots, A_n}{B} \to& \Hom{\Cl{\Cart{\bb{C}}}}{[A_1],\ldots, [A_n]}{[B]}\\
				& = \Hom{\Cart{\bb{C}}}{[A_1,\ldots, A_n]}{[B]} \\
				& = \Hom{\bb{C}}{A_1,\ldots, A_n}{B}
			\end{align*}
			is the identity: 
		\end{enumerate}
		That this preserves projections and identities directly follows from definitions and the axioms of a clone. For universality, assume that there is some other cartesian category, $\mcal{C}$ with a clone morphism, 
		\[
		F: \bb{C}\to \Cl{\mcal{C}} 
		\]
		Then, define the strict cartesian functor, $F^*: \Cart{\bb{C}}\to \mcal{C}$ as:
		\begin{enumerate}
			\item On objects: $F^*([A_1,\ldots, A_n]) = F(A_1)\times\ldots \times F(A_n)$
			\item On homs: 
			\begin{align*}
				\Hom{\Cart{\bb{C}}}{[A_1,\ldots, A_n]}{[B_1,\ldots, B_m]} &\to  \Hom{\mcal{C}}{F(A_1)\times \ldots \times F(A_n)}{F(B_1)\times\ldots\times F(B_m)}\\
				(f_1,\ldots, f_m)&\mapsto  \langle F(f_1),\ldots, F(f_m)\rangle
			\end{align*}
		\end{enumerate}
		Then the following diagram commutes:
		\[\begin{tikzcd}[cramped]
			{\Cl{\Cart{\bb{C}}}} && {\Cl{\mcal{C}}} \\
			\\
			{\bb{C}}
			\arrow["{\eta_{\bb{C}}}", from=3-1, to=1-1]
			\arrow["{\Cl{F^*}}", from=1-1, to=1-3]
			\arrow["F"', from=3-1, to=1-3]
		\end{tikzcd}\]
		
		Note that $F^*$ is unique in this regard: If $G$ was another such map, it would act on objects just like $F^*$ does in order to strictly preserve products. If $\Cl{G}\circ \eta_{{\bb{C}}} = F$, then,for any $f_i: A_1,\ldots, A_n\to B_i$ in $\bb{C}$, 
		\[
		F(f_i) = \Cl{G}( f_i)  
		\]
		Since $\langle f_1,\ldots, f_m\rangle $ in $\bb{C}$ is defined as $(f_1,\ldots, f_m)$ and since $G$ is strict cartesian, 
		\begin{align*}
			G((f_1,\ldots, f_m)) & = G(\langle f_1,\ldots, f_m\rangle) \\
			& = \langle Gf_1,\ldots, Gf_m\rangle \\
			& = \langle F(f_1),\ldots, F(f_m)\rangle  = F^*((f_1,\ldots, f_n))
		\end{align*}
	\end{proof}

	This construction explicitly adds product types to the language of clones. Morphisms into a product are already determined by the product structure since $\Hom{\mcal{C}}{X}{A\times B} \cong \Hom{\mcal{C}}{X}{A}\times \Hom{\mcal{C}}{X}{B}$. This construction is free in the sense that there are no more  morphisms coming out of the product than required.

	In summary, we obtain adjunctions: 
	\begin{equation}\label{adjunctions}
		\begin{tikzcd}
			&& {\category{CartCat}} \\
			\\
			{\category{Clone}} \\
			\\
			&& {\category{Cat}}
			\arrow[""{name=0, anchor=center, inner sep=0}, "{\overline{(-)}}"{description}, curve={height=18pt}, from=3-1, to=5-3]
			\arrow[""{name=1, anchor=center, inner sep=0}, "{\mcal{F}}"{description}, curve={height=18pt}, from=5-3, to=3-1]
			\arrow[""{name=2, anchor=center, inner sep=0}, "{\text{Cart}}"{description}, curve={height=18pt}, from=3-1, to=1-3]
			\arrow[""{name=3, anchor=center, inner sep=0}, "{\text{Cl}}"{description}, curve={height=24pt}, from=1-3, to=3-1]
			\arrow[""{name=4, anchor=center, inner sep=0}, "U"{description}, from=1-3, to=5-3]
			\arrow[""{name=5, anchor=center, inner sep=0}, "F"{description}, curve={height=30pt}, from=5-3, to=1-3]
			\arrow["\dashv"{anchor=center, rotate=-146}, draw=none, from=1, to=0]
			\arrow["\dashv"{anchor=center, rotate=-180}, draw=none, from=5, to=4]
			\arrow["\dashv"{anchor=center, rotate=147}, draw=none, from=2, to=3]
		\end{tikzcd}
	\end{equation}
	where $U = \overline{(-)}\circ \text{Cl}$ and $F =\text{Cart}\circ \mcal{F} $. Let $\mcal{C}$ be a cartesian category. Then, the category $U\mcal{C}$ has: 
	\begin{enumerate}
		\item Objects: $\ob{U\mcal{C}} = \ob{\overline{\Cl{\mcal{C}}}} = \ob{\Cl{\mcal{C}}} = \ob{\mcal{C}}$.
		\item Homsets: $\Hom{U\mcal{C}}{A}{B} = \Hom{\overline{\Cl{\mcal{C}}}}{A}{B} = \Mhom{\Cl{\mcal{C}}}{A}{B} = \Hom{\mcal{C}}{A}{B} $ 
	\end{enumerate}
	Therefore, $U$ is the forgetful functor $\category{CartCat}\to  \category{Cat}$, which means that $F$ is the free cartesian category functor. Intuitively, we can think of the category of clones as a pitstop on the road between \category{Cat} and \category{CartCat}: the functor $\mcal{F}$ adds half-cartesian structure to a category $\mcal{C}$, and the functor $\text{Cart}$ adds the other half. These constructions also tell us what the free cartesian category should look like. 
	
	\begin{corollary}
		The free cartesian category $F\mcal{C}$ on a category $\mcal{C}$ can be constructed as:
		\begin{itemize}
			\item Objects: $\ob{F\mcal{C}}$ is the free monoid on $\ob{\mcal{C}}$.
			\item Hom-sets:
			\begin{align*}
				&\Hom{F\mcal{C}}{[A_1,\ldots, A_n]}{[B_1,\ldots, B_m]} \\& = \prod_{i=1}^{m} \sum_{j=1}^{n} \Hom{\mcal{C}}{A_j}{B_i}\\
				& = \{((j_1,f_1),\ldots, (j_m, f_m) ) \mid j_i \in \{1,\ldots,n\}, f\in \Hom{\mcal{C}}{A_{j_i}}{B_i}   )\}
			\end{align*}
			So a morphism $f: [A_1,\ldots, A_n]\to [B_1,\ldots, B_m]$ is an $m$-tuple of morphisms, $A_{j_i}\to B_i$ for $1\leq i \leq m$.
			
			\item Identities: The morphism $id_{[A_1,\ldots, A_n]}$ is the tuple $((1, id_{A_1}),\ldots, (n, id_{A_n}))$ 
			\item Composition: If $((j_1,f_1), \ldots, (j_m,f_m)): [A_1,\ldots, A_n]\to [B_1,\ldots, B_m]$ and $((k_1,g_1),\ldots, (k_l,g_l)): [B_1,\ldots, B_m]\to [C_1,\ldots, C_l]$, then, the composite is given by:
			\begin{align*}
				&((k_1,g_1),\ldots, (k_l,g_l))\circ ((j_1,f_1), \ldots, (j_m,f_m))\\& = ((k_1,g_1)[(j_1,f_1),\ldots, (j_m,f_m)] ,\ldots, (k_l,g_l)[(j_1,f_1),\ldots, (j_m,f_m)]   )\\
				& = ((j_{k_1}, g_1\circ f_{k_1}),\ldots, (j_{k_l}, g_l\circ f_{k_l}))
			\end{align*}
		\end{itemize}
	\end{corollary}

	\subsection{Algebras for Clones}
	
	Abstract clones arose as a way to define an algebraic theory in universal algebra where they are typically considered with one object. In this case, they have a simpler definition: 
	\begin{definition}\cite{cohn}
		A (one-sorted) abstract clone is given by a set, $\mcal{T}(n)$ for every natural number $n$ such that there $n$ distinguished projection operations, $pr^1_n,\ldots, pr^n_n \in \mcal{T}(n)$, along with a substitution operation,
		\[
		\mcal{T}(n) \times \mcal{T}(m)^n \to \mcal{T}(m); \qquad (f,g_1,\ldots, g_n)\mapsto f[g_1,\ldots, g_n]
		\]
		such that 
		\begin{enumerate}
			\item $pr^k_n[g_1,\ldots, g_n] = g_k$
			\item $f[pr^1_n,\ldots, pr^n_n] = t$
			\item $f[g_1,\ldots, g_n][h_1,\ldots, h_m] = f[g_1[h_1,\ldots, h_m],\ldots, g_n[h_1,\ldots, h_m]]$
		\end{enumerate}
	\end{definition}

	Thus, a one-sorted abstract clone is a set of $n$-ary operations, such that projections are among the $n$-ary operations, and such that operations can be substituted. 	
	Each one-sorted clone determines an algebraic structure. For example, we could consider the abstract clone of monoids, which has an identity, $id \in T(1)$ and a binary operation $mult \in T(2)$, subject to certain equations. Similarly, one could consider groups, rings, and so on. An algebra for a clone allows us to consider actual models of such structures. An algebra for the clone of monoids is an actual monoid, an algebra for the clone of groups is an actual group, and so on. 
	
	\begin{definition}\cite{gould}
		Let $\mcal{T}$ be a (one-sorted) clone.	A $\mcal{T}$-algebra, $(A, \interpret{-}_A)$ is given by:
		\begin{enumerate}
			\item A set $A$,
			\item For every $n$, and every $f \in \mcal{T}(n)$, an interpretation, $\interpret{f}_A: A^n\to A$. We just write $\interpret{f}$ when the set can be gathered form context. 
		\end{enumerate}
		such that 
		\begin{enumerate}
			\item The projection is interpreted as a projection: For all $n,k$, $$\interpret{pr^k_n} (a_1,\ldots, a_n) = a_k$$
			\item Substitution is \textit{actual} substitution: For all $f\in \mcal{T}(n)$ and $g_1,\ldots, g_n \in \mcal{T}(m)$,
			\[
			\interpret{f[g_1,\ldots, g_n]}(a_1,\ldots, a_m) = \interpret{f}(\interpret{g_1}(a_1,\ldots, a_m),\ldots, \interpret{g_n}(a_1,\ldots, a_m))
			\]
		\end{enumerate}
	\end{definition}
	
	Thus, $\mcal{T}$-algebras for a clone define algebraic objects that are `defined' by $\mcal{T}$. The notion of a morphism of $\mcal{T}$-algebras defines the appropriate notion of morphism between these algebraic objects. When $\mcal{T}$ is the clone of monoids, a morphism of $\mcal{T}$-algebras is a monoid homomorphism, when $\mcal{T}$ is the clone of groups, it is a group homomorphism, and so on. 
	
	\begin{definition}
		Let $(A, \interpret{-}_A),(B, \interpret{-}_B)$ be $\mcal{T}$-algebras. A morphism, $f:(A,\interpret{-}_A)\to (B,\interpret{-}_B)$ is a function, $f: A\to B$ such that 
		\[
		\interpret{k}_A(a_1,\ldots, a_n) = \interpret{k}_B(f(a_1),\ldots, f(a_n))
		\]
		for all $k \in \mcal{T}(n)$ and for all $n \in \mathbb{N}$.
	\end{definition}

	Hence, for every one sorted clone, $\mcal{T}$, there is a category, $\mcal{T}$-\category{Alg} of $\mcal{T}$-algebras and their morphisms. There is clearly a forgetful functor, 
	\[
	U: \mcal{T}\text{-}\category{Alg} \to \category{Set}
	\] 	
	which sends a $\mcal{T}$-algebra $(A,\interpret{-}_A)$ to its underlying set $A$, and a $\mcal{T}$-algebra morphism, $f$ to its underlying function, and this functor is monadic \cite{nlab-clone}.  As a result, the category of $\mcal{T}$-algebras is well-behaved, i.e, it has nice limits.

	We will extend the notion of an algebra to multi-sorted clones, i.e, the clones that were dealt with in the previous section. First, however, we turn to a more general notion of an  algebraic theory and algebras for it, which we will use as a sanity check. 
	
	\begin{definition}\cite{vitale}\label{t-alg}
		\begin{enumerate}
			\item An algebraic theory is a small category $\mcal{T}$ with finite products, i.e, an object of \category{CartCat}. 
			
			\item A  $\mcal{T}$-algebra for an algebraic theory $\mcal{T}$ is a functor, 	
			$	A: \mcal{T} \to \category{Set}
			$ that preserves (not necessarily strictly) finite products. 
			
			\item A morphism of $\mcal{T}$-algebras, $f: A\to B$ is a natural transformation between the functors. 
			
			\item The category of $\mcal{T}$-algebras, $\mcal{T}$-\category{Alg} for an algebraic theory is the category of product preserving functors, $A: \mcal{T}\to \category{Set}$ and natural transformations.
		\end{enumerate}

		Note that the definition uses the cartesian structure of the category $\category{Set}$. Since it has a cartesian structure, \category{Set} can also be made into a clone by considering $\Cl{\category{Set}}$. Explicitly, this clone has:
		\begin{enumerate}
			\item Objects: Sets
			\item Hom-sets: $\Mhom{\Cl{\category{Set}}}{A_1,\ldots, A_n}{B} = \Hom{\category{Set}}{A_1\times\ldots\times A_n}{B}$. 
			\item Projections: $A_1\times\ldots\times A_n\to A_j$ is the function $(a_1,\ldots, a_n)\mapsto a_j$.
			\item Substitution: Given functions $t:B_1\times\ldots\times B_m \to C$ and $u_i: A_1,\ldots, A_n\to B_i$ for $1\leq i\leq m$, define $t[u_1,\ldots, u_n]$ to be the function $(a_1,\ldots, a_n) \mapsto t(u_1(a_1,\ldots, a_n),\ldots, u_n(a_1,\ldots, a_n))$.
		\end{enumerate}
		
		Thus, we can imitate the definition of an algebra for an algebraic theory with the following definition. 
		\begin{definition}
			A $\bb{C}$-algebra for a clone $\bb{C}$ is a clone morphism, $\bb{C}\to \Cl{\category{Set}}$.
		\end{definition}
		To analogously define the notion of a $\bb{C}$-algebra morphism, we require the notion of a natural transformation between clone morphisms. Fortunately, the first definition that the reader might have considered suits our purposes. 
		\begin{definition}\cite{hermida}
			Let $f,g: \bb{C}\to \bb{D}$ be  clone morphisms. A  transformation, $\eta: f\Rightarrow g$  between them is a family of maps, $\{\eta_A \in \Mhom{\bb{D}}{fA}{gA}\}_{A\in \ob{\bb{C}}}$ such that for any $t: A_1,\ldots, A_n\to B$, the following are equal
			\[
			\eta_B[f(t)] = g(t)\Big[\eta_{A_1}[pr^1_{fA_\bullet}],\ldots, \eta_{A_n}[pr^n_{fA_\bullet}]\Big]
			\]
			Note that this type checks since $pr^j_{fA_\bullet}: fA_1,\ldots, fA_n\to fA_j$ and $\eta_{A_j}: fA_j\to gA_j$, and $g(t): g_1,\ldots, gA_n\to gB$. 
		\end{definition}
		
		\begin{remark}\label{transformation}
			There is a functor $M:\category{Clone}\to \category{MultiCat}$ where \category{MutliCat} is the category of multicategories (which we define in the next section), and a transformation $\eta: f\Rightarrow g$ between clones as defined above is equivalently a transformation between the multicategory morphisms  $ Mf$ and $ Mg$. The notion of a transformation between multicategory morphisms is well-known, and can be found in \textit{e.g.} \cite{psaville}, \cite{gould}, \cite{hermida}. 
		\end{remark}

		As a justification that this is a sensible definition of a 2-morphism, we check that the following data determines a category for fixed clones $\bb{C}$ and $\bb{D}$:
		\begin{enumerate}
			\item Objects: Clone morphisms, $f: \bb{C}\to \bb{D}$
			\item Morphisms: Transformations, $\eta: f\Rightarrow g$
		\end{enumerate}
		Note that there is always an identity transformation, $id_f: f\Rightarrow f$ for a clone morphism $f: \bb{C}\to \bb{D}$ defined by: ${id_{f_A}} :=  pr^1_{fA}$ for $A\in \bb{C}$. 
		
		Further, transformations $\eta: f\Rightarrow g$ and $\epsilon: g\Rightarrow h$ can be composed just like natural transformations of functors. Define $(\epsilon\ast \eta)_{A} := \epsilon_A [ \eta_A] $. Then, if $t: A_1,\ldots, A_n\to B$ is a multimap, we have the following equalities: 
		\[
		\eta_B[f(t)] = g(t)\Big[\eta_{A_1}[pr^1_{fA_\bullet}],\ldots, \eta_{A_n}[pr^n_{fA_\bullet}]\Big]
		\]
		and
		\[
		\epsilon_B[g(t)]  = h(t)\Big[\epsilon_{A_1}[pr^1_{gA_\bullet}],\ldots, \epsilon_{A_n}[pr^n_{gA_\bullet}]\Big]
		\]
		Using these, associativity and unitality,  we infer:
		\begin{align*}
			&(\epsilon\ast\eta)_B [f(t)] \\
			& = \epsilon_B[ \eta_B] [f(t)]\\
			& = \epsilon_B[\eta_B[f(t)]]\\
			& = \epsilon_B\Big[g(t)\big[\eta_{A_1}[pr^1_{fA_\bullet}],\ldots, \eta_{A_n}[pr^n_{fA_\bullet}]\big]\Big]\\
			& = \epsilon_B\big[g(t)\big] \Big[\eta_{A_1}[pr^1_{fA_\bullet}],\ldots, \eta_{A_n}[pr^n_{fA_\bullet}]\Big]\\
			& = h(t)\Big[\epsilon_{A_1}[pr^1_{gA_\bullet}],\ldots, \epsilon_{A_n}[pr^n_{gA_\bullet}] \Big]\Big[\eta_{A_1}[pr^1_{fA_\bullet}],\ldots, \eta_{A_n}[pr^n_{fA_\bullet}]\Big]\\
			& = h(t) \Big[  \epsilon_{A_1}[pr^1_{gA_\bullet}]\Big[\eta_{A_1}[pr^1_{fA_\bullet}],\ldots, \eta_{A_n}[pr^n_{fA_\bullet}]\Big],\ldots, \epsilon_{A_n}[pr^n_{gA_\bullet}]\Big[\eta_{A_1}[pr^1_{fA_\bullet}],\ldots, \eta_{A_n}[pr^n_{fA_\bullet}]\Big]                     \Big]\\
			& = h(t)\Big[ \epsilon_{A_1}\big[\eta_{A_1} [pr^1_{fA_\bullet}] \big],\ldots, \epsilon_{A_n}\big[\eta_{A_n}[pr^n_{fA_\bullet}]\big]\Big]\\
			&=h(t) \Big[ \epsilon_{A_1}[\eta_{A_1}][pr^1_{fA_\bullet}],\ldots, \epsilon_{A_n}[\eta_{A_n}][pr^n_{fA_\bullet}]  \Big]\\
			& = h(t)\Big[\epsilon\ast\eta_{A_1} [pr^1_{fA_\bullet}],\ldots, \epsilon\ast\eta_{A_n}[pr^n_{fA_\bullet}] \Big]
		\end{align*}
		which shows that $\epsilon\ast \eta$ is also a transformation. Associativity and unitality follow from componentwise associativity and unitality. Thus, for any two clones, $\bb{C}$ and $\bb{D}$, there is a category, $\Hom{\category{Clone}}{\bb{C}}{\bb{D}}$ of clone morphisms and their transformations. If $\mcal{C}$ and $\mcal{D}$ are cartesian categories, let $\Hom{\category{ProdPres}}{\mcal{C}}{\mcal{D}}$ denote the category of (not necessarily strict) product preserving functors and natural transformations.

		\begin{lemma}
			Let $\bb{C}$ be a clone and let $\mcal{D}$ be a cartesian category. There is an equivalence of categories,
			\[
			\Hom{\category{ProdPres}}{\Cart{\bb{C}}}{\mcal{D}} \simeq \Hom{\category{Clone}}{\bb{C}}{\Cl{\mcal{D}}}
			\]
		\end{lemma}
		\begin{proof}
			The constructions here are a little more general version of the natural isomorphism of \textit{sets},
			\[
			\Hom{\category{CartCat}}{\Cart{\bb{D}}}{\mcal{D}} \cong \Hom{\category{Clone}}{\bb{C}}{\Cl{\mcal{D}}}
			\] induced by the adjunction $\text{Cart}\dashv \text{Cl}$  in \ref{adjunctions}. 
			
			Define the functor, $F: \Hom{\category{Clone}}{\bb{C}}{\Cl{\mcal{D}}}\to \Hom{\category{ProdPres}}{\Cart{\bb{C}}}{\mcal{D}}$ as:
			\begin{enumerate}
				\item On objects: A clone morphism $f: \bb{C} \to \Cl{\mcal{D}}$ gets mapped to a product preserving functor, $Ff: \Cart{\bb{C}} \to \mcal{D}$ which acts as:
				\begin{enumerate}
					\item On objects: $[A_1,\ldots, A_n] \mapsto fA_1\times\ldots\times fA_n$
					\item On hom-sets: 
					\begin{align*}
						\Hom{\Cart{\bb{C}}}{[A_1,\ldots, A_n]}{[B_1,\ldots, B_m]} &\to \Hom{\mcal{D}}{fA_1\times\ldots\times fA_n}{fB_1\times\ldots\times fB_m}\\
						(t_1,\ldots, t_m) &\mapsto \langle f(t_1),\ldots, f(t_m)\rangle 
					\end{align*}
				\end{enumerate}
				Note that the map $Ff$ may \textit{not} be strictly product preserving since the products in $\Cart{\bb{C}}$ are strictly associative, while those in $\mcal{D}$ might not be.
				
				\item On morphisms: A transformation, $\eta: f\Rightarrow g$ gets mapped to a natural transformation, $F\eta: Ff\Rightarrow Fg$, defined to have components:
				\[
				(F\eta)_{[A_1,\ldots, A_n]} : = \eta_{A_1}\times\ldots\times\eta_{A_n}: (Ff)([A_1,\ldots, A_n]) \to (Fg)([A_1,\ldots, A_n])
				\]
				This type checks since $\eta_{A}: fA \to gA$ is an arrow in $\Cl{\mcal{D}}$. For naturality, assume that $(t_1,\ldots, t_m): [A_1,\ldots, A_n]\to [B_1,\ldots, B_m]$ is a morphism in $\Cart{\bb{C}}$. Then, the following diagram commutes:
				\[\begin{tikzcd}[cramped]
					{[A_1,\ldots, A_n]} && {Ff([A_1,\ldots, A_n])} && {Fg([A_1,\ldots, A_n])} \\
					\\
					{[B_1,\ldots, B_m]} && {Ff([B_1,\ldots, B_m])} && {Fg([B_1,\ldots, B_m])}
					\arrow["{(t_1,\ldots, t_m)}", from=1-1, to=3-1]
					\arrow["{\langle ft_1,\ldots, ft_m\rangle}", from=1-3, to=3-3]
					\arrow["{\eta_{A_1}\times\ldots\times \eta_{A_n}}", from=1-3, to=1-5]
					\arrow["{\eta_{B_1}\times\ldots\times \eta_{B_m}}", from=3-3, to=3-5]
					\arrow["{\langle gt_1,\ldots, gt_m\rangle}", from=1-5, to=3-5]
				\end{tikzcd}\]
				Since
				\begin{align*}
					&	\langle gt_1,\ldots, gt_m\rangle (\eta_{A_1}\times\ldots, \eta_{A_n}) \\
					& = \langle gt_1\circ (\eta_{A_1}\times\ldots\times \eta_{A_n}),\ldots, gt_m\circ (\eta_{A_1},\ldots, \eta_{A_n})\rangle \\
					& =\langle gt_1[\eta_{A_1}[\pi^1],\ldots, \eta_{A_n}[\pi^n] ],\ldots, gt_m[{\eta_{A_1}}[\pi^1],\ldots, \eta_{A_n}[\pi^n]]\rangle \\
					& = \langle  \eta_{B_1}[ft_1],\ldots, \eta_{B_m}[ft_m]\rangle \\
					& = \langle \eta_{B_1} [\pi^1][  ft_1,\ldots, ft_m] ,\ldots, \eta_{B_m}[\pi^n][ft_1,\ldots, ft_m]\rangle \\
					& = \langle \eta_{B_1}\circ \pi^1,\ldots, \eta_{B_n}\circ \pi^n\rangle \circ \langle ft_1,\ldots, ft_m\rangle\\
					& = (\eta_{B_1}\times\ldots\times\eta_{B_n})\circ \langle ft_1,\ldots, ft_m\rangle
				\end{align*}
				using definitions.		
			\end{enumerate}
			The mapping $F$ can easily be seen to be functorial. It is  essentially surjective on objects: if $g: \Cart{\bb{C}} \to \mcal{D}$ is a product preserving functor, then define a clone morphism $f: \bb{C}\to \Cl{\mcal{D}}$ as:
			\begin{enumerate}
				\item On objects: $A\mapsto g([A])$
				\item On hom-sets: 
				\begin{align*}
					\Mhom{\bb{C}}{A_1,\ldots, A_n}{B} &\to \Mhom{\Cl{\mcal{D}}}{g([A_1]),\ldots, g([A_n]) }{g([B])}\\
					t &\mapsto  g((t))\circ \zeta_{[A_1,\ldots, A_n]}
				\end{align*}
				If $t: A_1,\ldots, A_n\to B$ is a morphism in $\bb{C}$, then, $(t):[A_1,\ldots, A_n]\to [B]$ is a morphism in $\Cart{\bb{C}}$, which then gets mapped to $g((t)): g([A_1,\ldots, A_n])\to gB$, so everything type checks. The morphism $\zeta_{[A_1,\ldots, A_n]}$ is the isomorphism $g[A_1]\times\ldots g[A_n] \cong g([A_1,\ldots, A_n])$ which exists since $g$ preserves products.
				
			\end{enumerate}
			Then, $Ff:\Cart{\bb{C}}\to \mcal{D}$ is isomorphic to $g$: The natural isomorphism is given by $\zeta_{[A_1,\ldots, A_n]}: Ff([A_1,\ldots, A_n]) = g[A_1]\times\ldots\times g[ A_n] \to g([A_1,\ldots, A_n]) $.

			We now show that the functor is full and faithful, i.e, for $f,g: \bb{C}\to \Cl{\mcal{D}}$, there is an isomorphism
			\[
			F:\Hom{\category{Clone}}{\bb{C}}{\Cl{\mcal{D}}}(f,g) \cong \Hom{\category{ProdPres}}{\Cart{\bb{C}}}{\mcal{D}}(Ff, Fg)
			\]
			between the set of transformations, $\eta: f\Rightarrow g$ and the set of natural transformations, $\epsilon: Ff\Rightarrow Fg$. We invert the mapping. If $\eta: \Rightarrow Fg$ is a natural transformation, it has components $\eta_{[A_1,\ldots, A_n]}$ for an object $[A_1,\ldots, A_n] \in \ob{\Cart{\bb{C}}}$. Define $G\eta: f\Rightarrow g$ to have components $(G\eta)_{A}:= \eta_{[A]}: (Ff)([A]) = fA\to (Fg)([A]) = gA $. To check that this is indeed a transformation, we prove the following sub-lemma. 
			
			\begin{sublemma}
				For a natural transformation $\eta: Ff\Rightarrow Fg$, we have an equality,
				\[
				\eta_{[A_1,\ldots, A_n]} = \eta_{A_1}\times\ldots\times \eta_{A_n}
				\]
			\end{sublemma}
			\begin{proof}
				Note that $(Ff)([A_1,\ldots, A_n]) = fA_1\times\ldots fA_n$. There are projections, $[A_1,\ldots, A_n]\to [A_j]$ given by $pr^j_{A_\bullet}$. So, by naturality, and projection preservation, the following diagram commutes:
				\[\begin{tikzcd}[cramped]
					{Ff([A_1,\ldots, A_n])} && {Fg([A_1,\ldots, A_n])} \\
					\\
					{Ff([A_j])} && {Fg([A_j])}
					\arrow["{\pi^j_{fA_1,\ldots, fA_n}}"', from=1-1, to=3-1]
					\arrow["{\eta_{[A_1,\ldots, A_n]}}", from=1-1, to=1-3]
					\arrow["{\pi^j_{gA_1,\ldots, gA_n}}", from=1-3, to=3-3]
					\arrow["{\eta_{[A_j]}}"', from=3-1, to=3-3]
				\end{tikzcd}\]
				Then consider the following diagram:
				\[\begin{tikzcd}[cramped]
					&& {fA_1\times\ldots \times fA_n} \\
					\\
					{fA_1} && \dots && {fA_n} \\
					{gA_1} && \dots && {gA_n} \\
					\\
					&& {gA_1\times\ldots \times gA_n}
					\arrow["{\pi^1_{fA_1,\ldots, fA_n}}"', from=1-3, to=3-1]
					\arrow["{\pi^n_{fA_1,\ldots, fA_n}}", from=1-3, to=3-5]
					\arrow["{\eta_{[A_1]}}", from=3-1, to=4-1]
					\arrow["{\pi^1_{gA_1,\ldots, gA_n}}", from=6-3, to=4-1]
					\arrow["{\pi^n_{gA_1,\ldots, gA_n}}"', from=6-3, to=4-5]
					\arrow["{\eta_{[A_1,\ldots, A_n]}}"{description}, curve={height=-30pt}, from=1-3, to=6-3]
					\arrow["{\eta_{[A_1]}\times\ldots \eta_{[A_n]}}"{description}, curve={height=30pt}, from=1-3, to=6-3]
					\arrow["{\eta_{[A_n]}}"', from=3-5, to=4-5]
				\end{tikzcd}\]
				The morphism $\eta_{[A_1,\ldots, A_n]}$ satisfies $\pi^j_{gA_1,\ldots, gA_n}\circ\eta_{[A_1,\ldots, A_n]} = \eta_{[A_j]}\circ \pi^j_{fA_1,\ldots, fA_n}$ for all $j$ since the previous diagram always commutes. Thus, $\eta_{[A_1,\ldots, A_n]}$ and $\eta_{[A_1]}\times\ldots \times \eta_{[A_n]}$ satisfy the same universal property, and hence must be equal.
				
			\end{proof}
			
			Now let $t:A_1,\ldots, A_n\to B$ be a multimap in $\bb{C}$. This becomes a morphism $(t): [A_1,\ldots, A_n]\to [B]$. Then by naturality of $\eta$, the following diagram commutes: 
			
			\[\begin{tikzcd}[cramped]
				{\substack{Ff([A_1,\ldots,A_n]) \\\quad\\=f A_1\times\ldots\times fA_n}} &&& {\substack{Fg([A_1,\ldots,A_n]) \\\quad\\=g A_1\times\ldots\times gA_n}} \\
				\\
				{\substack{Ff([B]) \\ \quad \\=fB}} &&& {\substack{Fg([B]) \\ \quad \\=gB}}
				\arrow["{\substack{\eta_{[A_1,\ldots, A_n]}\\ \quad\\= \eta_{[A_1]}\times\ldots\times \eta_{[A_n]}}}", from=1-1, to=1-4]
				\arrow["{\substack{f(t)}}", from=1-1, to=3-1]
				\arrow["{\substack{\eta_{[B]}}}", from=3-1, to=3-4]
				\arrow["{\substack{g(t)}}", from=1-4, to=3-4]
			\end{tikzcd}\]
			Consequently,
			\begin{align*}
				G\eta_{B} [f(t)] & = \eta_{[B]} \circ f(t) \\
				&= g(t)\circ \eta_{[A_1]}\times\ldots \times \eta_{[A_n]}\\
				& = g(t)\circ \langle G\eta_{A_1}[\pi^1_{fA_\bullet}],\ldots, G\eta_{A_n}[\pi^n_{fA_\bullet}]\rangle 
			\end{align*}
			which shows that $G\eta$ is a transformation $f\Rightarrow g$. The mappings $\eta\mapsto F\eta$ and $\eta\mapsto G\eta$ are inverses by virtue of the Sub-Lemma. This completes the proof that the two sets are in bijective correspondence, and that $F$ is fully faithful, and that $F$ is an equivalence of categories.	
		\end{proof}

		Then since \category{Set} is a cartesian category, using Definition \ref{t-alg}, we obtain the following corollary. 
		\begin{corollary}
			For any clone $\bb{C}$, the following categories are equivalent:
			\[
			{\Cart{\bb{C}}}\text{-}\category{Alg} = \Hom{\category{ProdPres}}{\Cart{\bb{C}}}{\category{Set}}\simeq \Hom{\category{Clone}}{\bb{C}}{\Cl{\category{Set}}} 
			\]
		\end{corollary}
		
		This corollary serves as a sanity-check since it is basically saying that  algebras and  morphisms of algebras for a clone are the same as algebras and  morphisms of algebras for the algebraic theory it induces. This is a slightly general version of the fact that the category of algebras for a single sorted clone is equivalent to the category of algebras for the Lawvere theory it induces \cite[Theorem 1.5.4]{gould}.

		We, thus arrive at a definition for the category of algebras  for a multi-sorted clone. 
		\begin{definition}
			Let $\bb{C}$ be a clone. Then,
			\begin{enumerate}
				\item A $\bb{C}$-algebra is a clone morphism $A:\bb{C}\to \Cl{\category{Set}}$.
				\item A morphism $f: A\to B$ of $\bb{C}$-algebras is a transformation, $f: A\Rightarrow B$.
				\item The category of $\bb{C}$-algebras, $\bb{C}$-\category{Alg} is the category of $\bb{C}$-algebras and transformations between them, $\Hom{\category{Clone}}{\bb{C}}{\Cl{\category{Set}}}$.
			\end{enumerate}
		\end{definition}

		Clones serve as an algebraic presentation of the syntax for a programming language. The category of algebras for a clone allow us to study different models of the programming language which informs the study of the programming language itself.	
	\end{definition}

	\section{Multicategories and their Algebras}	
	
	We briefly go through the theory of multicategories, with the perspective that they are a  linear analogue to the theory of clones described in the previous section.

	\subsection{Multicategories are Linear Clones}
	Multicategories provide a generalization of clones. These structures can also interpret languages that can construct programs with finitely many free variables. However, the language is not required to satisfy the renaming rule:
	\begin{prooftree}
		\AxiomC{$\Delta\vdash M:B$}\RightLabel{(There is a renaming $\Delta\to \Gamma$)}
		\UnaryInfC{$\Gamma\vdash M^\Gamma:B$}
	\end{prooftree}
	Hence, multicategories can interpret a linear version of the language for clones \cite{ded-2}. Such a language has a `do nothing' program with only one free variable:
	\begin{prooftree}
		\AxiomC{$\vdash A:\text{Type}$}
		\UnaryInfC{$x:A\vdash x:A$}
	\end{prooftree}
	Substitution in this language is a linear version of substitution in a clone
	\begin{prooftree}
		\AxiomC{$x_1:B_1,\ldots, x_m:B_m\vdash M:C$}
		\AxiomC{$\{\Delta_i \vdash N_i:B_i\}_{1\leq i\leq m}$}
		\BinaryInfC{$\Delta_1,\ldots, \Delta_m\vdash M\langle N_1,\ldots, N_m\rangle:C$}
	\end{prooftree}
	Note that the contexts of the terms that get plugged in also vary, as compared to the way substitution in a clone works.

	\begin{definition}[\text{e.g.} \cite{leinster}]
		A multicategory $\bb{C}$ consists of the following data:
		\begin{enumerate}
			\item A set of objects, $\ob{\bb{C}}$,
			\item For every sequence fo objects $A_1,\ldots, A_n, B$, a set of multimaps, $\Mhom{\bb{C}}{A_1,\ldots, A_n}{B}$,
			\item For every object $A$, an identity multimap, $id_A\in \Mhom{\bb{C}}{A}{A}$,	
			\item A substitution operation:	
			\begin{align*}
				\Mhom{\bb{C}}{B_1,\ldots, B_m}{C}\times \prod_{i=1}^{n} \Mhom{\bb{C}}{A^1_i,\ldots, A^{k_i}_i}{B_i} &\to \Mhom{\bb{C}}{A^1_1,\ldots, A^{k_1}_1,\ldots, A^1_n,\ldots, A^{k_n}_n}{C}\\
				t, u_1,\ldots, u_n &\mapsto t\langle u_1,\ldots, u_n\rangle
			\end{align*}
		\end{enumerate}
		such that the following axioms hold:
		\begin{enumerate}
			\item Left unitality: For $id_{B}: B\to B$ and  $u: A^1_1,\ldots, A^{k_1}_1\to B$, we have
			\[
			id_{B}\langle u\rangle = u
			\] 
			\item Right Unitality: For $t: B_1,\ldots, B_m\to C$, and $id_{B_i}: B_i\to B_i$ for $1\leq i\leq m$, the following law holds
			\[
			t\langle id_{B_1},\ldots, id_{B_m}\rangle
			\]
			\item Associativity: For $t: C_1,\ldots, C_l\to D$, and collections $\{u_i: B^1_i,\ldots, B^{k_i}_i\}_{1\leq i\leq l}$ and $\{v_i^j: \Delta^j_i\to B^j_i\}_{1\leq j\leq k_i, 1\leq i\leq l}$, 
			\[
			t\langle u_1,\ldots, u_l\rangle \langle v^1_1,\ldots, v^{k_1}_1,\ldots, v^1_l,\ldots, v^{k_l}_l\rangle  = t\langle u_1\langle v^1_1,\ldots, v^{k_1}_1\rangle,\ldots, u_l\langle v^1_l,\ldots, v^{k_l}_l\rangle \rangle \
			\]
		\end{enumerate}
	\end{definition}
	We follow the convention in \cite{psaville}  of representing substitution in a clone by $t[u_1,\ldots, u_n]$ and substitution in a multicategory by $t\langle u_1,\ldots, u_n\rangle$.

	As with the case of clones, it is possible to obtain a category by restricting a multicategory $\bb{C}$ to its unary contexts. Define $\overline{\bb{C}}$ as having the same objects as $\bb{C}$, and hom-sets by $\Hom{\overline{\bb{C}}}{A}{B} = \Mhom{\bb{C}}{A}{B}$. Identities are as in the multicategory: $id_A\in \Mhom{\bb{C}}{A}{A}$, and composition is given by $f:A\to B,g: B\to C\mapsto g\langle f\rangle $. 
	
	\begin{definition}
		A morphism of multicategories, $f: \bb{C}\to \bb{D}$ is given by the following data:
		\begin{enumerate}
			\item A function, $f_{ob}: \ob{\bb{C}}\to \ob{\bb{D}}$
			\item A family of maps, $ f_{A_1,\ldots, A_n;B}: \Mhom{\bb{C}}{A_1,\ldots, A_n}{B}\to \Mhom{\bb{D}}{fA_1,\ldots, fA_n}{fB}$
		\end{enumerate}
		that satisfy the following axioms
		\begin{enumerate}
			\item Identities are preserved: $f(id_A) = id_{FA}$.
			\item Composition is preserved: $f(t\langle u_1,\ldots, u_m\rangle) = f(t)\langle f(u_1),\ldots, f(u_m)\rangle$.
		\end{enumerate}
	\end{definition}
	
	Thus, we have a category $\category{MultiCat}$ of multicategories and their morphisms. The mapping $\overline{(-)}$ clearly extends to a functor to $\category{Cat}$. 
	
	\begin{lemma}\cite{hermida}\label{disc-mult}
		The functor $\overline{(-)}: \category{MultiCat}\to \category{Cat}$ has a left adjoint $J$. 
	\end{lemma}
	\begin{proof}Let $\mcal{C}$ be a category.
		Since a multicategory does not require any projection operations, we can leave the hom-sets $\Mhom{J\mcal{C}}{A_1,\ldots, A_n}{B}$ empty whenever $n > 1$. Thus, we define the category ${J}\mcal{C}$ as:
		\begin{enumerate}
			\item Objects: $\ob{J\mcal{C}} = \ob{\mcal{C}}$,
			\item Homsets: 
			\begin{align*}
				\Mhom{J\mcal{C}}{A_1,\ldots, A_n}{B} = \begin{cases}
					\Hom{\mcal{C}}{A_1}{B}, &\text{ if $n = 1$}\\
					\emptyset, &\text{ otherwise}
				\end{cases}
			\end{align*}
			\item Identities: $id_A \in \Hom{\mcal{C}}{A}{A} = \Mhom{J\mcal{C}}{A}{A}$.
			\item Composition: As in the category: if $t: B\to C$ and $u: A\to B$, then, $t\circ u: A\to C$. Since there are no maps of other arities, this completely defines composition.
		\end{enumerate}
		Also note that $\overline{J\mcal{C}}$ returns $\mcal{C}$. Therefore, we set $\eta_{\mcal{C}}: \mcal{C}\to \overline{J\mcal{C}}$ to be the identity functor. This pair is also universal and this is easy to verify due to the simplicity of the construction.
	\end{proof}

	Every clone determines a multicategory in a canonical way. Let $\bb{C}$ be a clone, and define the multicategory $M\bb{C}$ to have the same objects and same homsets as $\bb{C}$. Identities are given by the first projection, i.e, $id_A \in \Mhom{M\bb{C}}{A}{A}$ is $pr^1_{A}\in \Mhom{\bb{C}}{A}{A}$. Finally composition is given as follows: if $t: B_1,\ldots, B_m\to C$ is a multimap and $u_i: A^1_i,\ldots, A^{k_i}_i\to B_i$ for $1\leq i\leq m$ are a collection of maps, define
	\[
	t\langle u_1,\ldots, u_m\rangle  = t[ u_1[pr^{1,1}_{A_\bullet^\bullet},\ldots, pr^{1,k_i}_{A^\bullet_\bullet} ],\ldots, u_m[pr^{m,1}_{A_\bullet^\bullet},\ldots, pr^{m,k_m}_{A_\bullet^\bullet}]  ]
	\]	
	where $pr^{i,j}_{A_\bullet^\bullet}$ is the projection $A^1_1,\ldots, A^{k_1}_1,\ldots, A^1_m,\ldots, A^{k_m}_m \to A^j_i$. A clone morphism $f: \bb{C}\to \bb{D}$ automatically becomes a morphism of multicategories, $Mf: M\bb{C}\to M\bb{D}$ by virtue of preserving projections and composition:
	\begin{align*}
		f(t\langle u_1,\ldots, u_m\rangle) &= f(t[ u_1[pr^{1,1}_{A_\bullet^\bullet},\ldots, pr^{1,k_i}_{A^\bullet_\bullet} ],\ldots, u_m[pr^{m,1}_{A_\bullet^\bullet},\ldots, pr^{m,k_m}_{A_\bullet^\bullet}]  ])\\
		& = f(t)[f(u_1)[pr^{1,1}_{fA^\bullet_\bullet},\ldots, pr^{1,k_i}_{fA_\bullet^\bullet}],\ldots, f(u_m)[pr^{m,1}_{fA_\bullet^\bullet},\ldots, pr^{m,k_m}_{fA_\bullet^\bullet}]]\\
		& = f(t)\langle f(u_1),\ldots, f(u_m)\rangle 
	\end{align*}
	Hence, we obtain a functor, $M: \category{Clone}\to \category{MultiCat}$. As a result, we can transport the cartesian category $\category{Set}$ to  $\category{MultiCat}$ through \category{Clone} as $M\Cl{\category{Set}}$, which has
	\begin{enumerate}
		\item Objects: Sets,
		\item Hom-sets: $\Mhom{M\Cl{\category{Set}}}{A_1,\ldots, A_n}{B} = \Hom{\category{Set}}{A_1\times\ldots\times A_n}{B}$,
		\item Identities: $id_A \in \Mhom{M\Cl{\category{Set}}}{A}{A}$ is the identity function on $A$,
		\item Composition: If $t: B_1\times\ldots\times B_m\to C$, and $u_i: A^1_i\times\ldots\times A^{k_i}_i\to B_i$ for $1\leq i\leq m$ are given, 
		\begin{align*}
			t\langle u_1,\ldots, u_m\rangle &= t[u_1[\pi^{1,1}_{A^\bullet_\bullet},\ldots, \pi^{1,k_1}_{A^\bullet_\bullet}],\ldots, u_m[\pi^{m,1}_{A^\bullet_\bullet},\ldots, \pi^{m,k_m}_{A^\bullet_\bullet}]] \\
			& = t\circ (u_1\times\ldots \times u_m)
		\end{align*}
		where $u_1\times\ldots\times u_m$ is the unique map, $A_1^1\times\ldots A^{k_1}_1\times\ldots \times A^1_m\times\ldots\times A^{k_m}_m\to B_1\times\ldots\times B_m$ induced in the following diagram:
		\[\begin{tikzcd}[cramped]
			& {A_1^1\times\ldots A^{k_1}_1\times\ldots \times A^1_m\times\ldots\times A^{k_m}_m} \\
			{A^1_1\times\ldots\times A^{k_1}_1} & \dots & {A^1_m\times\ldots\times A^{k_m}_m} \\
			\\
			{B_1} & \dots & {B_m} \\
			& {B_1\times\ldots\times B_m}
			\arrow["{[\pi^{1,1},\ldots, \pi^{1,k_1}]}"{description}, from=1-2, to=2-1]
			\arrow["{[\pi^{m,1},\ldots, \pi^{m,k_m}]}"{description}, from=1-2, to=2-3]
			\arrow["{u_1}", from=2-1, to=4-1]
			\arrow[from=5-2, to=4-1]
			\arrow[from=5-2, to=4-3]
			\arrow["{u_m}"', from=2-3, to=4-3]
		\end{tikzcd}\]
	\end{enumerate}

	This construction uses only the monoidal aspects of the cartesian product $\times$ on \category{Set}. This suggests that any monoidal category can be made into a multicategory.

	\subsection{Multicategories and Monoidal Categories}

	The following construction can be replicated for any monoidal category, by choosing a particular bracketing. For simplicity, we consider strict monoidal categories.
	For a (strict) monoidal category $(\mcal{C},\otimes, I)$, define the multicategory $U\mcal{C}$ as having:
	\begin{enumerate}
		\item Objects: $\ob{U\mcal{C}} = \ob{\mcal{C}}$,
		\item Hom-sets: $\Mhom{U\mcal{C}}{A_1,\ldots, A_n}{B} := \Hom{\mcal{C}}{A_1\otimes\ldots\otimes A_n}{B}$,
		\item Identities: These are the same as $\mcal{C}$: $id_A \in \Hom{\mcal{C}}{A}{A}$,
		\item Substitution: If $t:B_1\otimes \ldots \otimes B_m\to C $ and $u_i: A^1_i\otimes\ldots\otimes A^{k_i}_i\to B_i$, then define 
		\[
		t\langle u_1,\ldots, u_n\rangle : = t\circ (u_1\otimes \ldots \otimes u_m)
		\]
	\end{enumerate}
	If $f: (\mcal{C},\otimes, I)\to (\mcal{D},\odot, J)$ is a strict monoidal functor, define the multicategory morphism $Uf: U\mcal{C}\to U\mcal{D}$ as:
	\begin{enumerate}
		\item On objects: $Uf(A) = f(A)$.
		\item On hom-sets: $\Mhom{U\mcal{C}}{A_1,\ldots, A_n}{B} \to \Mhom{U\mcal{D}}{fA_1,\ldots, fA_n}{B}$ is the mapping $t: A_1\otimes\ldots\otimes A_n\to B \mapsto f(t): fA_1\otimes\ldots\otimes f_An\to fB$.
	\end{enumerate}
	This mapping extends to a functor $U:\category{StrictMonCat}\to \category{MultiCat}$. This mapping has a left adjoint, $F$, which makes the adjunction monadic \cite{hermida}. The construction is briefly sketched in \cite{leinster}, and we recapitulate it.  For a multicategory $\bb{C}$, define the strict monoidal category $F\bb{C}$  to have:
	\begin{enumerate}
		\item Objects: $\ob{F\bb{C}}$ is the free monoid on $\ob{\bb{C}}$
		\item Hom-sets: A typical morphism $t: [A_1,\ldots, A_n]\to [B_1,\ldots, B_m]$ consists of an $m-$partition of $n$: $P= (1=i_0\leq \ldots \leq i_m = n)$ and an $m-$tuple $(t_1,\ldots, t_m)$ where $t_i: A_{i-1},\ldots, A_{i}\to B_i $ is a multimorphism in $\bb{C}$.	
	\end{enumerate}
	
	This is just a monoidal counterpart to the construction of the free cartesian category on a clone sketched in Lemma \ref{free-cart}. We then obtain a diagram similar to \ref{adjunctions}: 
	
	\begin{equation}
		\begin{tikzcd}[cramped]
			&& {\category{StrictMonCat}} \\
			\\
			{\category{MultiCat}} \\
			\\
			&& {\category{Cat}}
			\arrow[""{name=0, anchor=center, inner sep=0}, "{\overline{(-)}}"{description}, curve={height=18pt}, from=3-1, to=5-3]
			\arrow[""{name=1, anchor=center, inner sep=0}, "J"{description}, curve={height=18pt}, from=5-3, to=3-1]
			\arrow[""{name=2, anchor=center, inner sep=0}, "F"{description}, curve={height=18pt}, from=3-1, to=1-3]
			\arrow[""{name=3, anchor=center, inner sep=0}, "U"{description}, curve={height=24pt}, from=1-3, to=3-1]
			\arrow[""{name=4, anchor=center, inner sep=0}, "{U'}"{description}, from=1-3, to=5-3]
			\arrow[""{name=5, anchor=center, inner sep=0}, "{F'}"{description}, curve={height=30pt}, from=5-3, to=1-3]
			\arrow["\dashv"{anchor=center, rotate=-146}, draw=none, from=1, to=0]
			\arrow["\dashv"{anchor=center, rotate=-180}, draw=none, from=5, to=4]
			\arrow["\dashv"{anchor=center, rotate=152}, draw=none, from=2, to=3]
		\end{tikzcd}
	\end{equation}
	where $U' := \overline{(-)}\circ U$ and $F'= F\circ J$. As with clones, the functor $U'$ is the forgetful functor that forgets monoidal structure, and $F'$ is the usual free strict monoidal category functor which takes a category $\mcal{C}$ to the monoidal category $F'\mcal{C}$ with objects set equal to the free monoid on $\ob{\mcal{C}}$ and hom-sets given by 
	\begin{align*}
		\Hom{F'\mcal{C}}{[A_1,\ldots, A_n]}{[B_1,\ldots, B_m]} = \begin{cases}
			\emptyset,& \text{if $n \neq m$ }\\
			\Hom{\mcal{C}}{A_1}{B_1}\times\ldots \times \Hom{\mcal{C}}{A_n}{B_n},& \text{otherwise}
		\end{cases}
	\end{align*}

	\subsection{Algebras for Multicategories}

	Multicategories serve as a `meta-compositional structure'  in that they define a theory of composition. For example, \cite[Chapter 6]{act} defines the multicategory of \category{Cospans}, which serve as a theory of composition for circuit diagrams.
	An algebra for a multicategory, then, is a realization of the theory of composition. For example, algebras for  the multicategory \category{Cospans} are actual circuit diagrams.

	\begin{definition}{\cite[Definition 2.15]{gould}}
		An algebra for a multicategory $\bb{C}$ in a multicategory $\bb{D}$ is a morphism of multicategories, $A: \bb{C}\to \bb{D}$. 
		
		We can specialize to the multicategory of sets just like in the case of clones to obtain: A $\bb{C}$-algebra for a multicategory $\bb{C}$ is a multicategory morphism, $A: \bb{C}\to M\Cl{\category{Set}}$.		
	\end{definition}
	
	To define the category of algebras, we need the notion of a transformation, which was hinted in Remark \ref{transformation}.
	\begin{definition}[\cite{hermida, gould, psaville,leinster}]
		Let $f,g: \bb{C}\to \bb{D}$ be a morphism of multicategories. A transformation $\eta: f\Rightarrow g$ is given by a collection of morphisms, $\{\eta_A \in \Mhom{\bb{D}}{fA}{gA}\}_{A\in \ob{\bb{C}}}$ such that for all multimaps, $t: A_1,\ldots, A_n\to B$, the following equation holds:
		\[
		\eta_B \langle f(t)\rangle = g(t)\langle \eta_{A_1},\ldots, \eta_{A_n}\rangle
		\] 
	\end{definition}
	
	As before, there is a category $\Hom{\category{MultiCat}}{\bb{C}}{\bb{D}}$ of morphisms of multicategories, $f: \bb{C}\to \bb{D}$ and transformations between them. Though this fact is enough for our purposes, we mention that this notion of transformation makes the category $\category{MultiCat}$ into a 2-category \cite{gould, hermida}.

	\begin{definition}
		Let $\bb{C}$ be a multicategory.
		\begin{enumerate}
			\item  A morphism of $\bb{C}$-algebras $A,B: \bb{C}\to M\Cl{\category{Set}}$ is a transformation betweeen them.
			\item The category $\bb{C}$-\category{Alg} is the category of $\bb{C}$-algebras, and morphisms between them.
		\end{enumerate} 	 
	\end{definition}
	
	Hence, for every multicategory $\bb{C}$, it is possible to define the category  $\bb{C}$-\category{Alg} in a manner analogous to the way it is done for abstract clones. 
	
	\section{Algebras for Multi-Ary Structures}
	
	With clones and multicategories serving as templates, we  state a couple of requirements that a general category of multi-ary structures, M-\category{Struct} should have, so that a category of algebras can be defined:
	\begin{enumerate}
		\item The category M-\category{Struct} should have an appropriate notion of morphism and  of a 2-morphism, i.e, M-\category{Struct} should form a 2-category.
		
		\item There should be a canonical object, $\bb{M} \in $ M-\category{Struct}. 
	\end{enumerate}
	
	In the case of clones and multicategories (i.e,  M-\category{Struct} = \category{Clones} or M-\category{Struct}= \category{MultiCat}) condition (1) is satisfied by their respective notions of morphisms and transformations. Condition (2) is satisfied by the clone/multicategory of sets.
	
	Given these requirements,  the category of algebras for a specific multi-ary structure $\bb{C} \in$ M-\category{Struct} can be defined as the category of morphisms $\bb{C}\to \bb{M}$ and 2-morphisms between them, just like in the case of clones and multicategories.

\chapter{Multi-ary Structures for Effectful Programming}

\section{Premulticategories and their Algebras}
	
	\subsection{Premulticategories Model Effectful Languages}
	
	The simply typed $\lambda$-calculus is syntax for  pure computation: any program that can be constructed in the $\lambda$-calculus cannot interact with the environment, i.e, programs cannot throw exceptions, read or write text, etc . However, these features are required for a programming language to be useful
	
	Moggi introduced two notable extensions of the $\lambda$-calculus, namely the monadic metalanguage, $\lambda_{ml}$ and the computational lambda calculus, $\lambda_c$ to model such effectful computations. These can  be thought of as prototypes of the functional programming languages Haskell and OCaml respectively. 
	
	The language $\lambda_{ml}$ separates pure terms and effectful terms explicitly with a Kleisli triple, while in $\lambda_c$, every term is assumed to have an effect. There is one major problem with modelling  a language like $\lambda_c$ in a multicategory.		Assume that we are working in a multicategory. Then,  for a term $t: B_1,B_2\to C$ and terms $u_1: A_1\to B_1$ and $u_2: A_2\to B_2$, we have the following equality: 
	
	\begin{align*}
		t[u_1, id_{B_2}][ id_{A_1}, u_2] &= t[u_1[id_{A_1}], id_{B_2}[u_2]]\\& = t[id_{B_1}[u_1], u_2[id_{A_2}]]\\& = t[id_{B_1}, u_2][u_1, id_{A_2}]
	\end{align*}

	However, this equality does not hold when effects are taken into account: Let $t: \mathbb{N},\mathbb{N}\to \mathbb{N}$ be the function 
	$a,b\mapsto a+b$.  Let $u_1: \mathbb{N},\mathbb{N}\to \mathbb{N}$ be the function
	 that prints ``hello'', and acts on inputs as  $a,b\mapsto a$  	and  let $u_2: \mathbb{N}, \mathbb{N}\to \mathbb{N}$ be the function that prints ``world'' and acts on inputs as  $a,b\mapsto a$. Then, the function $t[u_1, id_{\mathbb{N}}][id_{\mathbb{N}}, u_2]$ is the function that adds two numbers and prints ``hello'' followed by ``world'', while $t[id_{\mathbb{N}}, u_2][u_1, id_{\mathbb{N}}]$ adds two numbers and prints ``world'' followed by ''hello''. Therefore, while working with a language like $\lambda_c$, it is important to note how substitution is sequenced. In the example, the term $t[u_1,u_2]$ could mean two different things: evaluate $u_1$ and then evaluate $u_2$ or evaluate $u_2$ and then evaluate $u_1$. Hence, we have to work in a language that has a `do nothing' program:
	\begin{prooftree}
		\AxiomC{$\vdash A:\text{Type}$}
		\UnaryInfC{$x:A\vdash x:A$}
	\end{prooftree}
	In the language, terms can be substituted only one at a time: 
	\begin{prooftree}
		\AxiomC{$\Gamma, x:B,\Gamma'\vdash t:C$}
		\AxiomC{$\Delta\vdash u: B$}
		\BinaryInfC{$\Gamma,\Delta,\Gamma'\vdash t[\Gamma,u,\Gamma']$}
	\end{prooftree}
	To deal with languages of this sort, Staton and Levy developed the notion of a premulticategory.
	
	\begin{definition}\cite{staton-levy}
		A premulticategory $\bb{C}$ consists of the following data:
		\begin{enumerate}
			\item A set of objects $\ob{\bb{C}}$,
			\item For every sequence of objects, $A_1,\ldots, A_n, B$, a set of terms $\Mhom{\bb{C}}{A_1,\ldots, A_n}{B}$,
			\item Identity morphisms, $id_A \in \Mhom{\bb{C}}{A}{A}$,
			\item A unary substitution operation:
			\begin{align*}
				\Mhom{\bb{C}}{\Gamma,A,\Gamma'}{B} \times \Mhom{\bb{C}}{\Delta}{A} &\to \Mhom{\bb{C}}{\Gamma,\Delta,\Gamma'}{B}\\
				t,u&\mapsto [\Gamma, u,\Gamma']
			\end{align*}
		\end{enumerate}
		such that the following laws hold:
		\begin{enumerate}
			\item Left unitality: For any $u: \Delta \to A$, 
			\[
			id_A[\emptyset, u,\emptyset] = u
			\]
			\item Right unitality: For any $t: \Gamma,A,\Gamma'\to B$, 
			\[
			t[\Gamma, id_A, \Gamma'] = t
			\]
			\item Associativity: For  $t: \Gamma_1,B,\Gamma_1'\to C$, $u: \Gamma_2,A,\Gamma_2'\to B$ and $v: \Delta\to A$, we have
			\[
			t[\Gamma_1, u,\Gamma_1'][\Gamma_1,\Gamma_2, v,\Gamma_2',\Gamma_1'] = t[\Gamma_1, u[\Gamma_2,v,\Gamma_2'], \Gamma']
			\]
		\end{enumerate}
	\end{definition}
	
	First, notice that every multicategory is a premulticategory. Given  morphisms $t: A_1,\ldots, A_i,\ldots, A_n \to B$ and $u: \Delta\to A_i$, define
	\[
	t[A_1,\ldots, A_{i-1}, u,A_{i+1},\ldots,A_n] : = t\langle id_{A_1},\ldots, id_{A_i}, u, id_{A_{i+1}},\ldots, id_{A_n}\rangle 
	\]
	However, as observed above, a multicategory also satisfies the property that the order of substitution does not matter. There may be other such terms in a premulticategory for which the order of substitution does not matter, and these are called central elements \cite{notions-of-computation}. 
	\begin{definition}\cite{staton-levy}\label{central-term}
		Let $\bb{C}$ be a premulticategory, and let $u: \Gamma\to A$ and  $u': \Gamma'\to B$ be morphisms. 
		\begin{enumerate}
			\item The morphims $u$  and $u'$ are said to commute with each other if for any chosen term $t: \Delta_1, A,\Delta_2,B,\Delta_3\to C$, and any chosen term $t': \Delta_1',B,\Delta_2', A,\Delta_3'\to C$, there are equalities:			
			\begin{align*}
				t[\Delta_1, u,\Delta_2,B,\Delta_3][\Delta_1,\Gamma, \Delta_2, v,\Delta_3]& = t[\Delta_1,A,\Delta_2,v,\Delta_3][\Delta_1,u,\Delta_2,\Gamma',\Delta_3]\\			
				t'[\Delta_1', B,\Delta_2', u, \Delta_3'][\Delta_1', v, \Delta_2', \Gamma,\Delta_3' ] &= t'[\Delta_1', v,\Delta_2', A,\Delta_3'][\Delta_1',\Gamma', \Delta_2', u,\Delta_3']
			\end{align*} 
			This is essentially saying that $u$ and $u'$ commute with each other if it does not matter in which order they are substituted in another term. 
			
			\item The morphism $u$ is said to be central if it commutes with every other morphism in $\bb{C}$.			
		\end{enumerate}
	\end{definition}
	\begin{example}
		The identity morphism $id_A:A\to A$ in a premulticategory is always central.
	\end{example}

	A multicategory may be characterized as a premulticategory where all the morphisms are central \cite{staton-levy}. This was how Lambek initially defined a multicategory \cite{ded-2}. This fact can easily be deduced from Lemma \ref{multicat-assoc} which is proved in Chapter \ref{chapter5}. For now we note that given a morphism $t: A_1,\ldots, A_n\to B$ in a premulticategory and a collection of central morphisms $u_i: \Delta_i\to A_i$ for $1\leq i\leq n$, we can define a simultaneous substitution operation as:
	\[
	t\langle u_1,\ldots, u_n\rangle = t[\emptyset, u_1,A_2,\ldots, A_n][\Delta_1, u_2,A_3,\ldots, A_n]\ldots[\Delta_1,\ldots, \Delta_{n-1}, u_n,\emptyset]
	\]	
	and this operation satisfies the same axioms of a multicategory. In particular if $u_1,\ldots, u_n$ and $v^1_1,\ldots, v^{k_1}_1,\ldots, v^{1}_n,\ldots, v^{k_n}_n$ are a collection of central morphisms, 
	\begin{align}\label{mult-assoc}
		t\langle u_1,\ldots, u_l\rangle \langle v^1_1,\ldots, v^{k_1}_1,\ldots, v^1_l,\ldots, v^{k_l}_l\rangle  = t\langle u_1\langle v^1_1,\ldots, v^{k_1}_1\rangle,\ldots, u_l\langle v^1_l,\ldots, v^{k_l}_l\rangle \rangle \
	\end{align}

	The notion of centrality behaves analogously to the way it would in common algebraic structures.

	\begin{lemma}\label{cent-closed}
		Let $\bb{C}$ be a premulticategory. If $t: \Gamma,A,\Gamma'\to B$ and $u: \Delta\to A$ are both central, so is $\theta:= t[\Gamma,u,\Gamma']: \Gamma,\Delta,\Gamma'\to B$.
	\end{lemma}
	\begin{proof}
		Let $r: \Delta_1,B,\Delta_2,C,\Delta_3\to D$ and $v: \Lambda\to C$ be arbitrary multimaps. Then, we need to show the equality: 
		\begin{align*}
			&r[\Delta_1, \theta, \Delta_2,C,\Delta_3][\Delta_1, \Gamma,A,\Gamma', \Delta_2, v, \Delta_3]\\
			& = r[\Delta_1,B,\Delta_2,v,\Delta_3][\Delta_1, \theta, \Delta_2,\Lambda,\Delta_3]
		\end{align*}
		
		Expanding:
		\begin{align*}
			&r[\Delta_1, \theta, \Delta_2,C,\Delta_3][\Delta_1, \Gamma,\Delta,\Gamma', \Delta_2, v, \Delta_3] \\
			& = r[\Delta_1, t[\Gamma,u,\Gamma'], \Delta_2,C,\Delta_3][\Delta_1, \Gamma,\Delta,\Gamma', \Delta_2, v, \Delta_3] \\
			& = r[\Delta_1,t,\Delta_2,C,\Delta_3][\Delta_1,\Gamma,u,\Gamma',\Delta_2,C,\Delta_3][\Delta_1, \Gamma,\Delta,\Gamma', \Delta_2, v, \Delta_3] 
		\end{align*}
		
		Let $s =r[\Delta_1,t,\Delta_2,C,\Delta_3]: \Delta_1,\Gamma,A,\Gamma',\Delta_2,C,\Delta_3\to D $. Then, since $u$ is central, 
		\begin{align*}
			& = s[\Delta_1,\Gamma,u,\Gamma',\Delta_2,C,\Delta_3][\Delta_1, \Gamma,\Delta,\Gamma', \Delta_2, v, \Delta_3]  \\
			&= s[\Delta_1,A,\Gamma', \Delta_2, v,\Delta_3][\Delta_1,\Gamma,u,\Gamma',\Delta_2,\Lambda,\Delta_3]\\
			& = r[\Delta_1,t,\Delta_2,C,\Delta_3][\Delta_1,A,\Gamma', \Delta_2, v,\Delta_3][\Delta_1,\Gamma,u,\Gamma',\Delta_2,\Lambda,\Delta_3]
		\end{align*}
		Now since $t$ is central, 
		\begin{align*}
			r[\Delta_1,t,\Delta_2,C,\Delta_3][\Delta_1,A,\Gamma,\Delta_2,v,\Delta_3] = r[\Delta_1,B,\Delta_2,v,\Delta_3][\Delta_1,t,\Delta_2,\Lambda,\Delta_3]
		\end{align*}
		So, the above chain of equalities continues as:
		\begin{align*}
			&r[\Delta_1, \theta, \Delta_2,C,\Delta_3][\Delta_1, \Gamma,\Delta,\Gamma', \Delta_2, v, \Delta_3]\\
			& = r[\Delta_1,B,\Delta_2,v,\Delta_3][\Delta_1,t,\Delta_3,\Lambda,\Delta_3][\Delta_1,\Gamma,u,\Gamma',\Delta_2,\Lambda,\Delta_3]\\
			& = r[\Delta_1,B,\Delta_2,v,\Delta_3][\Delta_1,t[\Gamma,u,\Gamma'], \Delta_2,\Lambda,\Delta_3]\\
			& = r[\Delta_1,B,\Delta_2,v,\Delta_3][\Delta_1,\theta, \Delta_2,\Lambda,\Delta_3]
		\end{align*}
		which proves the equality. 
	\end{proof}

	This allows us to consider the subset of central morphisms of a premulticategory as a separate entity. 
	\begin{definition}
		Let $\bb{C}$ be a premulticategory. The centre of $\bb{C}$,  $Z(\bb{C})$ has the same objects as $\bb{C}$, and has hom-sets given by:
		\[
		\Mhom{Z\bb{C}}{A_1,\ldots, A_n}{B} = \{t: A_1,\ldots, A_n\to B \mid t \text{ is central}\}
		\]
		Since $Z\bb{C}$ has all central morphisms, it is a multicategory.
	\end{definition}

	\begin{example}
		Note that the multicategory of sets forms a premulticategory. A premulticategory that is not a multicategory is the premulticategory of stateful functions \cite{staton-levy}. For a fixed set of `states', $S$ define the premulticategory $\bb{C}_S$ as:
		\begin{enumerate}
			\item Objects: Sets
			\item Hom-sets: These are defined as:
			\[
			\Mhom{\bb{C}_S}{A_1,\ldots, A_n}{B} = \Hom{\category{Set}}{(A_1\times\ldots\times A_n)\times S}{B\times S}
			\]
			We write elements $t: A_1,\ldots, A_n\to B$ as tuples $(t, t_s)$ where $t: A_1\times\ldots\times A_n\times S\to B$ and $t_S: A_1\times \ldots \times A_n\times S\to S$. So a typical morphism changes is a program that takes an input and gives an output, while also taking an initial state and returning the state after the program has been executed.

			\item Identities: $id_A$ is the identity map, $A\times S\to A\times S$. 
			
			\item Substitution: If $(t,t_S): \Gamma,A,\Gamma'\to B$ and $(u,u_S): \Delta\to A$ are maps, then, the map obtained by substitution is the function that acts on terms as:
			\[
			(\vec{a}, \vec{b},\vec{c}, s) \mapsto t(\vec{a}, u(\vec{b},s), \vec{c},u_S(\vec{b},s))
			\]
			and acts on states as:
			\[
			(\vec{a}, \vec{b}, \vec{c}, s) \mapsto t_S(\vec{a}, u(\vec{b},s), \vec{c},u_S(\vec{b},s))
			\]
		\end{enumerate}
	\end{example}

	As with  multicategories and clones, restricting a premulticategory to its unary contexts gives a category. Before we show that this defines a functor, we recall the notion of a morphism of premulticategories.
	\begin{definition}\cite{staton-levy}
		A morphism of premulticategories, $f: \bb{C}\to \bb{D}$ consists of
		\begin{itemize}
			\item A function, $f_{ob}: \ob{\bb{C}}\to \ob{\bb{D}}$
			\item A family of maps, $f_{\Gamma;A}: \Mhom{\bb{C}}{A_1,\ldots, A_n}{B}\to \Mhom{\bb{D}}{fA_1,\ldots, fA_n}{fB}$
		\end{itemize}
		such that the following axioms are satisfied:
		\begin{enumerate}
			\item Identities are preserved: $f(id_A) = id_{fA}$.
			\item Composition is preserved: $f(t[\Gamma, u,\Gamma']) = f(t)[f\Gamma, f(u), f\Gamma']$.
		\end{enumerate}
		The category of premulticategories and their morphisms forms a category \category{PreMultCat}.
	\end{definition}

	\begin{lemma}
		There is a mapping $\overline{(-)}: \category{PreMultCat}\to\category{Cat} $ that restricts a premulticategory to its unary contexts, and this functor has a left adjoint, $J\dashv \overline{(-)}$.
	\end{lemma}
	\begin{proof}
		For a premulticategory $\bb{C}$, define the category $\overline{\bb{C}}$ to have $\ob{\overline{\bb{C}}} = \ob{\bb{C}}$ and to have hom-sets $\Hom{\overline{\bb{C}}}{A}{B} = \Mhom{\bb{C}}{A}{B}$. For premulticategory morphism $f: \bb{C}\to \bb{D}$, define the functor $\overline{f}$ to act on objects and hom-sets just like how it does in \category{PreMultCat}.
		For a category $\mcal{C}$, define $J\mcal{C}$ to be the discrete multicategory on a category (Lemma \ref{disc-mult}). The universality of $J\mcal{C}$ is straightforward to verify.	
	\end{proof}

	\subsection{Premulticategories and Premonoidal Categories}
	
	Premonoidal categories are structures that model effectful languages with product types. Given a  pair $(a,b)$ in an effectful language, there are two distinct ways to evaluate it: evaluate $a$ first and then evaluate $b$ or the other way around. Hence, if a category $\bb{C}$ interprets effectful languages with product types, there need to be functors:
	\[
	\ob{\bb{C}}\times \bb{C} \to \bb{C}, \quad\text{ and }\quad \bb{C}\times \ob{\bb{C}}\to \bb{C}
	\]
	where the first functor corresponds to evaluating the element on the right of a pair, and the second one corresponds to evaluating an element on the left of a pair. This kind of 
	`sesquiness' is the essential element in a premonoidal category.

	First, note that a strict monoidal category is a monoid in the cartesian monoidal category $\category{Cat}$ of small categories and their functors. Similarly, it is possible to give \category{Cat} another monoidal structure with the funny tensor. 
	\begin{definition}
		Let $\mcal{C}$ and $\mcal{D}$ be small categories. Define their funny tensor $\mcal{C}\fun\mcal{D}$ to be the pushout: 
		\[\begin{tikzcd}[cramped]
			{\ob{\mcal{C}}\times \ob{\mcal{D}}} && {\ob{\mcal{C}}\times \mcal{D}} \\
			\\
			{\mcal{C}\times \ob{\mcal{D}}} && {\mcal{C}\fun \mcal{D}}
			\arrow["{i_{\mcal{C}}\times 1}"', from=1-1, to=3-1]
			\arrow[""{name=0, anchor=center, inner sep=0}, "{1\times i_{\mcal{D}}}", from=1-1, to=1-3]
			\arrow["{j_1}"', from=3-1, to=3-3]
			\arrow["{j_2}", from=1-3, to=3-3]
			\arrow["\lrcorner"{anchor=center, pos=0.125, rotate=180}, draw=none, from=3-3, to=0]
		\end{tikzcd}\]
		where $i_{\mcal{D}}$ is the canonical functor, $\ob{\mcal{D}}\to \mcal{D}$. 
	\end{definition}

	\begin{proposition}\cite{alg-catkelly} 
		The funny tensor defines a symmetric monoidal closed structure  with unit $1$ (the terminal category) on \category{Cat}. Further, it is the only other closed monoidal structure on $\category{Cat}$ other than the cartesian monoidal structure. 
	\end{proposition}

	A strict premonoidal category is a monoid in this monoidal category \cite{notions-of-computation}. Explicitly, using the universal property of the pushout, this boils down to the following data: 
	\begin{enumerate}
		\item Two bifunctors:
		\[- \rtimes -:  \ob{\mcal{C}}\times \mcal{C}\to \mcal{C}, \quad \text{ and } \quad - \ltimes -: \mcal{C}\times \ob{\mcal{C}}\to \mcal{C}\]
		\item An object $I \in \mcal{C}$
	\end{enumerate}
	such that the following conditions are satisfied: 
	\begin{enumerate}
		\item The bifunctors agree on objects: $A\rtimes B = A\ltimes B$ for all $A,B\in \ob{\mcal{C}}$. We denote by $A\otimes B$ the image of $(A,B)$ under these bifunctors.
		\item  Unitality: For all $	A\in \ob{\mcal{C}}$, $f: A\to A'$ and 
		\begin{align*}
			I\otimes A = A \quad &\text{ and } \quad A\otimes I = A\\
			id_I\rtimes f = f \quad&\text{ and } \quad f\ltimes id_I = f
		\end{align*}
		\item Associativity: For all $A,B,C\in \ob{\mcal{C}}$:
		\begin{align*}
			(A\otimes B)\times C &= A\otimes (B\otimes C)
		\end{align*}
		For morphisms, $f: A\to A'$, $g: B\to B'$ and $h: C\to C'$, we have equalities:
		\begin{align*}
			(f \ltimes id_B) \ltimes id_C &= f \ltimes id_{B\otimes C}\\
			(id_A \rtimes g)\ltimes id_C & = id_A \rtimes (g\ltimes id_C)\\
			id_A\rtimes (id_B\rtimes h) &=(id_{A\otimes B})\rtimes h
		\end{align*}
		We denote each of the composites above by $f\otimes B\otimes C$, $A\otimes g\otimes C$ and $A\otimes B\otimes h$ respectively.
	\end{enumerate}
	It is possible to iterate tensoring to obtain expressions like $A_1\otimes\ldots A_{i-1}\otimes u\otimes A_{i+1}\otimes\ldots \otimes A_n$ when given some $u: A_i'\to A_i$. The strictness condition ensures that this expression is well-defined.

	\begin{example}

		Every strict premonoidal category $\mcal{C}$ defines a premulticategory $U\mcal{C}$ as follows:
		\begin{itemize}
			\item Objects: $\ob{U\mcal{C}} = \ob{\mcal{C}}$
			\item Hom-sets: $\Mhom{U\mcal{C}}{A_1,\ldots, A_n}{B} = \Hom{\mcal{C}}{A_1\otimes \ldots \otimes A_n}{B}$.
			\item Identities are as in $\mcal{C}$: $id_{A} \in \Mhom{U\mcal{C}}{A}{A}$.
			\item Composition: For a context $\Gamma =A_1,\ldots, A_n$ in $U\mcal{C}$, we write $|\Gamma|$ for the tensor $A_1\otimes\ldots\otimes A_n$.	
			If $t:\Gamma,A,\Gamma'\to B$ is a morphism, and $u: \Delta \to A_i$ is a morphism, then define:
			\begin{align*}
				t[\Gamma, u,\Gamma'] 		&:=t\circ (|\Gamma|\otimes u\otimes |\Gamma'|)
			\end{align*}
		\end{itemize}

		This could have been defined with the notion of a (non-strict) premonoidal category too. We use the strict version since they are simpler to understand in the context of premulticategories.
	\end{example}
	There is also a corresponding notion of \textit{centrality} in a premonoidal category, which is linked to the notion of centrality in a premulticategory.
	\begin{definition}\cite{notions-of-computation}
		Let $\mcal{C}$ be a strict premonoidal category.
		\begin{enumerate}
			\item A morphism $f: A\to A'$ commutes with a morphism $g: B\to B'$ if the following diagrams commute:
			\[\begin{tikzcd}[cramped]
				{A\otimes B} && {A'\otimes B} && {B\otimes A} && {B\otimes A'} \\
				\\
				{A\otimes B'} && {A'\otimes B'} && {B'\otimes A} && {B'\otimes A'}
				\arrow["{f\otimes B}", from=1-1, to=1-3]
				\arrow["{A\otimes g}"', from=1-1, to=3-1]
				\arrow["{f\otimes B'}"', from=3-1, to=3-3]
				\arrow["{A'\otimes g}", from=1-3, to=3-3]
				\arrow["{B\otimes f}", from=1-5, to=1-7]
				\arrow["{g\otimes A}"', from=1-5, to=3-5]
				\arrow["{B'\otimes f}"', from=3-5, to=3-7]
				\arrow["{g\otimes A'}", from=1-7, to=3-7]
			\end{tikzcd}\]
			This corresponds to saying that evaluating either element in a pair first has no result on the final computation.
			
			\item A morphism $f: A\to A'$ is said to be central if it commutes with every other morphism in $\mcal{C}$.			
		\end{enumerate} 	
	\end{definition}
	
	\begin{example}
		The identity maps $id_A: A\to A$ in a premonoidal category $\mcal{C}$ are always central.
	\end{example}

	It is possible to define the centre $Z(\mcal{C})$ of a premonoidal category $\mcal{C}$ as containing all objects of $\mcal{C}$, and all the central morphisms of $\mcal{C}$. Just like how the centre of a premulticategory is a multicategory, the centre of a premonoidal category is a monoidal category \cite{notions-of-computation}. In the construction of a premulticategory from a premonoidal category, these two notions of centrality correspond.

	\begin{proposition}\cite{staton-levy}\label{cent-to-cent}
		Let $\mcal{C}$ be a strict premonoidal category, and let $U\mcal{C}$ be the corresponding premulticategory. A morphism $t: \Gamma\to A$ in $U\mcal{C}$ is central if and only if it is central as a morphism in $\mcal{C}$.
	\end{proposition}

	We now extend the mapping $U$ to functors, and for this we need the  notion of a (strict) premonoidal functor. 
	\begin{definition}
		A strict premonoidal functor between strict premonoidal categories $\mcal{C}$ and $\mcal{D}$ is a functor $f: \mcal{C}\to \mcal{D}$ such that $f(A\otimes B) = f(A)\otimes f(B)$ and $f(I) = I$.
		
		The category \category{PreMonCat} is the category of strict premonoidal categories and strict premonoidal functors.
	\end{definition}
	\begin{remark}
		We do not ask for a premonoidal functor to preserve centrality so as to keep the analogy between premonoidal categories and premulticategories going. However, most sources require that centrality be preserved \cite{notions-of-computation}, \cite{premon-alg}. The reason this requirement is kept in these texts is because they work with general, non-strict premonoidal categories, and non-strict functors. To ensure two non-strict premonoidal functors compose, it is required for certain coherence isomorphisms to map to central ones \cite{roman}. However, since we are working with strict premonoidal categories, the coherence isomorphisms are the identities, which get mapped to identities and this ensures that strict premonoidal functors compose nicely. 
	\end{remark}

	A premonoidal functor, $f$	readily becomes a morphism of premulticategories, $Uf$, by setting:
	\begin{align*}
		\Mhom{U\mcal{C}}{A_1,\ldots, A_n}{B} &\to \Mhom{U\mcal{D}}{fA_1,\ldots, fA_n}{fB}\\
		t &\mapsto ft
	\end{align*}

	Hence, there is a functor, $U: \category{PreMonCat}\to \category{PreMultCat}$.	
	However, constructing a left adjoint $F$ does not seem to be a  straightforward construction involving sums and products of hom-sets (like in the case of clones and multicategories). We suspect that each homset of $F\bb{C}$ is rather some kind of colimit. This prevents us from obtaining a functor that fills in the following diagram: 
	
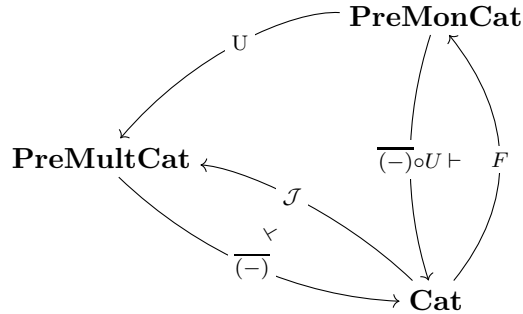
\[\begin{tikzcd}[cramped]
		&& {\category{PreMonCat}} \\
		\\
		{\category{PreMultCat}} \\
		\\
		&& {\category{Cat}}
		\arrow[""{name=0, anchor=center, inner sep=0}, "{\overline{(-)}}"{description}, curve={height=18pt}, from=3-1, to=5-3]
		\arrow[""{name=1, anchor=center, inner sep=0}, "{\mcal{J}}"{description}, curve={height=18pt}, from=5-3, to=3-1]
		\arrow["{\text{U}}"{description}, curve={height=24pt}, from=1-3, to=3-1]
		\arrow[""{name=2, anchor=center, inner sep=0}, "{\overline{(-)}\circ U}"{description}, curve={height=12pt}, from=1-3, to=5-3]
		\arrow[""{name=3, anchor=center, inner sep=0}, "F"{description}, curve={height=30pt}, from=5-3, to=1-3]
		\arrow["\dashv"{anchor=center, rotate=-180}, draw=none, from=3, to=2]
		\arrow["\dashv"{anchor=center, rotate=-149}, draw=none, from=1, to=0]
	\end{tikzcd}\]
	
	Here, the functor $F$ is the free (strict) premonoidal category functor which is defined on a category $\mcal{C}$ as:
	\[
	\sum_{n=1}^{\infty} \fun_{i=1}^n \mcal{C}
	\]
	This exists by the general fact that forgetful functor from the category of monoids on a monoidal category with coproducts has a left adjoint \cite{categories-work}.

	\subsection{Algebras for Premulticategories}	
	
	We now turn to defining the notion of an algebra for a premulticategory. As with clones and multicategories, we can make the following definition.
	\begin{definition}
		An algebra $A$ for a premulticategory $\bb{C}$ in a premulticategory $\bb{D}$ is a premulticategory morphism $A:\bb{C}\to \bb{D}$.
	\end{definition}
	
	The premulticategory of sets, \category{Set} is actually a multicategory, and hence all morphisms in this is central. Saying that an algebra for a premulticategory is a premulticategory morphism into \category{Set} would force all the effectful terms in the syntax to become central, and hence we would lose important information about how effects get evaluated. We could consider the premulticategory of stateful functions $\bb{C}_S$ instead. However, this does not seem to be a `canonical' choice since the set of states $S$ has to be chosen beforehand.

	The definition of a transformation, $\eta: f\Rightarrow g$ for both clones and multicategories made essential use of their respective simultaneous substitution operations. Premulticategories do not come equipped with these, but their centres do. This suggests the following definition. 
	
	\begin{definition}
		Let $f,g:\bb{C}\to \bb{D}$ be premulticategory morphisms. A transformation $\eta: f\Rightarrow g$ between them consists of a family of morphisms, $\{\eta_A: fA\to gA\}$ in $Z(\bb{D})$ such that for all $t:A_1,\ldots, A_n\to B$ there is an equality
		\[
		\eta_{B}\langle ft\rangle = gt\langle \eta_{A_1},\ldots, \eta_{A_n}\rangle 
		\]
	\end{definition}

	\begin{proposition}
		Given premulticategories $\bb{C}$ and $\bb{D}$, the following data determines a category, $\Hom{\category{PreMultCat}}{\bb{C}}{\bb{D}}$:
		\begin{enumerate}
			\item Objects: Premulticategory Morphisms $f: \bb{C}\to \bb{D}$.
			\item Morphisms: Transformations $\eta: f\Rightarrow g$
			\item Identities: The identity premulticategory morphism $id_f:f\Rightarrow f$ has as components  the identity morphisms $id_{fA}: fA\to fA$. Note that each component is central.
			\item Composition: If $\eta: f\Rightarrow g$ and $\epsilon: g\Rightarrow h$ are transformations, define the composite $\epsilon\ast \eta$ to have components $\epsilon_A\langle \eta_A\rangle$.
		\end{enumerate}
	\end{proposition}
	\begin{proof}
		The components of $\epsilon\ast \eta$ are all central by Lemma \ref{cent-closed}. Let $t: A_1,\ldots, A_n\to B$ be a morphism in $\bb{C}$. Then,
		\[
		\epsilon_A\langle gt\rangle = ht\langle \epsilon_{A_1},\ldots, \epsilon_{A_n}\rangle 
		\]
		and
		\[
		\eta_A\langle ft\rangle = gt\langle \eta_{A_1},\ldots, \eta_{A_n}\rangle
		\]
		Using \ref{mult-assoc} and the equations above,
		\begin{align*}
			\epsilon\ast \eta_A\langle ft\rangle &= \epsilon_A\langle\eta_A\rangle\langle ft\rangle\\
			& = \epsilon_A\langle \eta_A\langle ft\rangle\rangle\\
			& = \epsilon_A\langle gt\langle \eta_{A_1},\ldots, \eta_{A_n}\rangle\rangle \rangle\\
			& =\epsilon_A\langle gt \rangle \langle \eta_{A_1},\ldots, \eta_{A_n}\rangle \\
			& = ht\langle \epsilon_{A_1},\ldots, \epsilon_{A_n}\rangle  \langle \eta_{A_1},\ldots, \eta_{A_n}\rangle\\
			& = ht\langle \epsilon_{A_1} \langle \eta_{A_1}\rangle,\ldots, \epsilon_{A_n}\langle \eta_{A_n}\rangle\rangle\\
			& = ht \langle \epsilon\ast \eta_{A_1},\ldots, \epsilon \ast\eta_{A_n}\rangle
		\end{align*}
		which shows that composition is closed in the category. 
		Associativity and unitality follow from componentwise associativity and unitality in $Z(\bb{D})$.
	\end{proof}

	Though the definition is slightly unnatural in its insistence that the components be central, it allows us to define a category of algebras for a premulticategory. 
	\begin{definition}
		For a premulticategory $\bb{C}$, we define its category of algebras in a premulticategory $\bb{D}$ as the category $\Hom{\category{PreMultCat}}{\bb{C}}{\bb{D}}$.
	\end{definition}

	\section{Effectful Multicategories}

	We have seen that effectful languages can be modelled by premulticategories since it is not possible to simultaneously substitute general terms. However, it is possible to pick out a subset of terms, which are called \textit{values}, which behave rather nicely. Informally, a value is a non-effectful term in a programming language. As such, it does not matter when a value is substituted into another term, enabling us to shift them around as we wish. Since all values are \textit{central}, these can be modelled by a multicategory, while all the other effectful terms (computations) are modelled by a premulticategory. This leads to the notion of an effectful multicategory.
	
	\begin{definition}
		An effectful multicategory is a triple, $(\blue{\bb{C}_0}, \red{\bb{C}_1}, J)$ where 
		\begin{enumerate}
			\item $\blue{\bb{C}_0}$ is a multicategory and $\red{\bb{C}_1}$ is a premulticategory  which have the same objects. 
			\item $J$ is an identity on objects premulticategory morphism, $J: \blue{\bb{C}_0}\to \red{\bb{C}_1}$ that preserves centrality. 
		\end{enumerate} 
The morphisms in $\blue{\bb{C}_0}$ are said to be pure and will be denoted in blue, while those in $\red{\bb{C}_1}$ are said to be effectful, and will be denoted in red.
	\end{definition}
	
	\begin{remark}
		Effectful multicategories were called Freyd multicategories in \cite{staton-levy}. However, we prefer the term effectful multicategory due to their relationship to effectful categories. The term effectul category was suggested in \cite{roman} to avoid a clash with the notion of a Freyd category. 
	\end{remark}

	Note that there may well be other central terms in $\bb{C}$ that are not values. For example, when dealing with probability all computations are central and in this case the premulticategory $\red{\bb{C}_1}$ also reduces to a multicategory.
	
	The language that these structures model are more fine-grained \cite{modelling-env}, in that are two types of deductions. The first type indicates that a program is pure, $\Gamma\vdash^v M:A$, while the second type indicates that a program is effectful $\Gamma \vdash^c M:A$. Further, every pure program can be made into an effectful one:
	\begin{prooftree}
		\AxiomC{$\Gamma\vdash^v M:B$}
		\UnaryInfC{$\Gamma\vdash^c JM: B$}
	\end{prooftree}
	Preservation of centrality is required to say that a pure program gets transformed into a \textit{trivially effectful} one.

	\begin{example}
		Every premulticategory can be thought of as an effectful multicategory by setting $\blue{\bb{C}_0}:= Z(\bb{C})$. Then, there is a centrality preserving ioo-morphism $\blue{Z(\bb{C})}\to \red{\bb{C}}$. 
	\end{example}

	Restricting to unary contexts gives us two categories, $\overline{\blue{\bb{C}_0}}$ and $\overline{\red{\bb{C}_1}}$ with an identity on objects functor $\overline{J}$ between them. Hence, instead of a functor into \category{Cat}, we get a functor into $[\to,\category{Set}]$-\category{Cat}, the category of $[\to,\category{Set}]$-enriched categories. The left adjoint to this functor is the mapping that sends a pair of categories with an ioo functor between them, $\blue{\mcal{C}_0}\to \red{\mcal{C}_1}$ to their discrete multicategories, with a premulticategory morphism between them: $\blue{J\mcal{C}_0}\to \red{J\mcal{C}_1}$.

	Adding product types into this language takes us  to a structure called an \textit{effectful category}.
	\begin{definition}
		An (strict) effectful category is a triple $(\blue{\mcal{C}_0}, \red{\mcal{C}_1}, J)$ where:
		\begin{enumerate}
			\item $\blue{\mcal{C}_0}$ is a (strict) monoidal category, $\red{\mcal{C}_1}$ is a (strict) premonoidal category, and both have the same objects
			\item $J$ is an identity on objects premonoidal functor, which preserves centrality. 
		\end{enumerate}
	\end{definition}

	Obviously $\blue{\mcal{C}_0}$ gives rise to a multicategory $\blue{U\mcal{C}_0}$, and $\red{\mcal{C}_1}$ gives rise to a premulicategory $\red{U\mcal{C}_1}$, and $J$ gives rise to an identity on objects premulticategory morphism, that preserves centrality (by virtue of Lemma \ref{cent-to-cent}). Hence, every effectful category gives rise to an effectful multicategory.
	
	To extend this mapping to a functor between respective categories, we need a notion of morphisms.

	\begin{definition}
		A morphism $(\blue{\bb{C}_0},\red{\bb{C}_1}, J_{\bb{C}}) \to (\blue{\bb{D}_0}, \red{\bb{D}_1}, J_{\bb{D}})$  of effectful multicategories is a pair of premulticategory morphisms, $\blue{f_0}: \blue{\bb{C}_0}\to \blue{\bb{D}_0} $ and $\red{f_{1}}: \red{\bb{C}_1}\to \red{\bb{D}_1}$ that make the following diagram commute:
	\[\begin{tikzcd}[cramped]
		{\blue{\bb{C}_0}} && {\red{\bb{C}_1}} \\
		\\
		{\blue{\bb{D}_0}} && {\red{\bb{D}_1}}
		\arrow["{J_1}", from=1-1, to=1-3]
		\arrow["{\blue{f_{0}}}"', from=1-1, to=3-1]
		\arrow["{J_2}"', from=3-1, to=3-3]
		\arrow["{\red{f_{1}}}", from=1-3, to=3-3]
	\end{tikzcd}\]
		The category of effectful multicategories and their morphisms will be called \category{EffMultiCat}.
	\end{definition}
	
	Morphisms of effectful categories can be similarly defined, and this gives rise to a category \category{EffCat}. The mappings above give rise to a functor, $U:\category{EffMultiCat}\to \category{EffCat}$. As with premulticategories, constructing a left adjoint $F$ does not seem to be straightforward.

	The notion of a transformation, $\eta: (\blue{f_0}, \red{f_1})\Rightarrow (\blue{g_0}, \red{g_1})$ between morphisms of effectful categories could have been defined as a pair $(\blue{\eta_0}, \red{\eta_1})$ where $\blue{\eta_0}: \blue{f_0}\Rightarrow \blue{g_0}$  and $\red{\eta_1}: \red{f_1} \Rightarrow\red{g_1} $ are transformations of premulticategory morphisms. However, we derive a definition in Chapter \ref{chapter6} which arises abstractly.
	
	\chapter{The Category of Arrows and Commutative Squares}\label{chapter5}

	\section{A Funny Tensor}
	
	We assume that $\mcal{V}$ is a cosmos, i.e, it is symmetric closed monoidal, complete and cocomplete \cite{street}. The archetypal example of a cosmos is the category $\category{Set}$ with its monoidal structure given by $\times$, the cartesian product. The category $\to$ is the two object category with exactly one arrow between the objects: 
	
	\[\begin{tikzcd}[cramped]
		0 && 1
		\arrow["\diamond", from=1-1, to=1-3]
	\end{tikzcd}\]
	
	Then, category $[\to, \mcal{V}]$ consists of functors from $\to$ to $[\to,\mcal{V}]$. 
		 Explicitly, an object in this category is an arrow  $a \equiv a_{\diamond}: a_0\to a_1$ in $\mcal{V}$. A morphism $f: a\to b$ is given by a pair of morphisms, $(f_i: a_i\to b_i)_{i =0,1}$ such that the following square commutes:
	\[\begin{tikzcd}[cramped,column sep=small,row sep=scriptsize]
		{a_0} && {a_1} \\
		\\
		{b_0} && {b_1}
		\arrow["{a_\diamond}", from=1-1, to=1-3]
		\arrow["{f_0}"', from=1-1, to=3-1]
		\arrow["{b_\diamond}"', from=3-1, to=3-3]
		\arrow["{f_1}", from=1-3, to=3-3]
	\end{tikzcd}\]
	This category inherits a `pointwise' tensor $\otimes$ from the base category $\mcal{V}$:
	\[
	(a_\diamond:a_0\to a_1)\otimes (b_\diamond:b_0\to b_1) : = (a_\diamond \otimes b_\diamond:a_0\otimes b_0 \to a_1\otimes b_1)
	\]
	The unit for this monoidal structure is the arrow $I\to I$ in $\mcal{V}$. The structure of this category allows us to define another tensor, which is reminiscent of the funny tensor on $\category{Cat}$ discussed in the previous chapter.

	The funny tensor on \category{Cat} encodes sesquiness, which is an important property for evaluating tuples $(a,b)$ in an effectful,  call-by-value language. We show that a similar funny tensor can be defined on the category $[\to,\mcal{V}]$ for a \textit{good} $\mcal{V}$. This provides us with a framework to talk about sesquiness for substitution, which from a significant part of Chapter \ref{chapter6}.

	Let $\mcal{I}_n$ be the following category: 
	
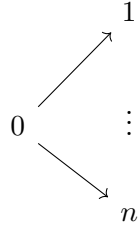
\[\begin{tikzcd}
		& 1 \\
		0 & \vdots \\
		& n
		\arrow[from=2-1, to=1-2]
		\arrow[from=2-1, to=3-2]
	\end{tikzcd}\]
	A colimit for a diagram, $\Gamma:\mcal{I}_n\to \mcal{C}$ is called a (wide) pushout. To give a natural transformation, $\alpha: \Gamma\Rightarrow\Gamma'$ is to give a family of morphisms, $\{\alpha_i: \Gamma(i)\to \Gamma'(i)\}_{0\leq i\leq n}$ such that $\alpha_i\circ \Gamma(f) =  \Gamma'(f)\circ \alpha_0$ for all $f: 0\to i$. So, we have the following lemma which will turn out to be the main tool for proving many of the results of this section. 
	
	\begin{lemma}\label{main}
		Let $\{a_i:A_0\to A_i\}_{1\leq i\leq n}$ and $\{b_i:B_0\to B_i\}_{1\leq i\leq n}$ be families of morphims in $\mcal{C}$. Let $(P,p)$ be a pushout for the diagram, $\{a_i: A_0\to A_i\}_{1\leq i\leq n}$, and let $(P',p')$ be \textit{any} co-cone for $\{b_i: B_0\to B_i\}_{1\leq i\leq n}$.

		When given a family of morphisms, $\{\alpha_i: A_i\to B_i\}_{0\leq i\leq n}$ such that $\alpha_i \circ a_i = b_i\circ  \alpha_0 $, there is a unique map, $q: P\to P'$ such that $q\circ p_i =   p'_i \circ \alpha_i$. 
		
		Further, when $P'$ is also a pushout, and when the $\alpha_i$ are isomorphisms, $q$ is also an isomorphism. 	
	\end{lemma}

	In the context of binary pushouts, this says  that if in the following diagram, the top square is a pushout square, and the bottom square is any commutative square, and the left and front faces commute, then, there is a unique map, $P\to P'$ that makes the right and back faces commute.
\begin{equation}\label{main2}
		\begin{tikzcd}
		& {A_1} && P \\
		{A_0} && {A_2} \\
		& {B_1} && {P'} \\
		{B_0} && {B_2}
		\arrow["{a_2}"{description, pos=0.7}, from=2-1, to=2-3]
		\arrow["{a_1}"{description}, from=2-1, to=1-2]
		\arrow["{p_1}"{description}, from=1-2, to=1-4]
		\arrow["{p_2}"{description}, from=2-3, to=1-4]
		\arrow["{b_1}"{description}, from=4-1, to=3-2]
		\arrow["{b_2}"{description}, from=4-1, to=4-3]
		\arrow["{p'_2}"{description}, from=4-3, to=3-4]
		\arrow["{p'_1}"{description, pos=0.3}, from=3-2, to=3-4]
		\arrow["{\alpha_0}"{description, pos=0.7}, from=2-1, to=4-1]
		\arrow["{\alpha_1}"{description, pos=0.7}, from=1-2, to=3-2]
		\arrow["{\alpha_2}"{description, pos=0.7}, from=2-3, to=4-3]
		\arrow["q"{description, pos=0.7}, dashed, from=1-4, to=3-4]
	\end{tikzcd}
\end{equation}

This simple observation will form the basis for many arguments about the funny tensor on $[\to,\mcal{V}]$.

	\begin{definition}
		Let $a$ and $b$ be objects of $[\to,\mcal{V}]$.
		Then define $a\fun b$ as the diagonal arrow in the following pushout diagram:
		\[\begin{tikzcd}[cramped,row sep=2.25em]
			{ a_0\otimes b_0} && {a_0\otimes b_1} \\
			\\
			{a_1\otimes b_0} && {(a\fun b)_1}
			\arrow["{1\otimes b_\diamond}"{description}, from=1-1, to=1-3]
			\arrow["{a_\diamond\otimes 1}"{description}, from=1-1, to=3-1]
			\arrow["{\iota'_{a,b}}"{description}, from=3-1, to=3-3]
			\arrow["{\iota_{a,b}}"{description}, from=1-3, to=3-3]
			\arrow["{(a\fun b)_\diamond}"{description}, dashed, from=1-1, to=3-3]
		\end{tikzcd}\]

		Given morphisms, $f: a\to c$ and $g: b\to d$, it is easily seen that the top, bottom, left, and back faces of the cube below commute. As a consequence, by \ref{main2},	there is a unique morphism, $(f\fun g)_1:  (a\fun b)_1 \to (c\fun d)_1$ that makes the front and right faces commute: 
		\[\begin{tikzcd}[cramped]
			& {a_0\otimes b_1} && {(a\fun b)_1} \\
			{a_0\otimes b_0} && {a_1\otimes b_0} \\
			& {c_0\otimes d_1} && {(c\fun d)_1} \\
			{c_0\otimes d_0} && {c_1\otimes d_0}
			\arrow["{1\otimes b_\diamond}"{description}, from=2-1, to=1-2]
			\arrow["{a_\diamond\otimes 1}"{description, pos=0.7}, from=2-1, to=2-3]
			\arrow["{\iota_{a,b}}"{description}, from=1-2, to=1-4]
			\arrow["{\iota'_{a,b}}"{description}, from=2-3, to=1-4]
			\arrow[from=4-1, to=3-2]
			\arrow[from=4-1, to=4-3]
			\arrow[from=4-3, to=3-4]
			\arrow["{f_0\otimes g_0}"{description, pos=0.8}, from=2-1, to=4-1]
			\arrow["{f_0\otimes g_1}"{description, pos=0.8}, from=1-2, to=3-2]
			\arrow["{f_1\otimes g_0}"{description, pos=0.8}, from=2-3, to=4-3]
			\arrow["{(f\fun g)_1}"{description, pos=0.8}, dashed, from=1-4, to=3-4]
			\arrow[from=3-2, to=3-4]
		\end{tikzcd}\]
		It is then readily verified, using the universal property of pushouts, that this information defines a bifunctor,
$		\fun  : [\to,\mcal{V}]\times [\to,\mcal{V}]\to [\to,\mcal{V}]$
		
	\end{definition}
	
	The rest of this section is devoted to showing that this bifunctor defines a symmetric monoidal structure on $[\to,\mcal{V}]$ with unit $(1_I: I\to I)$ where $1_I$ is the identity morphism.

	\begin{lemma}\label{sym}
		There are natural isomorphisms, 
		\begin{enumerate}
			\item $(\lambda^0_a,\lambda^1_a): (I\otimes a_0 \to I\fun a_0) \to (a_0\to a_1)$
			\item $(\rho^0_a,\rho^1_a): (a_0\otimes I\to a_1\fun I) \to (a_0\to a_1) $
			\item $(s^0_{a,b}, s^1_{a,b}): (a_0\otimes b_0\to a_1\fun b_1) \to (b_0\otimes a_0\to b_1\fun a_1)$
		\end{enumerate}
	\end{lemma}
	\begin{proof}
		For the unitors, consider the cubes:
		\[\begin{tikzcd}[cramped]
			& {a_0\otimes I} && {a_1\fun I} && {I\otimes a_1} && {1\fun a_1} \\
			{a_0\otimes I} && {a_1\otimes I} && {I\otimes a_0} && {I\otimes a_0} \\
			& {a_0} && {a_1} && {a_1} && {a_1} \\
			{a_0} && {a_1} && {a_0} && {a_0}
			\arrow[from=2-1, to=1-2]
			\arrow[from=2-1, to=2-3]
			\arrow[from=1-2, to=1-4]
			\arrow[from=2-3, to=1-4]
			\arrow[from=4-1, to=3-2]
			\arrow[from=4-1, to=4-3]
			\arrow[from=3-2, to=3-4]
			\arrow[from=4-3, to=3-4]
			\arrow["{\rho_{a_0} = \rho^0_a}"{description, pos=0.6}, from=2-1, to=4-1]
			\arrow["{\rho_{a_0}}"{description, pos=0.7}, from=1-2, to=3-2]
			\arrow["{\rho_{a_1}}"{description, pos=0.6}, from=2-3, to=4-3]
			\arrow["{\rho^1_{a}}"{description}, dashed, from=1-4, to=3-4]
			\arrow[from=2-5, to=1-6]
			\arrow[from=2-5, to=2-7]
			\arrow[from=1-6, to=1-8]
			\arrow[from=2-7, to=1-8]
			\arrow["{\lambda_{a_0} = \lambda^0_a}"{description, pos=0.6}, from=2-5, to=4-5]
			\arrow[from=4-5, to=3-6]
			\arrow[from=4-5, to=4-7]
			\arrow[from=4-7, to=3-8]
			\arrow[from=3-6, to=3-8]
			\arrow["{\lambda^1_a}"{description}, dashed, from=1-8, to=3-8]
			\arrow["{\lambda_{a_0}}"{description, pos=0.6}, from=2-7, to=4-7]
			\arrow["{\lambda_{a_1}}"{description, pos=0.7}, from=1-6, to=3-6]
		\end{tikzcd}\]
		
		The squares at the bottom are obviously pushout squares, so an appeal to  \ref{main2} reveals the existence of $\rho^1_a$. Since  $\rho_{a_0}, \rho_{a_1}$ are all isos, $\rho^1_a$ also must be one. Thus, we obtain a mapping, $(\rho^0_a, \rho^1_a): (a_0\otimes I \to a_1\fun I) \to (a_0\to a_1)$. We show that this map is natural in $(a_0\to a_1)$. 
		
		Let $(f_0,f_1): (a:a_0\to a_1)\to (b:b_0\to b_1)$  be any map, and consider the cubes:	
		\[\begin{tikzcd}[cramped,column sep=tiny]
			& {a_0\otimes I} && {a_1\fun I} && {a_0\otimes I} && {a_1\fun I} \\
			{a_0\otimes I} && {a_1\otimes I} && {a_0\otimes I} && {a_1\otimes I} \\
			& {b_0\otimes I} && {b_1\fun I} && {a_0} && {a_1} \\
			{b_0\otimes I} && {b_1\otimes I} && {a_0} && {a_1} \\
			& {b_0} && {b_1} && {b_0} && {b_1} \\
			{b_0} && {b_1} && {b_0} && {b_0}
			\arrow[from=2-1, to=2-3]
			\arrow[from=2-1, to=1-2]
			\arrow[from=1-2, to=1-4]
			\arrow[from=2-3, to=1-4]
			\arrow[from=4-1, to=3-2]
			\arrow[from=3-2, to=3-4]
			\arrow[from=4-1, to=4-3]
			\arrow[from=4-3, to=3-4]
			\arrow[from=5-2, to=5-4]
			\arrow[from=6-1, to=5-2]
			\arrow[from=6-1, to=6-3]
			\arrow[from=6-3, to=5-4]
			\arrow["{f_1\fun 1}"{description, pos=0.7}, dashed, from=1-4, to=3-4]
			\arrow["{\rho^1_b}"{description, pos=0.7}, dashed, from=3-4, to=5-4]
			\arrow["{f_0\otimes 1}"{description, pos=0.7}, from=2-1, to=4-1]
			\arrow["{f_0\otimes 1}"{description, pos=0.7}, from=1-2, to=3-2]
			\arrow["{f_1\otimes 1}"{description, pos=0.7}, from=2-3, to=4-3]
			\arrow["{\rho_{b_0}}"{description, pos=0.7}, from=3-2, to=5-2]
			\arrow["{\rho_{b}^0 = \rho_{b_0}}"{description, pos=0.7}, from=4-1, to=6-1]
			\arrow["{\rho_{b_1}}"{description, pos=0.7}, from=4-3, to=6-3]
			\arrow[from=2-5, to=1-6]
			\arrow[from=2-5, to=2-7]
			\arrow[from=2-7, to=1-8]
			\arrow[from=1-6, to=1-8]
			\arrow["{\rho^1_a = \rho_{a_0}}"{description, pos=0.6}, from=2-5, to=4-5]
			\arrow[from=4-5, to=3-6]
			\arrow[from=3-6, to=3-8]
			\arrow[from=4-5, to=4-7]
			\arrow[from=4-7, to=3-8]
			\arrow[from=6-5, to=5-6]
			\arrow[from=5-6, to=5-8]
			\arrow[from=6-5, to=6-7]
			\arrow[from=6-7, to=5-8]
			\arrow["{f_0}"{description}, from=4-5, to=6-5]
			\arrow["{\rho_{a_0}}"{description, pos=0.6}, from=1-6, to=3-6]
			\arrow["{\rho_{a_0}}"{description, pos=0.6}, from=2-7, to=4-7]
			\arrow["{\rho^1_a}"{description}, dashed, from=1-8, to=3-8]
			\arrow["{f_0}"{description, pos=0.6}, from=3-6, to=5-6]
			\arrow["{f_1}"{description}, from=4-7, to=6-7]
			\arrow["{f_1}"{description}, from=3-8, to=5-8]
		\end{tikzcd}\]
		
		However, note that each of the `legs' $a_i\otimes I\to b_i$ are equal by naturality of $\otimes$ in $\mcal{V}$, i.e, $\rho_{b_i} \circ (f_i\otimes1) = f_i\circ \rho_{a_i}$ for $i \in \{1,2\}$. The following cube is then obtained:
\[\begin{tikzcd}[cramped]
	& {a_0\otimes I} && {a_1\fun I} \\
	{a_0\otimes I} && {a_1\otimes I} \\
	\\
	& {b_0} && {b_1} \\
	{b_0} && {b_1}
	\arrow[from=2-1, to=1-2]
	\arrow[from=1-2, to=1-4]
	\arrow[from=2-1, to=2-3]
	\arrow[from=2-3, to=1-4]
	\arrow[from=5-1, to=4-2]
	\arrow[from=4-2, to=4-4]
	\arrow[from=5-1, to=5-3]
	\arrow[from=5-3, to=4-4]
	\arrow["{\substack{ \rho_{b_0} \circ (f_0 \otimes 1)\\ = \\f_0\circ \rho_{a_0} }}"{description, pos=0.7}, from=2-1, to=5-1]
	\arrow["{\substack{\rho_{b_0}\circ (f_0\otimes 1)\\ = \\f_0\circ \rho_{a_0} }}"{description, pos=0.7}, from=1-2, to=4-2]
	\arrow["{\substack{ \rho_{b_1}\circ (f_1\otimes 1)\\ = \\ f_1\circ \rho_{a_0}}}"{description, pos=0.7}, from=2-3, to=5-3]
	\arrow["q"{description, pos=0.7}, dashed, from=1-4, to=4-4]
\end{tikzcd}\]

		where $q$ is the unique morphism that makes the right and back sides of the cube commute. However, both $\rho_{b}^1\circ (f_1\fun 1)$ and $f_1\circ \rho^1_a$ satisfy this property, and hence, we conclude that they are both equal. This shows that the required naturality square commutes in $[\to,\mcal{V}]$. A completely analogous chain of reasoning shows that there is a natural transformation, $(\lambda^0_a,\lambda^1_a): (I\otimes a_0 \to I\fun a_1) \to (a_0\to a_1)$.

		For the braiding, consider the cube
		\[\begin{tikzcd}[cramped]
			& {a_0\otimes b_1} && {a_1\fun b_1} \\
			{a_0\otimes b_0} && {a_1\otimes b_0} \\
			& {b_1\otimes a_0} && {b_1\fun a_1} \\
			{b_0\otimes a_0} && {b_0\otimes a_1}
			\arrow[from=2-1, to=1-2]
			\arrow[from=2-1, to=2-3]
			\arrow[from=1-2, to=1-4]
			\arrow[from=2-3, to=1-4]
			\arrow["{s^0_{a,b} = s_{a_0,b_0}}"{description, pos=0.7}, from=2-1, to=4-1]
			\arrow[from=4-1, to=3-2]
			\arrow[from=3-2, to=3-4]
			\arrow[from=4-1, to=4-3]
			\arrow[from=4-3, to=3-4]
			\arrow["{s_{a_1,b_0}}"{description, pos=0.7}, from=2-3, to=4-3]
			\arrow["{s^1_{a,b}}"{description}, dashed, from=1-4, to=3-4]
			\arrow["{s_{a_0,b_1}}"{description, pos=0.7}, from=1-2, to=3-2]
		\end{tikzcd}\]
		The map, $s^1_{a,b}$ exists by Lemma \ref{main}, and makes the cube commute. Further, it is an iso since each of the other $s_{a_i,b_j}$ are all isos. Showing naturality can be done with an argument analogous to the previous one. 
	\end{proof}

	Before describing the associator, we show that there is a nice characterization of $n$-ary applications of the bifunctor $\fun$ as wide pushouts, a fact that will be useful while reasoning about $3$-ary and $4$-ary tensors which appear in the coherence axioms.

	\begin{lemma}\label{stillpush}
		Consider the family of morphisms, $\{A_0\to A_i\}_{1\leq i\leq n+1}$. Assume that $(P, \{p_i: A_i\to P\}_{1\leq i\leq n})$ is a wide pushout of the family $\{A_0\to A_i\}_{1\leq i\leq n}$, and that $Q$ is a  pushout of the diagram, 
		
		\[\begin{tikzcd}[cramped]
			& P && Q \\
			{A_0} && {A_{n+1}}
			\arrow["p"{description}, from=2-1, to=1-2]
			\arrow["{a_{n+1}}"{description}, from=2-1, to=2-3]
			\arrow["{q_1}"{description}, from=1-2, to=1-4]
			\arrow["{q_2}"{description}, from=2-3, to=1-4]
		\end{tikzcd}\]
		where $p : = p_1\circ a_1= \ldots = p_n\circ a_n$. 
		Then, $Q$ is a wide pushout for the family $\{A_0\to A_i\}_{1\leq i\leq n+1}$.
	\end{lemma}
	
	\begin{proof}
		Define $q'_i: A_i\to Q$ as
		\[
		q'_i : = \begin{cases}
			q_2, &\text{ if $i = n+1$}\\
			q_1\circ 	p_i, &\text{ otherwise}	
		\end{cases}
		\]
		Then, clearly, $(Q,q')$ is a co-cone for the diagram induced by $\{A_0\to A_i\}_{1\leq i\leq n+1}$.  Assume $(X,\{x_i: A_i\to X\}_{1\leq i\leq n+1})$ is another co-cone. Then, it is a co-cone for the diagram, $\{A_0\to A_i\}_{1\leq i\leq n}$, and hence, there is a unique morphism, $x: P\to X$ such that  $x_i = x\circ p_i$.

		Now, since there are morphisms, $P\to X$ and $A_{n+1}\to X$ that make the square commute below, there is a unique, $x': Q\to X$ such that $x_{n+1} = x'\circ q_2$ and $x = x'\circ q_1$.
		
		\[\begin{tikzcd}
			&&& X \\
			P && Q \\
			{A_0} && {A_{n+1}}
			\arrow["p"{description}, from=3-1, to=2-1]
			\arrow["{a_{n+1}}"{description}, from=3-1, to=3-3]
			\arrow["{q_1}"{description}, from=2-1, to=2-3]
			\arrow["{q_2}"{description}, from=3-3, to=2-3]
			\arrow["x"{description}, curve={height=-18pt}, from=2-1, to=1-4]
			\arrow["{x_{n+1}}"{description}, curve={height=18pt}, from=3-3, to=1-4]
			\arrow["{x'}"{description}, dashed, from=2-3, to=1-4]
		\end{tikzcd}\]
		
		So, there is a morphism, $x': Q\to X$ such that $x_i = x'\circ q_i'$. This morphism must be unique, since if there was another such morphism, $x'': Q\to X$ such that $x_i = x''\circ q_i'$, then, $x_i  = (x''\circ  q_1)\circ p_i $, which implies that $x''\circ q_1 = x$ by uniqueness. Then, we have that $x''\circ q_1 = x$ and $x''\circ q_2 = x_{n+1}$, which means that $x'' = x'$ by uniqueness again. 
	\end{proof}

	In general monoidal categories, there are many ways to define $n$-ary products. We adopt the convention that these are bracketed to the left.
	\begin{definition}
		Let $\mcal{C}$ be a category with a bifunctor, $\otimes: \mcal{C}\times \mcal{C}\to \mcal{C}$. For a family of objects, $\{a^i\}_{1\leq i\leq n}$ in $\mcal{C}$, define the $n$-ary product, $\bigotimes_{i=1}^n a^i$ recursively as follows:
		\begin{align*}
			\bigotimes_{i=1}^1 a^i &:= a^1\\
			\bigotimes_{i=1}^{n+1} a^i &:= (\bigotimes_{i=1}^n a^i)\otimes a^{n+1}
		\end{align*}
	\end{definition}

	Then, an immediate corollary of \ref{stillpush} is the following. 
	\begin{corollary}[Pushout characterization of the $n$-ary funny tensor]\label{pushout-char}
		Let $\{a^i\}_{1\leq i\leq n}$ be a family of objects in $[\to,
		\mcal{V}]$. Then, there are morphisms, $k_1,\ldots, k_n$, that make $({{\fun}}_{i=1}^n a^i)_1$ a pushout for the following diagram: 
		\[\begin{tikzcd}[cramped]
			& {\bigotimes_{i=1}^n a_{\delta_{i,1}}^i} \\
			{\bigotimes_{i=1}^n a^i_0} & \vdots & {({{\fun}}_{i=1}^n a^i)_1} \\
			& {\bigotimes_{i=1}^n a_{\delta_{i,n}}^i}
			\arrow[from=2-1, to=1-2]
			\arrow[from=2-1, to=3-2]
			\arrow["{k_1}"{description}, from=1-2, to=2-3]
			\arrow["{k_n}"{description}, from=3-2, to=2-3]
		\end{tikzcd}\]
		where the map, $\bigotimes_{i=1}^n a^i_0 \to \bigotimes_{i=1}^n a^i_{\delta_{i,j}}$ is the map that has identities tensored with the map $a^j_\diamond$ in the $j$th position (like $1\otimes \ldots 1\otimes a^j_\diamond \otimes 1\ldots \otimes 1 $). Further, the arrow, $\bigotimes_{i=1}^n a_0^i \to (\fun_{i=1}^n a^i)_1$ is the tensor $\fun_{i=1}^n a^i$ in $[\to,\mcal{V}]$.
	\end{corollary}
	\begin{proof}
		This is obvious for the case $n = 2$: $k_1 = \iota_{a,b}$ and $k_2 = \iota'_{a,b}$, which are the canonical morphisms into the pushout.

		Assume for induction that this is true for some $n$. Then, the tensor, $(\fun_{i=1}^n a^i)$ is the arrow  $(\bigotimes_{i=1}^n a^i_0 \to (\fun_{i=1}^n a^i)_1)$. Then the tensor $(\fun_{i=1}^{n+1} a^i)$ is calculated as 
		\[
		\Big(\bigotimes_{i=1}^n a^i_0 \to (\fun_{i=1}^n a^i)_1\Big) \fun (a_0^{n+1}\to a_1^{n+1})
		\]
		which is the pushout diagram:
		\[\begin{tikzcd}[cramped]
			& {(\fun_{i=1}^n a^i)_1 \otimes a_0^{n+1}} \\
			{(\bigotimes_{i=1}^n  a^i_0 )\otimes  a_{0}^{n+1}} & {} && {(\fun^{n+1}_{i=1} a^i)_1} \\
			& {(\bigotimes_{i=1}^n  a^i_0 )\otimes a_1^{n+1}}
			\arrow[from=2-1, to=1-2]
			\arrow[from=2-1, to=3-2]
			\arrow["{\iota'_{\fun_{i=1}^n a^i,  a^{n+1}}}"{description}, from=3-2, to=2-4]
			\arrow["{\iota_{\fun_{i=1}^n a^i,  a^{n+1}}}"{description}, from=1-2, to=2-4]
		\end{tikzcd}\]
		However, the object $(\fun_{i=1}^n a^i)_1\otimes a_{0}^{n+1}$ is the pushout of the following diagram (since $\mcal{V}$ is symmetric monoidal \textit{closed}, tensoring with a fixed object preserves colimits): 
		\[\begin{tikzcd}[cramped]
			& {(\bigotimes_{i=1}^n a^i_{\delta_{i,1}} ) \otimes a_0^{n+1}} \\
			{\bigotimes_{i=1}^n a_0^i \otimes a_0^{n+1}} & \vdots & {(\fun_{i=1}^n a^i)_1 \otimes a_0^{n+1}} \\
			& {(\bigotimes_{i=1}^n a^i_{\delta_{n,1}} ) \otimes a_0^{n+1}}
			\arrow[from=2-1, to=1-2]
			\arrow[from=2-1, to=3-2]
			\arrow["{k_1\otimes 1}"{description}, from=1-2, to=2-3]
			\arrow["{k_n\otimes 1}"{description}, from=3-2, to=2-3]
		\end{tikzcd}\]	Thus, by Lemma \ref{stillpush}, the following is a (wide) pushout diagram:
		\[\begin{tikzcd}[cramped]
			& {(\bigotimes_{i=1}^{n+1} a^i_{\delta_{i,1}})} \\
			{\bigotimes_{i=1}^{n+1} a_0^i} & \vdots && {(\fun_{i=1}^{n+1} a^i)_1} \\
			& {(\bigotimes_{i=1}^{n+1} a^i_{\delta_{i,n}})} \\
			& {(\bigotimes_{i=1}^{n+1} a^i_{\delta_{i,{n+1}}})}
			\arrow[from=2-1, to=1-2]
			\arrow[from=2-1, to=3-2]
			\arrow[from=2-1, to=4-2]
			\arrow["{\iota'}"{description}, from=4-2, to=2-4]
			\arrow["{\iota\circ (k_n\otimes 1)}"{description}, from=3-2, to=2-4]
			\arrow["{\iota \circ (k_1\otimes 1)}"{description}, from=1-2, to=2-4]
		\end{tikzcd}\]
		and the arrow, $\bigotimes_{i=1}^{n+1}a_0^i \to (\fun_{i=1}^{n+1}a^i)_1$ is required tensor. 
	\end{proof}
	
	This observation allows us to characterize a map, $(\fun_{i=1}^n a^i)_1 \to X$ as consisting of a family of maps, $\{\bigotimes_{i=1}^n a^i_{\delta_i,j}\to X\}_{1\leq j\leq n}$ such that the required diagram commutes.

	\begin{lemma}
		There is a natural isomorphism,
		\[
		(\alpha^0_{a,b,c} , \alpha^1_{a,b,c}): (a\fun b)\fun c\to a \fun (b\fun c)
		\]
	\end{lemma}
	\begin{proof}
		Consider the following diagram: 
		\[\begin{tikzcd}[cramped,column sep=scriptsize]
			& {(a_1\otimes b_0)\otimes c_0} && {((a\fun b)\fun c)_1} \\
			&& {(a_0\otimes b_1)\otimes c_0} \\
			{(a_0\otimes b_0)\otimes c_0} && {(a_0\otimes b_0)\otimes c_1} \\
			& {a_1\otimes (b_0\otimes c_0)} && {(a\fun (b\fun c))_1} \\
			&& {a_0\otimes (b_1\otimes c_0)} \\
			{a_0\otimes (b_0\otimes c_0)} && {a_0\otimes (b_0\otimes c_1)}
			\arrow[from=3-1, to=3-3]
			\arrow[from=3-1, to=2-3]
			\arrow[from=2-3, to=1-4]
			\arrow[from=3-1, to=1-2]
			\arrow[from=1-2, to=1-4]
			\arrow[from=3-3, to=1-4]
			\arrow[from=6-1, to=4-2]
			\arrow[from=6-1, to=6-3]
			\arrow[from=4-2, to=4-4]
			\arrow[from=6-3, to=4-4]
			\arrow[from=6-1, to=5-3]
			\arrow[from=5-3, to=4-4]
			\arrow["{\alpha^1_{a,b,c} = \alpha_{a_0,b_0,c_0}}"{description, pos=0.3}, from=3-1, to=6-1]
			\arrow["{\alpha_{a_1,b_0,c_0}}"{description, pos=0.3}, from=1-2, to=4-2]
			\arrow["{\alpha_{a_0,b_1,c_0}}"{description, pos=0.8}, curve={height=30pt}, from=2-3, to=5-3]
			\arrow["{\alpha_{a_0,b_0,c_1}}"{description, pos=0.3}, from=3-3, to=6-3]
			\arrow["{\alpha^1_{a,b,c}}"{description, pos=0.3}, dashed, from=1-4, to=4-4]
		\end{tikzcd}\]
		
		The plane on top is a wide pushout diagram by Lemma \ref{stillpush}. An analogous argument as the one used in Lemma \ref{stillpush} reveals that the plane at the bottom is also a wide pushout diagram. Thus, the map, $\alpha^1_{a,b,c} $ exists by Lemma \ref{main}, and is an isomorphism since 
		each of $\alpha_{a_i,b_j,c_k}$ are. Naturality also follows by routine diagram chasing of the style used in the proof of Lemma \ref{sym}.
	\end{proof}
	
	Proofs of all the coherence laws can be obtained in a similar fashion - by appealing to the respective coherence law in $ \mcal{V}$, and then applying Lemma \ref{main}. For example, here is a proof that the braiding described in $\ref{sym}$ satisfies the inverse law: consider the cube
	\[\begin{tikzcd}[cramped]
		& {a_1\otimes b_0} && {(a\fun b)_1} \\
		{a_0\otimes b_0} && {a_0\otimes b_1} \\
		& {b_0\otimes a_1} && {(b\fun a)_1} \\
		{b_0\otimes a_0} && {b_1\otimes a_0} \\
		& {a_1\otimes b_0} && {(a\fun b)_1} \\
		{a_0\otimes b_0} && {a_0\otimes b_1}
		\arrow[from=2-1, to=2-3]
		\arrow[from=2-1, to=1-2]
		\arrow[from=1-2, to=1-4]
		\arrow[from=2-3, to=1-4]
		\arrow["{s^1_{a,b} = s_{a_0,b_0}}"{description, pos=0.6}, from=2-1, to=4-1]
		\arrow[from=4-1, to=3-2]
		\arrow[from=4-1, to=4-3]
		\arrow[from=3-2, to=3-4]
		\arrow[from=4-3, to=3-4]
		\arrow["{s_{a_0,b_1}}"{description, pos=0.6}, from=2-3, to=4-3]
		\arrow["{s_{a_1,b_0}}"{description, pos=0.6}, from=1-2, to=3-2]
		\arrow["{s^1_{a,b}}"{description}, dashed, from=1-4, to=3-4]
		\arrow[from=6-1, to=5-2]
		\arrow[from=6-1, to=6-3]
		\arrow[from=5-2, to=5-4]
		\arrow[from=6-3, to=5-4]
		\arrow["{s^1_{b,a} = s_{b_0,a_0}}"{description, pos=0.6}, from=4-1, to=6-1]
		\arrow["{s_{b_1,a_0}}"{description, pos=0.6}, from=4-3, to=6-3]
		\arrow["{s_{b_0,a_1}}"{description, pos=0.6}, from=3-2, to=5-2]
		\arrow["{s^1_{b,a}}"{description}, dashed, from=3-4, to=5-4]
	\end{tikzcd}\]
	However since $s_{b_j, a_i}\circ s_{a_i,b_j} = 1_{a_i,b_j}$, the cube `compresses' to give: 
	\[\begin{tikzcd}[cramped]
		& {a_0\otimes b_1} && {(a\fun b)_1} \\
		{a_0\otimes b_0} && {a_1\otimes b_0} \\
		& {a_0\otimes b_1} && {(a\fun b)_1} \\
		{a_0\otimes b_0} && {a_1\otimes b_0}
		\arrow[from=2-1, to=1-2]
		\arrow[from=2-1, to=2-3]
		\arrow[from=1-2, to=1-4]
		\arrow[from=2-3, to=1-4]
		\arrow[from=4-1, to=3-2]
		\arrow[from=4-1, to=4-3]
		\arrow[from=3-2, to=3-4]
		\arrow[from=4-3, to=3-4]
		\arrow["{1_{a_0,b_0}}"{description, pos=0.6}, from=2-1, to=4-1]
		\arrow["{1_{a_0,b_1}}"{description, pos=0.6}, from=1-2, to=3-2]
		\arrow["{1_{a_1,b_0}}"{description, pos=0.6}, from=2-3, to=4-3]
		\arrow["{s^1_{b,a}\circ s^1_{a,b}}"{description, pos=0.6}, dashed, from=1-4, to=3-4]
	\end{tikzcd}\]
	The map $s^1_{b,a}\circ s^1_{a,b}$ is the unique map that makes the cube commute. However, $1_{(a\fun b)_1}$ also makes the cube commute. Hence, these two morphisms must be equal. Thus, the inverse law is established.  
	
	\begin{proposition}
		The bifunctor, $\fun: [\to,\mcal{V}]\times [\to,\mcal{V}] \to [\to,\mcal{V}]$ defines a symmetric monoidal structure on $[\to,\mcal{V}]$ with unit, $(I\to I)$.
	\end{proposition}

	\section{Some properties of the Funny Tensor}

	This section explores some straightforward properties of the funny tensor. Since $\mcal{V}$ is a cosmos, it is symmetric monoidal closed. As a consequence, for any $a \in \mcal{C}$, the functor, $-\otimes b: \mcal{V}\to \mcal{V}$ is a left adjoint, and hence preserves colimits. Similarly, the functor $- \fun b: [\to,\mcal{V}]\to [\to,\mcal{V}]$ also preserves colimits.

	\begin{lemma}
		Let $b \in [\to,\mcal{V}]$ be any object. Then, the functor $-\fun b$ is cocontinuous. 
	\end{lemma}
	\begin{proof}
		We show that co-equalizers and (arbitrary) coproducts are preserved. First, we note that $[\to,\mcal{V}]$ is cocomplete, with products and coequalizers computed pointwise. That is, if $\{a^i := a_\diamond^i: a^i_0\to a_1^i\}_{i\in I}$ is an $I$-indexed collection  in $[\to,\mcal{V}]$, the coproduct can be computed in $\mcal{V}$ as the arrow 
		\[\sum_{i\in I} a^i = \sum_{i\in I} a^i_0    \xrightarrow{\sum_{i\in I} a^i_\diamond} \sum_{i\in I} a^i_1 \]
		
		Fix some $b_\diamond: b_0\to b_1$. Consider the following diagram.
		\[\begin{tikzcd}[cramped]
			& {(\sum_{i\in I} a^i_0)\otimes b_1} && {(\sum_{i\in I} a^i)\fun b} \\
			\\
			{(\sum_{i\in I} a^i_0)\otimes b_0} && {(\sum_{i\in I} a^i_1)\otimes b_0} \\
			& {\sum_{i\in I} (a^i_0\otimes b_1)} && {\sum_{i \in I} (a^i\fun b)} \\
			\\
			{\sum_{i\in I} (a^i_0\otimes b_0)} && {\sum_{i\in I} (a^i_1\otimes b_0)}
			\arrow[from=3-1, to=3-3]
			\arrow["{j_1}"{description}, from=1-2, to=1-4]
			\arrow["{j_2}"{description}, from=3-3, to=1-4]
			\arrow["\cong"{description}, from=3-1, to=6-1]
			\arrow["{\sum_{i=1}^n (1\otimes b_\diamond)}"{description}, from=6-1, to=4-2]
			\arrow["{\sum_{i=1}^n (a^i_\diamond\otimes 1)}"', from=6-1, to=6-3]
			\arrow["{i_1}"{description, pos=0.6}, from=4-2, to=4-4]
			\arrow["{i_2}"{description, pos=0.6}, from=6-3, to=4-4]
			\arrow["\cong"{description}, from=3-3, to=6-3]
			\arrow["\cong"{description, pos=0.8}, from=1-2, to=4-2]
			\arrow["\cong"{description}, dashed, from=1-4, to=4-4]
			\arrow["{1\otimes b_\diamond}"{description}, from=3-1, to=1-2]
		\end{tikzcd}\]
		
		The top square clearly computes $(\sum_{i\in I} a^i) \fun b$ and  the bottom square computes $\sum_{i\in I} (a^i \otimes b)$. Since $\otimes$ preserves colimits in $\mcal{V}$, the maps, $(\sum_{i\in I} a^i_j)\otimes b_k \to \sum_{i\in I} (a^i_j\otimes b_k)$ are isomorphisms for $j,k\in \{0,1\}$. Finally, the square at the bottom is a pushout square since colimits commute past each other. Hence, the induced map, $(\sum_{i\in I} a^i)\fun b) \to \sum_{i\in I} (a^i\fun b)$ is an isomorphism.

		Similarly, given a diagram, $(f_0,f_1), (g_0,g_1): (a_\diamond: a_0\to a_1) \to (b_\diamond: b_0\to b_1)$  in $[\to,\mcal{V}]$, the coequalizer is computed in $\mcal{V}$ as:
		\[\begin{tikzcd}[cramped]
			{a_0} && {b_0} && {c_0} \\
			\\
			{a_1} && {b_1} && {c_1}
			\arrow["{a_\diamond}"', from=1-1, to=3-1]
			\arrow["{f_1}", shift left=2, from=3-1, to=3-3]
			\arrow["{f_0}", shift left=2, from=1-1, to=1-3]
			\arrow["{g_0}"', shift right=2, from=1-1, to=1-3]
			\arrow["{g_1}"', shift right=2, from=3-1, to=3-3]
			\arrow["{b_0}", from=1-3, to=3-3]
			\arrow["{h_0}", from=1-3, to=1-5]
			\arrow["{h_1}"', from=3-3, to=3-5]
			\arrow["{c_0}", dashed, from=1-5, to=3-5]
		\end{tikzcd}\]
		Here, $c_1$ is a universal cocone for $f_0,g_0: a_0\to b_0$ with $b_0\to c_1$ given by $h_1\circ b_0$, and so $c_\diamond: c_0\to c_1$ is the unique map that makes the square on the right commute. Let $d \in [\to,\mcal{V}]$ be given. Then, there are coequalizer diagrams,
		\[\begin{tikzcd}[cramped]
			{a_i\otimes d_j} && {b_i\otimes d_j} && {e^{i,j}} && {c_i\otimes d_j}
			\arrow["{f_i\otimes 1}", shift left=2, from=1-1, to=1-3]
			\arrow["{y^{i,j}}", from=1-3, to=1-5]
			\arrow["{g_i\otimes 1}"', shift right=2, from=1-1, to=1-3]
			\arrow["\cong"{description}, dashed, from=1-5, to=1-7]
			\arrow["{h_i\otimes 1}", curve={height=-24pt}, from=1-3, to=1-7]
		\end{tikzcd}\]
		in $\mcal{V}$ for each $i,j\in\{0,1\}$. Here, $e^{i,j}$ is the chosen coequalizer of $f_i\otimes 1$ and $g_i\otimes 1$, and the map $e^{i,j}\to c_i\otimes d_j$ is an isomorphism since $\otimes - d_j$ preserves colimits in $\mcal{V}$. Now, consider the following cube:
		
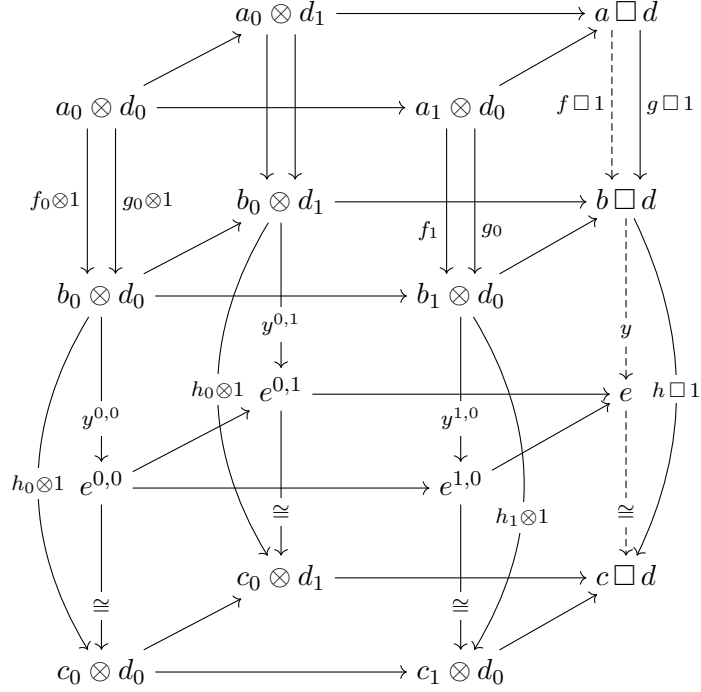
\[\begin{tikzcd}[cramped]
			& {a_0\otimes d_1} && {a\fun d} \\
			{a_0\otimes d_0} && {a_1\otimes d_0} \\
			& {b_0\otimes d_1} && {b\fun d} \\
			{b_0\otimes d_0} && {b_1\otimes d_0} \\
			& {e^{0,1}} && e \\
			{e^{0,0}} && {e^{1,0}} \\
			& {c_0\otimes d_1} && {c\fun d} \\
			{c_0\otimes d_0} && {c_1\otimes d_0}
			\arrow[from=2-1, to=1-2]
			\arrow[from=1-2, to=1-4]
			\arrow[from=2-3, to=1-4]
			\arrow[from=2-1, to=2-3]
			\arrow["{g_0\otimes 1}", shift left=2, from=2-1, to=4-1]
			\arrow[from=4-1, to=3-2]
			\arrow[from=4-1, to=4-3]
			\arrow[from=3-2, to=3-4]
			\arrow[from=4-3, to=3-4]
			\arrow[from=6-1, to=5-2]
			\arrow[from=6-1, to=6-3]
			\arrow[from=5-2, to=5-4]
			\arrow[from=6-3, to=5-4]
			\arrow[from=8-1, to=7-2]
			\arrow[from=8-1, to=8-3]
			\arrow[from=7-2, to=7-4]
			\arrow["\cong"{description, pos=0.7}, from=6-1, to=8-1]
			\arrow["{y^{0,0}}"{description, pos=0.7}, from=4-1, to=6-1]
			\arrow["\cong"{description, pos=0.7}, from=5-2, to=7-2]
			\arrow["{y^{0,1}}"{description, pos=0.7}, from=3-2, to=5-2]
			\arrow["\cong"{description, pos=0.7}, from=6-3, to=8-3]
			\arrow["\cong"{description, pos=0.7}, dashed, from=5-4, to=7-4]
			\arrow["{f\fun 1}"', shift right=2, dashed, from=1-4, to=3-4]
			\arrow[shift right=2, from=1-2, to=3-2]
			\arrow["y"{description, pos=0.7}, dashed, from=3-4, to=5-4]
			\arrow["{y^{1,0}}"{description, pos=0.7}, from=4-3, to=6-3]
			\arrow["{f_1}"'{pos=0.7}, shift right=2, from=2-3, to=4-3]
			\arrow[from=8-3, to=7-4]
			\arrow["{f_0\otimes 1}"', shift right=2, from=2-1, to=4-1]
			\arrow["{g_0}"{pos=0.7}, shift left=2, from=2-3, to=4-3]
			\arrow["{g\fun 1}", shift left=2, from=1-4, to=3-4]
			\arrow[shift left=2, from=1-2, to=3-2]
			\arrow["{h_0\otimes 1}"{description}, curve={height=30pt}, from=4-1, to=8-1]
			\arrow["{h\fun 1}"{description}, curve={height=-24pt}, from=3-4, to=7-4]
			\arrow["{h_0\otimes 1}"{description}, curve={height=30pt}, from=3-2, to=7-2]
			\arrow["{h_1\otimes 1}"{description, pos=0.6}, curve={height=-30pt}, from=4-3, to=8-3]
		\end{tikzcd}\]
		Since coequalizers are computed pointwise, $(e^{0,0}\to e)$ is the coequalizer of $(f\fun 1, g\fun 1): a\fun d\to b\fun d$.
		The object $ e\in \mcal{V}$ is the pushout of the square it is in since colimits commute with each other, and the	 morphism $e \to c\fun d$ is an isomorphism since each of the morphisms $e^{i,j}\to c_0\otimes d_j$ are. 		
	\end{proof}

	Since $[\to,\mcal{V}]$ has two monoidal products, it is natural to want to know how they are related.
	First, note that the funny tensor is always `behind' the inherited pointwise tensor, i.e, there is always a canonical morphism,
	$\mu: (a\fun b) \to a\otimes b
	$
	induced by \ref{main2}:
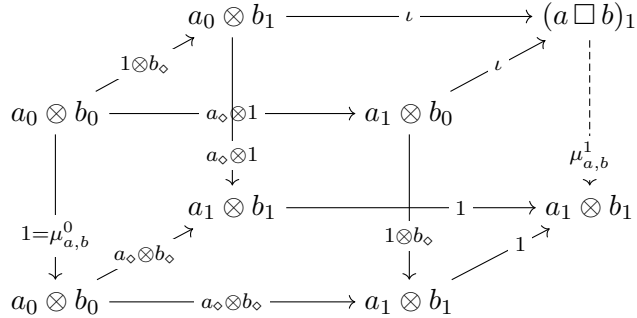
	\[\begin{tikzcd}[cramped]
		& {a_0\otimes b_1} && {(a\fun b)_1} \\
		{a_0\otimes b_0} && {a_1\otimes b_0} \\
		& {a_1\otimes b_1} && {a_1\otimes b_1} \\
		{a_0\otimes b_0} && {a_1\otimes b_1}
		\arrow["{a_\diamond\otimes 1}"{description}, from=2-1, to=2-3]
		\arrow["\iota"{description}, from=2-3, to=1-4]
		\arrow["1"{description, pos=0.7}, from=3-2, to=3-4]
		\arrow["{a_\diamond\otimes b_\diamond}"{description}, from=4-1, to=3-2]
		\arrow["{a_\diamond\otimes b_\diamond}"{description}, from=4-1, to=4-3]
		\arrow["1"{description, pos=0.7}, from=4-3, to=3-4]
		\arrow["{1 = \mu^0_{a,b}}"{description, pos=0.7}, from=2-1, to=4-1]
		\arrow["{1\otimes b_\diamond}"{description}, from=2-1, to=1-2]
		\arrow["{a_{\diamond}\otimes 1}"{description, pos=0.8}, from=1-2, to=3-2]
		\arrow["{1\otimes b_\diamond}"{description, pos=0.7}, from=2-3, to=4-3]
		\arrow["{\mu^1_{a,b}}"{description, pos=0.8}, dashed, from=1-4, to=3-4]
		\arrow["\iota"{description}, from=1-2, to=1-4]
	\end{tikzcd}\]
	Further, since both the monoidal structures have the same unit, there is an isomorphism
	\[
	1: (I \to I) \to (I\to I)\]	
	All this information points to the following fact.
	\begin{lemma}
		The identity functor is a lax monoidal functor between the monoidal categories $([\to,\mcal{V}], \fun, I\to I)$ and $([\to,\mcal{V}], \otimes, I\to I)$.
	\end{lemma}
	\begin{proof}
		Straightforward diagram chasing using coherence axioms in $\mcal{V}$, and Lemma \ref{main}. 
	\end{proof}

	Further, the naturality of $\mu$ also shows that when given $a,b \in [\to,\mcal{V}]$, there are natural transformations, as shown below:
	\[\begin{tikzcd}[cramped,column sep=scriptsize]
		{[\to,\mcal{V}]} &&& {[\to,\mcal{V}]} && {[\to,\mcal{V}]} &&& {[\to,\mcal{V}]}
		\arrow[""{name=0, anchor=center, inner sep=0}, "{-\otimes b}"{description}, curve={height=24pt}, from=1-1, to=1-4]
		\arrow[""{name=1, anchor=center, inner sep=0}, "{-\fun b}"{description}, curve={height=-24pt}, from=1-1, to=1-4]
		\arrow[""{name=2, anchor=center, inner sep=0}, "{a\fun -}"{description}, curve={height=-24pt}, from=1-6, to=1-9]
		\arrow[""{name=3, anchor=center, inner sep=0}, "{a\otimes -}"{description}, curve={height=24pt}, from=1-6, to=1-9]
		\arrow["{\mu_{-,b}}", shorten <=6pt, shorten >=6pt, Rightarrow, from=1, to=0]
		\arrow["{\mu_{a,-}}", shorten <=6pt, shorten >=6pt, Rightarrow, from=2, to=3]
	\end{tikzcd}\]

	When the monoidal product  $\otimes$ in $\mcal{V}$ is given by the cartesian product (like in \category{Set}), the componentwise product on $[\to,\mcal{V}]$ is also cartesian. As such, there are projections, $\pi^{a,b}:a\times b\to a$ and $\pi'^{a,b}: a\times b\to b$, which are \textit{natural} in $a,b$. So, we have natural transformations, $\pi^{-,b}: -\times b \Rightarrow id $  and $\pi^{a,-}: a\times -\Rightarrow id$, where $id$ is the identity functor on $[\to,\mcal{V}]$. Thus, we can equip the funny tensor with projections that are natural by considering the composites:
\[\begin{tikzcd}[cramped]
	{[\to,\cal{V}]} &&& {[\to,\cal{V}]} && {[\to,\cal{V}]} &&& {[\to,\cal{V}]}
	\arrow[""{name=0, anchor=center, inner sep=0}, "{-\times b}"{description}, from=1-1, to=1-4]
	\arrow[""{name=1, anchor=center, inner sep=0}, "{-\fun b}"{description}, curve={height=-30pt}, from=1-1, to=1-4]
	\arrow[""{name=2, anchor=center, inner sep=0}, "{a\fun -}"{description}, curve={height=-30pt}, from=1-6, to=1-9]
	\arrow[""{name=3, anchor=center, inner sep=0}, "{a\times -}"{description}, from=1-6, to=1-9]
	\arrow[""{name=4, anchor=center, inner sep=0}, "id"{description}, curve={height=30pt}, from=1-1, to=1-4]
	\arrow[""{name=5, anchor=center, inner sep=0}, "id"{description}, curve={height=30pt}, from=1-6, to=1-9]
	\arrow["{\mu_{-,b}}", shorten <=4pt, shorten >=4pt, Rightarrow, from=1, to=0]
	\arrow["{\pi^{-,b}}", shorten <=4pt, shorten >=4pt, Rightarrow, from=0, to=4]
	\arrow["{\mu_{a,-}}", shorten <=4pt, shorten >=4pt, Rightarrow, from=2, to=3]
	\arrow["{\pi'^{a,-}}", shorten <=4pt, shorten >=4pt, Rightarrow, from=3, to=5]
\end{tikzcd}\]

	\begin{lemma} When $\otimes$ on $\mcal{V}$ is a cartesian monoidal structure, there are projection maps, 
		\[
		a\fun b \to a \quad \text{ and }\quad a\fun b \to b
		\]
		out of the funny tensor in $[\to,\mcal{V}]$ that are
		natural in $a$ and in $b$ respectively.
	\end{lemma}

	Another relationship between the two tensors is the existence of an interchange law. Since $\otimes$ is symmetric monoidal, there is an isomorphism,
	\[
	(a\otimes b) \otimes (c\otimes d) \to (a\otimes c)\otimes (b\otimes d)
	\]
	This map induces an interchange law between the funny tensor and the inherited tensor, as seen in the following diagram:
	
	\[\begin{tikzcd}[ column sep=0.01ex]
		& {(a_1\otimes b_1)\otimes (c_0\otimes d_0)} && {((a\otimes b)\fun (c\otimes d))_1} \\
		{(a_0\otimes b_0)\otimes (c_0\otimes d_0)} && {(a_0\otimes b_0)\otimes (c_1\otimes d_1)} \\
		& {(a_1\otimes c_0)\otimes (b_1\otimes d_0)} && {((a\fun c)_1\otimes (b\fun d)_1)} \\
		{(a_0\otimes c_0)\otimes (c_0\otimes d_0)} && {(a_0\otimes c_1)\otimes  (b_0\otimes d_1)}
		\arrow[from=2-1, to=1-2]
		\arrow[from=2-1, to=2-3]
		\arrow[from=1-2, to=1-4]
		\arrow[from=2-3, to=1-4]
		\arrow[from=4-1, to=3-2]
		\arrow[from=4-1, to=4-3]
		\arrow[from=3-2, to=3-4]
		\arrow[from=4-3, to=3-4]
		\arrow["coh"{description, pos=0.7}, from=1-2, to=3-2]
		\arrow["{\zeta_{a,b,c,d}^0 = coh}"{description, pos=0.7}, from=2-1, to=4-1]
		\arrow["coh"{description, pos=0.7}, from=2-3, to=4-3]
		\arrow["{\zeta^1_{a,b,c,d}}"{description, pos=0.7}, dashed, from=1-4, to=3-4]
	\end{tikzcd}\]
	
	Since the coherence isomorphisms are natural, it follows that $\zeta$ is a natural transformation too by the typical argument that talks about the cube commuting. 
	
	When the tensor $\otimes $ on $\mcal{V}$ is assumed to be cartesian monoidal, i.e, if $\otimes$ is given by a choice of cartesian product structure on $\mcal{V}$, the inherited tensor on $[\to,\mcal{V}]$ is also cartesian monoidal. It is a known fact (\textit{e.g.} \cite{2-monoidal}) that  when a category $\mcal{C}$ has a cartesian monoidal structure, $(\mcal{C},\times, 1)$ and another monoidal structure $(\mcal{C}, \fun, I)$, it is possible to obtain an interchange morphism by the universal property of cartesian products, as:
	
	\[\begin{tikzcd}[cramped]
		&& {(a\times b)\fun (c\times d)} \\
		\\
		{(a\fun c)} && {(a\fun c)\times (b\fun d)} && {(b\fun  d)}
		\arrow["{\pi\fun \pi}"{description}, from=1-3, to=3-1]
		\arrow["{\pi'\fun \pi'}"{description}, from=1-3, to=3-5]
		\arrow["\pi"{description}, from=3-3, to=3-1]
		\arrow["{\pi'}"{description}, from=3-3, to=3-5]
		\arrow["{\zeta'_{a,b,c,d}}"{description}, dashed, from=1-3, to=3-3]
	\end{tikzcd}\]

	Though, on the face of it, there are two interchange laws, $\zeta$ and $\zeta'$ in $[\to,\mcal{V}]$, these can easily be seen to be equal since $\zeta$ satisfies the same universal property that $\zeta'$ does.  The morphism $\zeta$ is obtained as a morphism \textit{out} of the colimit $(a\times b)\fun (c\times d)$, while the morphism $\zeta'$ is obtained as a morphism \textit{into} the limit $(a\fun c)\times (b\fun d)$. 
	
	In general, it is known that the interchange law $\zeta'$ makes the 5-tuple $(\mcal{C}, \fun, \times, I, 1 )$ into a \textit{duoidal category}, i.e, a category with two monoidal structures that have an interchange law between them, and satisfy certain coherence axioms \cite{2-monoidal}.
	
	\begin{definition}\cite{2-monoidal}
		 A duoidal category is a five tuple, $(\mcal{C}, \fun, I, \otimes, J)$ where $(\mcal{C}, \fun, I)$ and $(\mcal{C}, \otimes, J)$ are monoidal categories with units $I$ and $J$ respectively, along with a morphism (the interchange law):
		 \[
		 \zeta_{a,b,c,d}: (A\otimes B)\fun (C\otimes D) \to (A\fun C)\otimes (B\fun D)
		 \]
		 which is natural in $A,B,C,D$, and three morphisms,
		 \[
		 \Delta_I: I \to I\otimes I, \qquad \mu_J: J\fun J \to J,\qquad \iota_J=\epsilon_I: I\to J
		 \]
		 which satisfy certain coherence axioms (see \text{e.g.} \cite{2-monoidal}).	 
	\end{definition}

With this discussion, and the  discussion above, we have illustrated the following fact.
	
	\begin{proposition}
		When the tensor product $\otimes$ on $\mcal{V}$ is given by the cartesian structure in $\mcal{V}$,  the		 5-tuple, $([\to,\mcal{V}], \fun, \otimes, I, I)$ is duoidal. 
	\end{proposition}
	
\begin{remark}
	We suspect that the condition that $\otimes$ is cartesian monoidal can be omitted, and that the category acquires a duoidal structure even otherwise. We work with $\mcal{V} = \category{Set}$ for most of Chapter \ref{chapter6}, and hence this does not pose a limitation. 
\end{remark}

\chapter{Multicategories Enriched in a Duoidal Category}\label{chapter6}

We introduce the notion of a $(\mcal{V}, \otimes, \fun)$-multicategory. Here, $\mcal{V}$ is a category with two monoidal structures, $\otimes$ and $\fun$, that agree on their unit $I$, and such that there is an interchange law, 
\[
\zeta_{a,b,c,d}: (a\otimes b)\fun (c\otimes d) \to (a\fun c) \otimes (b\fun d)
\]
We further assume that this law satisfies the coherence axioms  of a duoidal category introduced in the previous chapter. 

These structures subsume traditional multicategories as well as effectful multicategories, and hence serves as  a framework for talking about both pure and impure computation.

\section{Definition and Basic Theory}

We begin by motivating our reason for the base of enrichment to be duoidal. There are two equivalent presentations of the definition of a multicategory. The circle-i presentation was the one used by Lambek in his original definition \cite{ded-2}, and this is the one adopted by Staton and Levy since this gives a natural way to define a premulticategory \cite{staton-levy}. The other presentation, which use used in Chapter 3 uses a simultaneous substitution operation, 
\[
sub: \Mhom{\bb{C}}{A_1,\ldots, A_n}{B} \times (\prod_{i=1}^{n} \Mhom{\bb{C}}{\Delta_i}{A_i}) \to \Mhom{\bb{C}}{\Delta_1,\ldots, \Delta_n}{B}
\]
which is required to satisfy certain laws. Intuitively, this can be seen as `plugging in' a sequence of terms, $u_1,\ldots, u_n$ into a term $t$  to obtain some $t[u_1,\ldots, u_n]$. Here, it is assumed that the collection of things to be plugged in, comes with a certain structure -- as an ordered tuple. 

A $(V,\otimes, \fun)$-multicategory instead asks for an operation,
\[
sub: \Mhom{\bb{C}}{A_1,\ldots, A_n}{B} \otimes \Big(\fun_{i=1}^{n} \Mhom{\bb{C}}{\Delta_i}{A_i} \Big) \to \Mhom{\bb{C}}{\Delta_1,\ldots, \Delta_n}{B}
\]
which allows us to give a different structure to the collection of things to be plugged in. However, for this kind of composition operation to be sensible, we require it to satisfy slightly modified axioms of a multicategory. Modifying the unitality law is rather straightforward.
However, the associativity law, given by the commutative diagram:

\[\begin{tikzcd}[cramped]
	{\substack{\Mhom{\bb{C}}{A_1,\ldots, A_n}{B}\times \\(\prod_{i=1}^n (\Mhom{\bb{C}}{B_i^1,\ldots, B^{k_i}_i}{A_i} \times (\prod_{j=1}^{k_i} \Mhom{\bb{C}}{\Delta^j_i}{B^j_i})})} &&& {\substack{\Mhom{\bb{C}}{A_1,\ldots, A_n}{B}\\\times (\prod_{i=1}^n \Mhom{\bb{C}}{\Delta^1_i,\ldots, \Delta_1^{k_i} }{A_i})}} \\
	\\
	{\substack{(\Mhom{\bb{C}}{A_1,\ldots, A_n}{B} \times (\prod_{i=1}^n  \Mhom{\bb{C}}{B^1_i,\ldots, B^{k_i}_i}{A_i})\\\times (\prod_{i=1}^n \prod_{j=1}^{k_i} \Mhom{\bb{C}}{\Delta^j_i}{B^j_i}))}} \\
	\\
	{\substack{\Mhom{\bb{C}}{B^1_1,\ldots, B^{k_1}_1,\ldots, B^1_n,\ldots, B^{k_n}_n}{B}\\ \times (\prod_{i=1}^n \prod_{j=1}^{k_i} \Mhom{\bb{C}}{\Delta^j_i}{B^j_i}))}} &&& {\substack{\Mhom{\bb{C}}{\Delta^1_1,\ldots, \Delta^{k_i}_1,\ldots, \Delta_n^1,\ldots, \Delta^{k_n}_n}{B}}}
	\arrow["{1\times (\prod_{i=1}^n sub)}", from=1-1, to=1-4]
	\arrow["sub", from=1-4, to=5-4]
	\arrow["\cong"', from=1-1, to=3-1]
	\arrow["{sub\times 1}"', from=3-1, to=5-1]
	\arrow["sub"', from=5-1, to=5-4]
\end{tikzcd}\]
makes essential use of the isomorphism,

\[
\prod_{i=1}^{n} a^i \otimes b^i \xrightarrow{\sim} (\prod_{i=1}^{n} a^i) \otimes \prod_{i=1}^{n}  b^i
\]
that is inherent to any symmetric monoidal category. Since we are working with two monoidal structures, we settle for a (non-iso) morphism,

\[
{{\fun}}_{i=1}^n  (a^i\otimes b^i) \to \bigl({{\fun}}_{i=1}^n a^i\bigr) \otimes \bigl({{\fun}}_{i=1} b^i \bigr) 
\] 

Such a morphism arises from and defines an interchange law, which is why have the requirement that $(\mcal{V}, \otimes, \fun)$ is duoidal.

Thus, we assume that $\mcal{V}$ has two monoidal structures, $\otimes$ and $\fun$, with a common unit, $I$. We write the unitors and associators as  $\lambda_{\otimes}, \rho_{\otimes}, \alpha_\otimes$ and $\lambda_{\fun}, \rho_{\fun}, \alpha_{\fun}$ for the respective structures. Further, we assume that there is an interchange law, 
\[
\zeta_{a,b,c,d}: (a\otimes b)\fun (c\otimes d) \to (a\fun c)\otimes (b\fun d)
\]
The following are two  examples of such categories which we will mainly be interested in. 

\begin{example}\label{pure}
	If $(\mcal{V}, \otimes, I)$ is symmetric monoidal, there is an interchange morphism,
	\[
	(a\otimes b)\otimes (c\otimes d) \to (a\otimes c)\otimes (b\otimes ds)
	\]
	This makes $(\mcal{V},\otimes, \otimes, I,I)$ duoidal, and the `two' monoidal structures have the same unit, $I$. 
	In particular, $(\category{Set}, \times, \times, 1,1)$ is an example of such a duoidal category.
\end{example}

\begin{example}\label{effectful}
	As seen in the previous chapter, if $(\mcal{V}, \times, 1)$ is a cartesian monoidal category which is complete and cocomplete, then, \\$([\to,\mcal{V}], \times, \fun, 1, 1)$ is a duoidal category where $\fun$ is the funny tensor.
	When $\mcal{V}=\category{Set}$, the 5-tuple $([\to,\category{Set}], \times, \fun, 1, 1)$ is a duoidal category.
\end{example}

\begin{notation}
We shall work with the same general assumption (unless otherwise mentioned) that $\otimes$ and $\fun$ make $\mcal{V}$ a duoidal category, and that they agree on their units, $I$. Hence, from now on, use the triple $(\mcal{V},\otimes, \fun)$ to denote this category. 
\end{notation}

\begin{definition}
	A $(\mcal{V},\otimes,\fun)$-multicategory, $\bb{C}$ consists of the following data:
	\begin{itemize}
		\item A set of objects, $\ob{\bb{C}}$
		\item For every sequence of objects, $A_1,\ldots, A_n,B$, an object, 
		\[
		\Mhom{\bb{C}}{A_1,\ldots, A_n}{B} \in \mcal{V}
		\]
		\item For every object $A$, an arrow
		\[
		id_A: I \to \Mhom{\bb{C}}{A}{A}
		\]
		\item A substitution operation,
		\begin{align*}
			&sub:\Mhom{\bb{C}}{A_1,\ldots, A_n}{B}\otimes (\fun_{i=1}^n \Mhom{\bb{C}}{\Delta_i}{A_i})\\
			&	\qquad\qquad\qquad\qquad\qquad\qquad\qquad \to \Mhom{\bb{C}}{\Delta_1,\ldots, \Delta_n}{B}
		\end{align*}
	\end{itemize}
	subject to the following axioms
	
	\begin{enumerate}
		\item Left unitality: 
		
		\[\begin{tikzcd}
			{\Mhom{\bb{C}}{A}{A}\otimes \Mhom{\bb{C}}{\Delta}{A}} && {\Mhom{\bb{C}}{\Delta}{A}} \\
			\\
			{I\otimes \Mhom{\bb{C}}{\Delta}{A}}
			\arrow["{id_A\otimes 1}", from=3-1, to=1-1]
			\arrow["sub", from=1-1, to=1-3]
			\arrow["\lambda"', from=3-1, to=1-3]
		\end{tikzcd}\]

		\item Right unitality: 
		
	\[\begin{tikzcd}[cramped]
		{\Mhom{\bb{C}}{A_1,\ldots, A_n}{B}\otimes (\fun_{i=1}^n \Mhom{\bb{C}}{A_i}{A_i})} && {\Mhom{\bb{C}}{A_1,\ldots, A_n}{B}} \\
		\\
		{\Mhom{\bb{C}}{A_1,\ldots, A_n}{B}\otimes (\fun_{i=1}^n I)} && {\Mhom{\bb{C}}{A_1,\ldots, A_n}{B}\otimes I}
		\arrow["sub", from=1-1, to=1-3]
		\arrow["{1\otimes (\fun_{i=1}^n id_{A_i})}", from=3-1, to=1-1]
		\arrow["{1\otimes coh}"', from=3-1, to=3-3]
		\arrow["\lambda"', from=3-3, to=1-3]
	\end{tikzcd}\]
		
		\item Associativity: 
	\[\begin{tikzcd}[cramped]
		{\substack{\Mhom{\bb{C}}{A_1,\ldots, A_n}{B}\\\otimes (\fun_{i=1}^n (\Mhom{\bb{C}}{B_i^1,\ldots, B_i^{k_i}}{A_i} \otimes (\fun_{j=1}^{k_i} \Mhom{\bb{C}}{\Delta^j_i}{B^j_i})  ))}} &&& {\substack{          \Mhom{\bb{C}}{A_1,\ldots, A_n}{B}\\\otimes (\fun_{i=1}^n \Mhom{\bb{C}}{\Delta^1_i,\ldots, \Delta^{k_i}_i}{A_i})                             }} \\
		\\
		{\substack{(\Mhom{\bb{C}}{A_1,\ldots, A_n}{B}\otimes   (\fun_{i=1}^n \Mhom{\bb{C}}{B^1_i,\ldots, B_i^{k_i}}{A_i}) ) \\\otimes (\fun_{i=1}^n \fun_{j=1}^{k_i} \Mhom{\bb{C}}{\Delta^j_i}{B^j_i})}} \\
		\\
		{\substack{\Mhom{\bb{C}}{B^1_1,\ldots, B^{k_1}_1,\ldots, B^1_n,\ldots, B^{k_n}_n}{B}\\\otimes (\fun_{i=1}^n \fun_{j=1}^{k_i}\Mhom{\bb{C}}{\Delta^j_i}{B^j_i}) }} &&& {\substack{\Mhom{\bb{C}}{\Delta^1_1,\ldots, \Delta^{k_1}_1,\ldots, \Delta^1_n,\ldots, \Delta^{k_n}_n}{Bdd}}}
		\arrow["{1\otimes (\fun_{i=1}^n sub)}", from=1-1, to=1-4]
		\arrow["sub", from=1-4, to=5-4]
		\arrow["coh"', from=1-1, to=3-1]
		\arrow["{sub\otimes 1}"', from=3-1, to=5-1]
		\arrow["sub"', from=5-1, to=5-4]
	\end{tikzcd}\]

	\end{enumerate}	
	
	The map $coh$ in the associativity axiom makes use of the interchange law, and the map $coh$ in the right unitality axiom makes use of the fact that $I$ is also a unit for the monoidal structure $\fun$. 
\end{definition}

\begin{example}
	A $(\mcal{V},\otimes, \otimes)$-multicategory (where $\mcal{V}$ is a symmetric monoidal category considered as a duoidal category) is the same as a $\mcal{V}$-enriched multicategory. When $\mcal{V} =\category{Set}$, this is the same as a multicategory.  
\end{example}

	Just as every category can be trivially thought of as a multicategory, any $(\mcal{V},\otimes)$-enriched category can be thought of as a $(\mcal{V},\otimes,\fun)$-multicategory when $\mcal{V}$ is \textit{nice}.
	Assume that $\mcal{V}$ has an initial object $0$, and that $\otimes$ is cartesian monoidal. If $\mcal{C}$ is a $(\mcal{V},\otimes)$-enriched category, define the $(\mcal{V},\otimes,\fun)$-multicategory $J\mcal{C}$ as follows: 
	\begin{enumerate}
		\item Objects: Same as that of $\mcal{C}$
		\item Multi-homs: 
		\begin{align*}
			\Mhom{J\mcal{C}}{A}{B} & = \Mhom{\mcal{C}}{A}{B}\\
			\Mhom{J\mcal{C}}{A_1,\ldots, A_n}{B} & = 0
		\end{align*}
		when $n>2$.
		\item Identities are as in $\mcal{C}$: 
		\[\begin{tikzcd}[cramped]
			I & {\Mhom{\mcal{C}}{A}{A}} & {\Mhom{J\mcal{C}}{A}{A}}
			\arrow["{id_A}", from=1-1, to=1-2]
			\arrow["{=}", from=1-2, to=1-3]
		\end{tikzcd}\]
		\item Substitution is  derived from composition in $\mcal{C}$ when the context is unary: 
		\[\begin{tikzcd}[cramped]
			{\Mhom{J\mcal{C}}{A}{B}\otimes \Mhom{J\mcal{C}}{B}{C}} && {\Mhom{J\mcal{C}}{A}{C}} \\
			{\Mhom{\mcal{C}}{A}{B}\otimes \Mhom{\mcal{C}}{B}{C}} && {\Mhom{\mcal{C}}{A}{B}}
			\arrow["sub", from=1-1, to=1-3]
			\arrow["{=}", no head, from=1-1, to=2-1]
			\arrow["{\circ_{\mcal{C}}}", from=2-1, to=2-3]
			\arrow["{=}"', no head, from=2-3, to=1-3]
		\end{tikzcd}\]
		Otherwise, we use the fact that the initial object acts as an annihilator for $\otimes$:
		\[\begin{tikzcd}[cramped]
			{\small{\Mhom{J\mcal{C}}{A_1,\ldots, A_n}{B}\otimes (\fun_{i=1}^n \Mhom{J\mcal{C}}{\Delta_i}{A_i})}} & {\small{\Mhom{J\mcal{C}}{\Delta_1,\ldots, \Delta_n}{B}}} \\
			{\small{0\otimes (\fun_{i=1}^n \Mhom{J\mcal{C}}{\Delta_i}{A_i})}} & {\small{0}}
			\arrow["sub", from=1-1, to=1-2]
			\arrow["{=}", no head, from=1-1, to=2-1]
			\arrow["\iota"', from=2-2, to=1-2]
			\arrow["\cong", from=2-1, to=2-2]
		\end{tikzcd}\]
		where $\iota$ is the unique map from the initial object in $\mcal{V}$.
	\end{enumerate}
	Checking that this is a $(\mcal{V},\otimes, \fun)$-multicategory is straightforward. Note that the tensor $\fun$ never comes into play. This illustrates the fact that in a $(\mcal{V},\otimes, \fun)$-multicategory, $\otimes$ is responsible for composition, while $\fun$ provides structure for the inputs.

\begin{definition}
	Let $\bb{C}$ and $\bb{D}$ be  $(\mcal{V},\otimes,\fun)$-multicategories. Then a $(\mcal{V},\otimes, \fun)$-multicategory morphism, $f: \bb{C}\to\bb{D}$ consists of the following data:
	\begin{itemize}
		\item A mapping of objects, $f_{ob}:\ob{\bb{C}}\to \ob{\bb{D}}$,
		\item A family of maps,
		\[
		f_{\Gamma;A}:\Mhom{\bb{C}}{\Gamma}{B} \to \Mhom{\bb{D}}{f\Gamma}{fB}
		\]
		such that the following diagrams commute: 
	\end{itemize}
	\begin{enumerate}
		\item Identities are preserved: 
		\[\begin{tikzcd}[cramped]
			{\Mhom{\bb{C}}{A}{A}} && {\Mhom{\bb{D}}{fA}{fA}} \\
			I
			\arrow["{id_A^{\bb{C}}}", from=2-1, to=1-1]
			\arrow["{id_{fA}^{\bb{D}}}"', from=2-1, to=1-3]
			\arrow["{f_{A:A}}", from=1-1, to=1-3]
		\end{tikzcd}\]
		\item Composition is preserved: 
		\[\begin{tikzcd}[cramped,row sep=scriptsize]
			{\small\Mhom{\bb{C}}{A_1,\ldots, A_n}{B}\otimes (\fun_{i=1}^n \Mhom{\bb{C}}{\Delta_i}{A_i})} && {\small\Mhom{\bb{D}}{\Delta_1,\ldots, \Delta_n}{B}} \\
			\\
			{\small\Mhom{\bb{D}}{fA_1,\ldots, fA_n}{fB}\otimes (\fun_{i=1}^n \Mhom{\bb{D}}{f\Delta_i}{fA_i}} && {\small\Mhom{\bb{D}}{f\Delta_1,\ldots, f\Delta_n}{fB}}
			\arrow["{sub^{\bb{C}}}", from=1-1, to=1-3]
			\arrow["{f_{A_1,\ldots, A_n;B}\otimes (\fun_{i=1}^n f_{\Delta_i;A_i})}", from=1-1, to=3-1]
			\arrow["{sub^{\bb{D}}}", from=3-1, to=3-3]
			\arrow["{f_{\Delta_1,\ldots, \Delta_n;B}}", from=1-3, to=3-3]
		\end{tikzcd}\]
	\end{enumerate}
\end{definition}

It is clear that morphisms of $(\mcal{V},\otimes,\fun)$-multicategories compose, and hence we obtain a category, $(\mcal{V},\otimes, \fun)$-\category{MultiCat}. Just as restricting a multicategory  to its unary contexts gives a category, restricting a $(\mcal{V},\otimes, \fun)$-multicategory to its unary contexts gives a $(\mcal{V},\otimes)$-category. Further, this mapping is functorial, and has a left adjoint, given by the construction of $J$, outlined above. As a consequence of this general fact, we re-obtain the following adjunction from Chapter \ref{chapter3}:

\begin{lemma}
	Assume that $\otimes$ in the duoidal category $(\mcal{V},\otimes,\fun)$ is cartesian monoidal, and that $\mcal{V}$ has an initial object $0$. Then,	there is an adjunction,
	\[\begin{tikzcd}[cramped]
		{(\mcal{V},\otimes, \fun)\text{-}\category{MultiCat}} && {(\mcal{V},\otimes)\text{-}\category{Cat}}
		\arrow[""{name=0, anchor=center, inner sep=0}, "{\overline{(-)}}"', shift right=3, from=1-1, to=1-3]
		\arrow[""{name=1, anchor=center, inner sep=0}, "J"', shift right=3, from=1-3, to=1-1]
		\arrow["\dashv"{anchor=center, rotate=-90}, draw=none, from=1, to=0]
	\end{tikzcd}\]
	where $\overline{(-)}$ restricts a $(\mcal{V},\otimes,\fun)$-multicategory to its unary contexts, and $J$ is the discrete multicategory on a category. 
\end{lemma}

\section{Effectful Multicategories Revisited}

The motivation behind introducing $(\mcal{V},\otimes,\fun)$-multicategories was to talk about effectful multicategories. This section will provide a proof of the fact that $([\to,\category{Set}],\times, \fun)$-multicategories are equivalent to  effectful multicategories.

\subsection{Unpacking definitions}

We state a couple of  technical lemmas which will allow us to have a peek into the internal functioning of  a $([\to,\category{Set}],\times, \fun)$-multicategory. These lemmas will allow us to derive a simpler presentation  of $([\to,\category{Set}],\times, \fun)$-multicategories. The proofs for these can be found in Appendix \ref{appendix1}.

\begin{notation}
	Many proofs here will involve using the pushout characterization of the $n$-ary funny tensor (Corollary \ref{pushout-char}). To write statements about these, we will require maps of form $f_0^1\otimes \ldots f_0^{i-1}\otimes f_1^{i}\otimes f_0^{i+1}\otimes\ldots f_0^n$, and we use a special notation for these. 
	
	Let $f_0^i: a_0^i\to b_0^i $ be a collection of morphisms for $i=1,\ldots, n$ and assume that there is another collection $f_1^i: a_1^i\to b_1^i$ for $i=1,\ldots, n$. Then, we use the expression $\prod_{i=1}^{n} f_0^i | f_1^j$ to denote the morphism $\prod_{i=1}^{n} a_{\delta_{i,j}}^i \to \prod_{i=1}^{n} b_{\delta_{i,j}}^i$ which looks like $f_0^1\otimes\ldots f_0^{j-1}\otimes f_1^j\otimes f_0^{j+1}\otimes\ldots, f^n_0$.
\end{notation}

\begin{lemma}
\label{unwrap-1}
	Let $a,c \in [\to,\category{Set}]$, and $\{b_i\}_{1\leq i\leq n}$ be a family of elements in $[\to,\category{Set}]$. A map,
	\[
	f: a\times (\fun_{i=1}^{n} b^i )\to c
	\]
	in $[\to,\category{Set}]$
	determines and is determined by the following information:
	\begin{enumerate}
		\item A map,
		\[
		f_0:	a_0\times (\prod_{i=1}^{n} b^i_0) \to c_0
		\]
		in \category{Set}.
		\item For every $j \in \{1,\ldots, n\}$, a map
		\[
		f_{1,j}:	a_1\times (\prod_{i=1}^{n} b^i_{\delta_{i,j}}) \to c_1
		\] in \category{Set}.
	\end{enumerate}
	such that the following diagrams commute for all $j,k\in \{1,\ldots, n\}$:
	
	\[\begin{tikzcd}
		{a_0\times (\prod_{i=1}^n b_0^i)} && {a_1\times (\prod_{i=1}^n b^i_{\delta_{i,j}})} \\
		{c_0} && {c_1}
		\arrow["{f_0}", from=1-1, to=2-1]
		\arrow["{a_\diamond \times (\prod_{i=1}^n 1 | b^j_1)}", from=1-1, to=1-3]
		\arrow["{f_{1,j}}", from=1-3, to=2-3]
		\arrow["{c_\diamond}", from=2-1, to=2-3]
	\end{tikzcd}\]
	\[\begin{tikzcd}
		& {a_1\times (\prod_{i=1}^n b^i_0)} \\
		\\
		{a_1\times (\prod_{i=1}^n b^i_0)} && {c_1} \\
		\\
		& {a_1\times (\prod_{i=1}^n b^i_0)}
		\arrow["{1\times (\prod_{i=1}^n 1 | b^j_1)}"{description}, from=3-1, to=1-2]
		\arrow["{1\times (\prod_{i=1}^n 1 | b^k_1)}"{description}, from=3-1, to=5-2]
		\arrow["{f_{1,j}}"{description}, from=1-2, to=3-3]
		\arrow["{f_{1,k}}"{description}, from=5-2, to=3-3]
	\end{tikzcd}\]
\end{lemma}

Now that maps $f\times (\fun_{i=1}^n b^i) \to c$ have been deconstructed, we explore,  what the Associativity axiom of a $(\mcal{V},\otimes, \fun)$-multicategory is saying.

\begin{lemma}\label{unwrap-2}

	Given some $a \in [\to,\category{Set}]$, and families, $\{b^i\}_{1\leq i\leq n}$ and $\{c^{i,j}\}_{1\leq i\leq n, 1\leq j\leq k_i}$ in $[\to,\category{Set}]$,
	to give maps,
	\begin{align*}
		f^i&: b^i\times (\fun_{j=1}^{k_i} c^{i,j} ) \to e^i, \qquad 1\leq i\leq n\\
		g&: a\times (\fun_{i=1}^n e^i) \to q\\
		x&: a\times (\fun_{i=1}^n b^i) \to d\\
		y&: d\times (\fun_{i=1}^n \fun_{j=1}^{k_i} c^{i,j}) \to q
	\end{align*}
	such that the following square commutes
	\begin{equation}\label{main-square}
		\begin{tikzcd}
			{a\times (\fun_{i=1}^n (b^i \times (\fun_{j=1}^{k_i} c^{i,j})) )} &&& {a\times (\fun_{i=1}^n e^i)} \\
			{(a\times (\fun_{i=1}^n b^i))\times (\fun_{i=1}^n \fun_{j=1}^{k_i} c^{i,j})} \\
			{d\times (\fun_{i=1}^n \fun_{j=1}^{k_i} c^{i,j})} &&& q
			\arrow["coh", from=1-1, to=2-1]
			\arrow["{x\times 1}", from=2-1, to=3-1]
			\arrow["y", from=3-1, to=3-4]
			\arrow["{1\times (\fun_{i=1}^n f^i)}", from=1-1, to=1-4]
			\arrow["g", from=1-4, to=3-4]
		\end{tikzcd}
	\end{equation}
	is equivalent to saying that the maps
	\begin{flalign*}
		f^i_0&: b^i_0\times \prod_{j=1}^{k_i}c^{i,j}_0\to e^i_0, & \{f^i_{1,m}: b^i_1\times \prod_{j=1}^{k_i}c^{i,j}_{\delta_{j,m}}\to e^i_1\}_{1\leq m \leq k_i}  \\
		g_0&: a_0\times \prod_{i=1}^{n} e^i_0 \to q_0, & \{g_{1,l}: a_1\times \prod_{i=1}^{n} e^i_{\delta_{i,l}}\to q_1\}_{1\leq l \leq n}\\
		x_0&: a_0\times \prod_{i=1}^{n} b^i_0\to d_0, & \{x_{1,l}: a_1\times \prod_{i=1}^{n} b^i_{\delta_{i,l}} \to d_1\}_{1\leq l\leq n} \\		
		y_0&: d_0\times \prod_{i=1}^{n}\prod_{j=1}^{k_i} c^{i,j}_0\to q_0, & \{y_{1,l,m}: d_1\times \prod_{i=1}^{n}\prod_{j=1}^{k_i}c^{i,j}_{\delta_{i,l}\cdot\delta_{j,m}} \to q_1\}_{\substack{1\leq l \leq n,\\ 1\leq m\leq k_l}}
	\end{flalign*}
	induced by the previous Lemma, make the following squares commute for all $l\in\{1,\ldots, n\}$ and $m\in \{1,\ldots k_j\}$. 
	\begin{equation}\label{first-square}
		\begin{tikzcd}
			{a_0\times (\prod_{i=1}^n (b^i_0 \times (\prod_{j=1}^{k_i} c_0^{i,j})))} && {a_0\times (\prod_{i=1}^n e_0^i)} \\
			{(a_0\times (\prod_{i=1}^n b^i_0)) \times (\prod_{i=1}^n \prod_{j=1}^{k_i} c^{i,j}_0)} \\
			{d_0\times (\prod_{i=1}^n \prod_{j=1}^{k_i} c^{i,j}_0)} && {q_0}
			\arrow["{1\times (\prod_{i=1}^n f^i_0)}", from=1-1, to=1-3]
			\arrow["coh"', from=1-1, to=2-1]
			\arrow["{x_0\times 1}"', from=2-1, to=3-1]
			\arrow["{y_0}"', from=3-1, to=3-3]
			\arrow["{g_0}", from=1-3, to=3-3]
		\end{tikzcd}
	\end{equation}
	
	\begin{equation}\label{second-square}
		\begin{tikzcd}
			{a_1\times (\prod_{i=1}^n (b^i_{\delta_{i,l}}\times (\prod_{j=1}^{k_i} c^{i,j}_{\delta_{i,l}\cdot\delta_{j,m}})  ))}	 &&& {a_1\times (\prod_{i=1}^n e^i_{\delta_{i,l}})} \\
			{(a_1\times (\prod_{i=1}^n b^i_{\delta_{i,l}}))\times \prod_{i=1}^n \prod_{j=1}^{k_i} c^{i,j}_{\delta_{i,l}\cdot\delta_{j,m}}} \\
			{d_1\times (\prod_{i=1}^n \prod_{j=1}^{k_i} c^{i,j}_{\delta_{i,l}\cdot\delta_{j,m}})} &&& {q_1}
			\arrow["{\substack{1\times (f^1_0\times \ldots\times f^l_{1,m}\times\ldots\times f_0^n)}}", from=1-1, to=1-4]
			\arrow["coh"', from=1-1, to=2-1]
			\arrow["{x_{1,l}\times 1}"', from=2-1, to=3-1]
			\arrow["{y_{{1,l,m}}}"', from=3-1, to=3-4]
			\arrow["{g_{1,l}}", from=1-4, to=3-4]
		\end{tikzcd}
	\end{equation}
\end{lemma}

Using these two results, we can explicitly spell out what a $([\to,\category{Set}],\times, \fun)$-multicategory should look like.

\begin{lemma}\label{unwrapped}
	A $([\to,\category{Set}],\times, \fun)$-multicategory, $\bb{C}$ uniquely determines and is determined by the following data:
	\begin{itemize}
		\item A set of objects, $\ob{\bb{C}}$
		\item For every sequence of objects, $\Gamma$ and an object $B$, an arrow
		\[
		\Mhom{\bb{C}}{\Gamma}{B}_\diamond: \Mhom{\bb{C}}{\Gamma}{B}_0 \to \Mhom{\bb{C}}{\Gamma}{B}_1
		\]
		in the category $\category{Set}$. We denote terms in $\Mhom{\bb{C}}{\Gamma}{B}_0$ with blue colored symbols, $\color{blue}t,u,v,\ldots$, and denote terms in $\Mhom{\bb{C}}{\Gamma}{B}_1$ with symbols in red, $\color{red} {t},{u},{v},\ldots$. Further, the action of the arrow will be notated as: $\color{blue}t\color{black}\mapsto \color{red}\overline{\color{blue}{t}}$. 
		\[
		id_A :1\to \Mhom{\bb{C}}{A}{A}_0 
		\]
		or equivalently some $\color{blue}id_A\color{black}\in \Mhom{\bb{C}}{A}{A}_0$.
		\item For every $A_1,\ldots, A_n,B$ and $\Delta_1,\ldots, \Delta_n$, a \textit{pure} composition operation,
		
		\begin{align*}
			&psub:	\Mhom{\bb{C}}{A_1,\ldots, A_n}{B}_0\times (\prod_{i=1}^{n}\Mhom{\bb{C}}{\Delta_i}{A_i})_0 \\ &\qquad\qquad\qquad\qquad\qquad\qquad\qquad\to \Mhom{\bb{C}}{\Delta_1,\ldots, \Delta_n}{B}_0\\
			&\qquad\qquad\qquad\qquad\quad \color{blue} t\color{black}, \color{blue} u_1\color{black},\ldots, \color{blue} u_n\color{black} \mapsto \bterms{t}{u_1}{u_n}
		\end{align*}
		\item  For every $A_1,\ldots, A_n,B$ and $\Delta_1,\ldots, \Delta_n$, an operation,
		
		\begin{align*}
			&	esub_j^n:\Mhom{\bb{C}}{A_1,\ldots, A_n}{B}_1\times (\prod_{i=1}^{n}\Mhom{\bb{C}}{\Delta_i}{A_i})_{\delta_{i,j}}\\
			& \qquad\qquad\qquad\qquad\qquad\qquad\qquad \to \Mhom{\bb{C}}{\Delta_1,\ldots, \Delta_n}{B}_1\\
			& \qquad\qquad\quad \;\;\; \color{red} t \color{black}, \color{blue} u_1\color{black},\ldots,\color{red}u_j\color{black},\ldots, \color{blue} u_n \color{black}\mapsto \mterms{t}{u_1}{u_j}{u_n}
		\end{align*}
		for every $j \in \{1,\ldots, n\}$.
	\end{itemize}
	such that the following conditions hold (whenever the expressions make sense): 
	
	\begin{enumerate}
		\item Composition coherence: 
		\begin{enumerate}
			\item For all pure terms, $\color{blue}t\color{black}, \color{blue} u_1\color{black},\ldots, \color{blue} u_n$  and for any $j\in\{1,\ldots, n\}$, the following equation always holds:	
			\begin{align*}
				\color{red}\overline{\bterms{t}{u_1\color{black},\ldots,\color{blue} u_j}{u_n}} = \mterms{\overline{\color{blue}t}}{u_1}{\overline{\color{blue} u_j}}{u_n}
			\end{align*}
			\item For an effectful term, $\color{red} t$ and pure terms $\color{blue}u_1\color{black},\ldots,\color{blue} u_n$, and any
			$j,k\in \{1,\ldots, n\}$, the following equality holds:
			\begin{equation*}
				\mterms{t}{u_1}{\overline{\color{blue}u_j}\color{black},
					\ldots, \color{blue}u_k}{u_n} = \mterms{t}{u_1}{\color{blue}u_j\color{black},\ldots,\color{red}\overline{\color{blue}u_k}}{u_n}
			\end{equation*}
		\end{enumerate}
		\item Left unitality: 
		\begin{enumerate}
			\item For any pure term, $\color{blue} u$, 
			\begin{align*}
				\color{blue} id_A\color{black}(\color{blue}u\color{black}) = \color{blue}u
			\end{align*}
			\item For any effectful term $\color{red} u$,
			\begin{align*}
				\color{red}	\overline{\color{blue} id_A}\color{black}[ \color{red}u\color{black}] = \color{red} u
			\end{align*}
		\end{enumerate}

		\item Right unitality:
		\begin{enumerate}
			\item For any pure term, $\color{blue}t$,
			\begin{align*}
				\bterms{t}{id_{A_1}}{id_{A_n}} = \color{blue} t
			\end{align*}
			\item For an effectful term, $\color{red} t$, and any $j\in \{1,\ldots, n\}$,			\begin{align*}
				\mterms{t}{id_{A_1}}{\overline{\color{blue}id_{A_j}}}{id_{A_n}} = \color{red}t
			\end{align*}
		\end{enumerate}
		
		\item Associativity
		\begin{enumerate}
			\item For pure terms, $\color{blue} t\color{black},\{\color{blue} u_i\color{black}\}_{1\leq i\leq n}, \{\color{blue} v^j_i\color{black}\}_{\substack{1\leq i\leq n\\ 1\leq j \leq k_i}}$,
			\begin{align*}
				&	\bterms{t}{\bterms{u_1}{v_1^1}{v^{k_1}_1}}{\bterms{u_n}{v_1^n}{v_n^{k_n}}} 
				\\& = \bterms{\bterms{t}{u_1}{u_n}}{v^1_1\color{black},\ldots, \color{blue} v^{k_1}_1}{v^1_n\color{black},\ldots,\color{blue}v^{k_n}_n}
			\end{align*}

			\item Choose some $l\in\{1,\ldots, n\}$ and $m\in \{1,\ldots, k_l\}$. Then, for effectful terms, $\color{red}t\color{black}, \color{red} u_l \color{black}, \color{red} v^m_l $ and pure terms  $\{\color{blue} u_i\color{black}\}_{\substack{1\leq i\leq n,\\ i\neq l}}$, $\{\color{blue}v^j_i\color{black}\}_{\substack{1\leq i\leq n,\\ 1\leq j\leq k_i, \\(i,j)\neq (l,m)}}$,
			\begin{center}
				\begin{align*}
					&	\mterms{t}{\bterms{u_1}{v^1_1}{v_1^{k_1}}}{\mterms{u_l}{v^1_l}{v^m_l}{v^{k_l}_{l}}}{\bterms{u_n}{v_n^1}{v_n^{k_n}}}\\
					& = \mterms{\mterms{t}{u_1}{u_l}{u_n}}{v^1_1\color{black},\ldots,\color{blue} v^{k_1}_1}{\color{blue} v^1_l,\ldots, \color{red}v^m_l\color{black},\ldots, \color{blue}v^{k_l}_l}{v^1_n\color{black},\ldots,\color{blue}v^{k_n}_n
					}
				\end{align*}
			\end{center}
		\end{enumerate}

	\end{enumerate}

\end{lemma}
\begin{proof}
	The fact that hom-sets in the multicategory determine arrows in \category{Set} is clear. The existence and composition coherence of the operations $psub$ and $esub_j$ follow from Lemma \ref{unwrap-1}. The unitality axioms are easily verified and the associativity conditions follow from Lemma \ref{unwrap-2}.
\end{proof}

Finally, we unwrap the notion of a morphism between $[\to,\category{Set}]$-multicategories, and this will help us define functors between the respective categories.

\begin{lemma}\label{unwrap-3}
	Giving a morphism of $([\to,\category{Set}],\times, \fun)$-multicategories,  $f: \bb{C}\to \bb{D}$ is equivalent to giving maps,
	\begin{align*}
		f^i_{\Gamma;A}: \Mhom{\bb{C}}{\Gamma}{A}_i \to \Mhom{\bb{D}}{f\Gamma}{fA}_i
	\end{align*}
	for $i=1,2$ such that the following conditions are met:
	
	\begin{enumerate}
		\item Values and computations get mapped sensibly: 
		\begin{equation}\label{coh-diag}
			\begin{tikzcd}[cramped]
				{\Mhom{\bb{C}}{\Gamma}{A}_0} && {\Mhom{\bb{C}}{\Gamma}{A}_1} \\
				{\Mhom{\bb{D}}{f\Gamma}{fA}_0} && {\Mhom{\bb{D}}{f\Gamma}{fA}_1}
				\arrow["{\Mhom{\bb{C}}{\Gamma}{A}_\diamond}", from=1-1, to=1-3]
				\arrow["{\Mhom{\bb{D}}{f\Gamma}{fA}_\diamond}", from=2-1, to=2-3]
				\arrow["{f^0_{\Gamma;A}}"{description}, from=1-1, to=2-1]
				\arrow["{f^1_{\Gamma;A}}"{description}, from=1-3, to=2-3]
			\end{tikzcd}
		\end{equation}
		Equationally, if $\color{blue}t\color{black}:\Gamma\to A$, then, 
		\begin{equation}\label{coh-eq}
			f^1(\color{red}\overline{\color{blue} t}\color{black}) = \color{red}\overline{\color{black}f^0(\color{blue}t\color{black})}\color{black}
		\end{equation}

		\item Identities are preserved:
		\begin{equation}\label{id-diag}
			\begin{tikzcd}[cramped]
				{\Mhom{\bb{C}}{A}{A}_0} & {\Mhom{\bb{D}}{fA}{fA}_0} \\
				I
				\arrow["{id_A}", from=2-1, to=1-1]
				\arrow["{f^0_{A;A}}", from=1-1, to=1-2]
				\arrow["{id_{FA}}"', from=2-1, to=1-2]
			\end{tikzcd}
		\end{equation}
		Or equationally, for every $A$,
		\begin{equation}\label{id-eq}
			f^0(\color{blue}id_A\color{black}) = \color{blue} id_{fA} \color{black}
		\end{equation}
		
		\item Pure substitution is preserved: 
		
		\begin{equation}\label{psub-diag}
			\begin{tikzcd}[cramped]
				{\substack{\Mhom{\bb{C}}{A_1,\ldots, A_n}{B}_0\times \prod_{i=1}^n \Mhom{\bb{C}}{\Delta_i}{A_i}_0}} && {\substack{\Mhom{\bb{D}}{fA_1,\ldots, fA_n}{B}_0\times \prod_{i=1}^n \Mhom{\bb{D}}{f\Delta_i}{fA_i}_0}} \\
				{\substack{\Mhom{\bb{C}}{\Delta_1,\ldots, \Delta_n}{B}_0}} && {\substack{\Mhom{\bb{D}}{f\Delta_1,\ldots, f\Delta_n}{fB}_0}}
				\arrow["{\substack{f^0_{A_1,\ldots,A_n;B}\\\times  \prod_{i=1}^n f_{\Delta_i;A_i}^0}}", from=1-1, to=1-3]
				\arrow["{\substack{psub^{\bb{C}}}}", from=1-1, to=2-1]
				\arrow["{\substack{f^0_{\Delta_1,\ldots, \Delta_n;B}}}", from=2-1, to=2-3]
				\arrow["{\substack{psub^{\bb{D}}}}", from=1-3, to=2-3]
			\end{tikzcd}
		\end{equation}
		For pure  terms, $\color{blue}t\color{black}: A_1,\ldots, A_n\to B$ and $\{\color{blue}u_i\color{black}: 
		\Delta_i\to A_i\}_{1\leq i\leq n}$, 
		\begin{equation}\label{psub-eq}
			f^0( \bterms{t}{u_1}{u_n} ) = \bterms{\color{black}f^0(\color{blue}t\color{black})}{\color{black} f^0(\color{blue}u_1\color{black})}{\color{black} f^0(\color{blue}u_n\color{black})}
		\end{equation}

		\item Effectful substitution is preserved: 
		\begin{equation}\label{esub-diag}
			\begin{tikzcd}[cramped]
				{\substack{\Mhom{\bb{C}}{A_1,\ldots, A_n}{B}_1\times \prod_{i=1}^n \Mhom{\bb{C}}{\Delta_i}{A_i}_{\delta_{i,j}}}} && {\substack{\Mhom{\bb{D}}{fA_1,\ldots, fA_n}{B}_1\times \prod_{i=1}^n \Mhom{\bb{D}}{f\Delta_i}{fA_i}_{\delta_{i,j}}}} \\
				{\substack{\Mhom{\bb{C}}{\Delta_1,\ldots, \Delta_n}{B}_1}} && {\substack{\Mhom{\bb{D}}{f\Delta_1,\ldots, f\Delta_n}{fB}_1}}
				\arrow["{\substack{f^1_{A_1,\ldots,A_n;B}\\\times  \prod_{i=1}^n f_{\Delta_i;A_i}^0 | f_{\Delta_j;A_j}^1}}", from=1-1, to=1-3]
				\arrow["{\substack{esub_j^{\bb{C}}}}", from=1-1, to=2-1]
				\arrow["{\substack{f^1_{\Delta_1,\ldots, \Delta_n;B}}}", from=2-1, to=2-3]
				\arrow["{\substack{esub_j^{\bb{D}}}}", from=1-3, to=2-3]
			\end{tikzcd}
		\end{equation}
		for every $j\in\{1,\ldots, n\}$. This translates to saying that for effectful terms $\color{red}t\color{black}: A_1,\ldots, A_n\to B$, $\color{red} u_j\color{black}: \Delta_j\to A_j$ and a collection of pure terms, $\{\color{blue}u_i\color{black}\}_{1\leq i\leq n, i\neq j}$, 
		\begin{equation}\label{esub-eq}
			f^1(\mterms{t}{u_1}{u_j}{u_n}) = \mterms{\color{black}f^1(\color{red}t\color{black})}{\color{black} f^0(\color{blue} u_1\color{black} )}{\color{black} f^1(\color{red} u_j\color{black} )}{\color{black} f^0(\color{blue} u_n\color{black} )}
		\end{equation}
	\end{enumerate}
	
\end{lemma}

\subsection{From Duoidally Enriched Multicategories to Effectful Multicategories }
\label{6.2.2}

With the characterizations obtained, we are ready to show that every $([\to,\category{Set}],\times,\fun)$-multicategory can be made into an effectful one.

\begin{proposition}
	Every $([\to,\category{Set}],\times, \fun)$-multicategory $\bb{C}$ determines an effectful category, $F\bb{C}$.
\end{proposition}
\begin{proof}
	The construction is natural. We define a multicategory, $F\bb{C}_0$ and a premulticategory $F\bb{C}_1$ as having the same objects as $\bb{C}$. The homsets are given as:
	\begin{align*}
		\Mhom{F\bb{C}_0}{\Gamma}{B} &: = \Mhom{\bb{C}}{\Gamma}{B}_0\\
		\Mhom{F\bb{C}_1}{\Gamma}{B} &: = \Mhom{\bb{C}}{\Gamma}{B}_1
	\end{align*}
	
	For an object $A$, the identity in $F\bb{C}_0$ is the element $\color{blue} id_A \color{black} \in \Mhom{\bb{C}}{A}{A}_0$, and the identity in $F\bb{C}_1$ is $\color{red}\overline{\color{blue} id_A} \color{black}\in \Mhom{\bb{C}}{A}{A}_1$. 
	
	Composition in $F\bb{C}_0$ is defined as:
	\begin{align*}
		\Mhom{F\bb{C}_0}{A_1,\ldots,A_i,\ldots A_n}{B} &\times \Mhom{F\bb{C}_0}{\Delta}{A_i}\\
		&\to \Mhom{F\bb{C}_0}{A_1,\ldots, A_{i-1},\Delta,A_{i+1},\ldots A_n}{B}\\
		\color{blue} t \color{black},\color{blue} u\color{black} 	&\mapsto \bterms{t}{id_{A_1}\color{black},\ldots, \color{blue} u}{id_{A_n}}
	\end{align*}
	
	Composition in $F\bb{C}_1$ is defined as:
	\begin{align*}
		\Mhom{F\bb{C}_1}{A_1,\ldots,A_i,\ldots A_n}{B} &\times \Mhom{F\bb{C}_1}{\Delta}{A_i}\\
		&\to \Mhom{F\bb{C}_1}{A_1,\ldots, A_{i-1},\Delta,A_{i+1},\ldots A_n}{B}\\
		\color{red} t \color{black},\color{red} u\color{black} 	&\mapsto \mterms{t}{id_{A_1}}{u}{id_{A_n}}
	\end{align*}
	
	Finally, the centre-preserving premulticategory morphism, 
	\[
	F\bb{C}_J: F\bb{C}_0\to F\bb{C}_1
	\]
	is the identity on objects morphism which acts on the multihom sets in the obvious way:
	\begin{align*}
		\Mhom{F\bb{C}_0}{\Gamma}{B}_0&\to \Mhom{F\bb{C}_0}{\Gamma}{B}_1\\
		\color{blue} t \color{black} &\mapsto \color{red}\overline{\color{blue} t}
	\end{align*}
	
	Checking that $F\bb{C}_0$ defines a multicategory is straightforward. Unitality axioms in $F\bb{C}_1$ are also straightforward.
	We check that  substitution in $F\bb{C}_1$ is associative: Let $\color{red}t\color{black}: A_1,\ldots, A_i,\ldots, A_n\to B$, $\color{red}u\color{black}: B_1,\ldots, B_j, \ldots, B_m\to A_i$ and $\color{red}v\color{black}: C_1,\ldots, C_p\to B_j$ be effectful terms. Then, from the associativity condition in Lemma \ref{unwrapped}:
	\begin{align*}
		&\mterms{t}{ \color{blue}id_{A_1} \color{black}(\color{blue}id_{A_1}\color{black}) }{\mterms{u}{id_{B_1}}{v}{id_{B_m}}}{id_{A_n}\color{black}(\color{blue}id_{A_n}\color{black})}\\
		& = \mterms{\mterms{t}{id_{A_1}}{u}{id_{A_n}} }{id_{A_1}}{\color{blue} id_{B_1}\color{black},\ldots, \color{red} v\color{black},\ldots, \color{blue}id_{B_m}}{id_{A_n}}
	\end{align*}
	which shows associativity in $F\bb{C}_1$. That all maps in $F\bb{C}_0$ are central follows from the associativity axiom in Lemma \ref{unwrapped}.

	Note that $F\bb{C}_J$ obviously preserves identities. To see why it preserves composition, we use composition coherence (a) from Lemma \ref{unwrapped}:
	\begin{align*}
		\color{red}\overline{\bterms{t}{id_{A_1}\color{black},\ldots, \color{blue} u}{id_{A_n}}}
		& = \mterms{\overline{\color{blue}t}}{id_{A_1}}{\overline{\color{blue} u}}{id_{A_n}}
	\end{align*}
	This shows that $F\bb{C}_J$ is indeed a premulticategory morphism.

	For preservation of centrality: let $\color{red}t\color{black}: A_1,\ldots, A_i, \ldots, A_j,\ldots, A_n\to B, \color{red}v\color{black}: C_1,\ldots, C_m\to A_j$ be  effectful terms and let $\color{blue} u\color{black}: D_1,\ldots, D_p\to A_i$ be a pure term. We show that $\color{red}\overline{\color{blue} u}$ is central in $F\bb{C}_1$, using the axioms of Lemma \ref{unwrapped}:
	\begin{align*}
		&	\mterms{\mterms{t}{id_{A_1}}{\color{blue} id_{A_i}\color{black},\ldots, \color{red} v}{id_{A_n}}}{id_{A_1}}{\color{red}\overline{\color{blue} u}\color{black},\ldots, \color{blue} id_{C_1}\color{black},\ldots, \color{blue} id_{C_m} }{id_{A_n}}\\
		&=	\mterms{\mterms{t}{id_{A_1}}{\color{blue} id_{A_i}\color{black},\ldots, \color{red} v}{id_{A_n}}}{id_{A_1}}{{\color{blue} u}\color{black},\ldots, \color{blue} \color{red}{\overline{\color{blue}{id_{C_1}}}}\color{black},\ldots, \color{blue} id_{C_m} }{id_{A_n}}\\
		&=		\mterms{t}{id_{A_1}\color{black}(\color{blue}id_{A_1}\color{black})\color{black},\ldots, \color{blue} id_{A_i}\color{black}(\color{blue}u\color{black})}{\color{red}v \color{black}[\color{red} \overline{\color{blue}id_{C_1}}\color{black},\color{blue} id_{C_2}\color{black},\ldots, \color{blue} id_{C_m}  \color{black}]}{id_{A_n}\color{black}(\color{blue}id_{A_n}\color{black})}\\
		&=\mterms{t}{id_{A_1}\color{black}(\color{blue}id_{A_1}\color{black})\color{black},\ldots, \color{blue} \color{blue}u\color{black}}{v}{id_{A_n}\color{black}(\color{blue}id_{A_n}\color{black})}\\
		&=	\mterms{t}{id_{A_1}\color{black}(\color{blue}id_{A_1}\color{black})\color{black},\ldots, \color{blue} \color{black}\color{blue}u\color{black}(\color{blue}id_{D_1}\color{black},\ldots,\color{blue} id_{D_p}\color{black})\color{black}}{\overline{\color{blue}id_{A_j}}\color{black}[\color{red}v\color{black}]}{id_{A_n}\color{black}(\color{blue}id_{A_n}\color{black})}\\
		&=\mterms{\mterms{t}{id_{A_1}\color{black},\ldots, \color{blue} u}{\overline{\color{blue}id_{A_j}}}{id_{A_n}}}{id_{A_1}\color{black},\ldots, \color{blue}id_{D_1}\color{black},\ldots, \color{blue}id_{D_p}}{v}{id_{A_n}}\\
		&=\mterms{\mterms{t}{id_{A_1}\color{black},\ldots, \color{red}\overline{\color{blue}u}}{{\color{blue}id_{A_j}}}{id_{A_n}}}{id_{A_1}\color{black},\ldots, \color{blue}id_{D_1}\color{black},
			\ldots, \color{blue} id_{D_p}}{v}{id_{A_n}}\\
	\end{align*} 
	which shows that it does not matter in which order $\color{blue}u$ and $\color{red} v$ are substituted in $\color{red} t$.
\end{proof}

Now we show that the construction $F$ above extends to a functor between the obvious categories.

\begin{lemma}\label{functor-f}
	The mapping $F$ defined above extends to a functor, 
	\[
	F: ([\to,\category{Set}],\times, \fun)\category{-MultiCat}\to \category{EffMultiCat}
	\]
\end{lemma}
\begin{proof}
	We need to define the action of $F$ on maps $f: \bb{C}\to \bb{D}$. Define the maps, $Ff_0: F\bb{C}_0\to F\bb{D}_0$ and $Ff_1: F\bb{C}_1\to F\bb{D}_1$ to act on objects as 
	\begin{align*}
		Ff_0(A) = Ff_1(A) : = f(A)
	\end{align*}
	and on hom-objects just like $f^0$ and $f^1$, i.e,
	\[\begin{tikzcd}
		{\Mhom{F\bb{C}}{\Gamma}{A}_i} && {\Mhom{F\bb{D}}{f\Gamma}{fA}_i} \\
		{\Mhom{\bb{C}}{\Gamma}{A}_i} && {\Mhom{\bb{D}}{f\Gamma}{fA}_i}
		\arrow["{Ff_{i_{\Gamma;A}}}", dashed, from=1-1, to=1-3]
		\arrow["{=}"{description}, from=2-1, to=1-1]
		\arrow["{=}"{description}, from=2-3, to=1-3]
		\arrow["{f_{\Gamma;A}^i}", from=2-1, to=2-3]
	\end{tikzcd}\]
	
	Both $Ff_0$ and $Ff_1$ define premulticategory morphisms: let $\color{blue} t\color{black}: A_1,\ldots, A_n\to B$ and $\color{blue} u\color{black}: B_1,\ldots, B_m\to A_i$ be effectful terms in $F\bb{C}_0$. Then, 
	\begin{align*}
		f^0(\bterms{t}{id_{A_1}\color{black},\ldots, \color{blue}u}{id_{A_n}}) &= \bterms{\color{black}f^0(\color{blue}t\color{black})}{\color{black}f^0(\color{blue}id_{A_1}\color{black})\color{black},\ldots, \color{blue} \color{black}f^0(\color{blue}u\color{black})}{\color{black}f^0(\color{blue}id_{A_n}\color{black})}\\
		& =\bterms{\color{black}f^0(\color{blue}t\color{black})}{id_{fA_1}\color{black},
			\ldots, f^0(\color{blue}u\color{black})}{id_{fA_n}}
	\end{align*}
	using \ref{id-eq} and \ref{psub-eq}, and this shows that $Ff_0$ preserves composition. 
	And similarly, if $\color{red}t \color{black}: A_1,\ldots, A_n\to B$ and $\color{red} u \color{black}: B_1,\ldots, B_m\to A_i$ are effectful terms, we have
	\begin{align*}
		f^1(\mterms{t}{id_{A_1}}{u}{id_{A_n}})  &=\mterms{\color{black}f^1(\color{red}t\color{black})}{\color{black}f^0(\color{blue}id_{A_1}\color{black})}{\color{black}f^1(\color{red}u\color{black})}{\color{black}f^0(\color{blue}id_{A_n}\color{black})}\\
		&	= \mterms{\color{black}f^1(\color{red}t\color{black})}{id_{fA_1}}{\color{black}f^1(\color{red}u \color{black})}{id_{fA_n}}
	\end{align*}
	using \ref{esub-eq} and \ref{id-eq}, and this proves that $Ff_1$ preserves composition.
	
	It is also easy to see that both $Ff_0$ and $Ff_1$ preserve identities
	\begin{align*}
		f^0(\color{blue} id_A\color{black}) = \color{blue} id_{fA}\\
		f^1(\color{red}\overline{\color{blue} id_A}\color{black}) = \color{red}\overline{\color{black}f^0(\color{blue}t\color{black})} \color{black}= \color{red}\overline{\color{blue}id_{fA}}
	\end{align*}
	using \ref{id-eq} and \ref{coh-eq}. Finally, that the square
	\[\begin{tikzcd}[cramped]
		{F\bb{C}_0} && {F\bb{C}_1} \\
		{F\bb{D}_0} && {F\bb{D}_1}
		\arrow["{F\bb{C}_J}", from=1-1, to=1-3]
		\arrow["{Ff_0}", from=1-1, to=2-1]
		\arrow["{F\bb{D}_J}", from=2-1, to=2-3]
		\arrow["{Ff_1}", from=1-3, to=2-3]
	\end{tikzcd}\]
	commutes follows from the fact that that both $Ff_0$ and $Ff_1$ act on objects in the same way, and that \ref{coh-diag} commutes. This shows that $f$ gets mapped to a morphism of effectful multicategories. 
	
	Finally,  $F$ preserves identities and composition and thus is functorial.
\end{proof}

\subsection{From  Effectful Multicategories to Duoidally Enriched Multicategories}

To show that an effectful multicategory defines a $([\to,\category{Set}],\times, \fun)$-multicategory, we first define an iterated substitution operation in an arbitrary premulticategory. 

\begin{definition}
	Let $\bb{C}$ be a premulticategory. Define the operation
	\begin{align*}
		isub:\Mhom{\bb{C}}{\Gamma,A_1,\ldots,A_n,\Gamma'}{B} \times &\Big(\prod_{i=1}^{n} \Mhom{\bb{C}}{A^i_1,\ldots, A_i^{k_i}}{A_i}\Big)\\ &\to \Mhom{\bb{C}}{\Gamma, A^1_1,\ldots, A_1^{k_1},\ldots, A^1_n,\ldots, A^{k_n}_n,\Gamma'}{B} \\
		(t, u_1,\ldots, u_n) &\mapsto t\langle \Gamma,u_1,\ldots, u_n,\Gamma' \rangle	
	\end{align*}
	recursively as:
	\begin{align*}
		t\langle \Gamma, u_1,\Gamma' \rangle &= t[\Gamma,u_1,\Gamma']\\
		t\langle \Gamma, u_1,\ldots, u_n,\Gamma' \rangle & = t\langle \Gamma,u_1,\ldots, u_{n-1},A_n,\Gamma' \rangle[\Gamma,\Delta_1,\ldots, \Delta_{n-1}, u_n, \Gamma'], \qquad (n\geq 2)
	\end{align*}
\end{definition}

Note that in defining this operation, we have chosen an order of evaluation, i.e, we evaluate the left-most term first. As such, the operation $isub$ is in no way canonical and we could have as well chosen $n!-1$ other ways to plug in the terms. Obviously, this operation itself can be iterated, and this observation is behind the following fact. 

\begin{lemma}[Splitting]
	Let $t:\Gamma,A_1,\ldots, A_n,\Gamma'\to B$ be a multimap, and let $\{u_i: \Delta_i \to A_i\}_{1\leq i\leq n}$ be a family of them. Then, 	
	\begin{align*}
		&t\langle \Gamma,u_1,\ldots, u_n,\Gamma' \rangle \\&= t\langle \Gamma,u_1,\ldots, u_j,A_{j+1},\ldots,A_n,\Gamma'  \rangle \langle \Gamma,\Delta_1,\ldots, \Delta_j,u_{j+1},\ldots, u_n,\Gamma' \rangle
	\end{align*}
	for any $1\leq j\leq n$.
\end{lemma}

This derived operation also inherits some amount of associativity, from the premulticategory. This is basically the fact that let-binding is associative in an effectful, call-by-value language.
\begin{lemma}\label{isub-assoc}
	Let $t: \Gamma,A,\Gamma' \to B$, $u: \Lambda,A^1,\ldots, A^{k},\Lambda'\to A$ and $v^i: \Delta^i\to A^i $. Then,
	\[
	t[\Gamma, u, \Gamma'] \langle \Gamma,\Lambda,v^1,\ldots, v^k,\Lambda',\Gamma' \rangle = t [\Gamma,u\langle \Lambda,v^1,\ldots, v^k,\Lambda' \rangle,\Gamma']
	\]
\end{lemma}
\begin{proof}The proof is by induction on $k$. 
\end{proof}

Further, this derived  operation is also preserved under premulticategory morphisms.
\begin{lemma}\label{isub-pres}
	Let $f: \bb{C}\to \bb{D}$ be a premulticategory morphism, and let $t: \Gamma,A_1,\ldots, A_n,\Gamma'\to B$ and $\{u_i: \Delta_i\to A_i\}_{1\leq i\leq n}$ be a collection of terms in $\bb{C}$. Then,
	\begin{align*}
		f(t\langle\Gamma, u_1,\ldots, u_n,\Gamma'\rangle) = f(t)\langle f\Gamma', f(u_1),\ldots, f(u_n),f\Gamma'\rangle
	\end{align*}
\end{lemma}
\begin{proof}
	When $n = 1$, we have
	\[
	f(t\langle \Gamma,u_1,\Gamma'\rangle ) = f(t[\Gamma, u_1,\Gamma']) = f(t)[f\Gamma,f(u_1),f\Gamma'] = f(t)\langle f\Gamma, f(u_1),f\Gamma'\rangle 
	\]
	Assume for induction that 
	\[
	f(t\langle \Gamma, u_1,\ldots, u_n,\Gamma'\rangle) = f(t)\langle f\Gamma, f(u_1),\ldots,f(u_n), f\Gamma'\rangle 
	\]
	Then for $n+1$
	\begin{align*}
		&	f(t\langle \Gamma,u_1,\ldots, u_n,u_{n+1},\Gamma'\rangle) \\& = f(t\langle \Gamma, u_1,\ldots, u_n,A_{n+1},\Gamma'\rangle[\Gamma, \Delta_1,\ldots, \Delta_n, u_{n+1},\Gamma'])\\
		& = f(t\langle \Gamma, u_1,\ldots, u_n,A_{n+1},\Gamma'\rangle)[f\Gamma,f\Delta_1,\ldots, f\Delta_n, f(u_{n+1}), f\Gamma']\\
		& = f(t)\langle f\Gamma, f(u_1),\ldots, f(u_n), fA_{n+1},f\Gamma'\rangle [f\Gamma, f\Delta_1,\ldots, f\Delta_n, f(u_{n+1}), f\Gamma']\\
		& = f(t)\langle f\Gamma, f(u_1),\ldots, f(u_n), f(u_{n+1}), f\Gamma'\rangle
	\end{align*}
\end{proof}

Recall the notion of a central term (Definition \ref{central-term}). To differentiate them, we notate them in the color blue, following the convention in the previous section. A term, $\color{blue} u\color{black}: \Gamma\to A$ is said to be central if for all other terms, $t: \Delta_1,A,\Delta_2,B,\Delta_3\to C$, $t': \Delta_1,B,\Delta_2,A,\Delta_3\to C$ and $v: \Gamma'\to B$, the following equations hold:
\begin{align*}
	&t[\Delta_1, \color{blue} u \color{black}, \Delta_2,B,\Delta_3] [\Delta_1,\Gamma, \Delta_2,v,\Delta_3]\\
	& = t[\Delta_1,A,\Delta_2,v,\Delta_3][\Delta_1,\color{blue}u\color{black}, \Delta_2,\Gamma',\Delta_3]
\end{align*}
\begin{align*}
	&t'[\Delta_1,B,\Delta_2,\color{blue}u\color{black},\Delta_3][\Delta_1,v,\Delta_2,\Gamma,\Delta_3]\\
	& = t'[\Delta_1,v,\Delta_2,A,\Delta_3][\Delta_1,\Gamma',\Delta_2,\color{blue}u\color{black},\Delta_3]
\end{align*}

This property also gets lifted to this iterated operation: plugging in a group of central terms followed by plugging in a group of arbitrary terms is the same as plugging in the group of arbitrary terms followed by plugging in the group of central terms. The proof is by induction, and we split it into two lemmas. The first deals with the base case.

\begin{lemma}\label{un-cent}
	Let $t: \Gamma, A_1,\Lambda, A_2,\ldots, A_n,\Gamma'\to B$ and $u_i: \Delta_i\to A_i$ be multimaps for $1\leq i\leq n$. Assume that $\blue{u_1}$ is central. Then, we have an equality:
	\begin{align*}
		&t[\Gamma, \blue{u_1},\Lambda, A_2,\ldots, A_n,\Gamma'] \langle \Gamma, \Delta_1,\Lambda, u_2,,\ldots, u_n,\Gamma' \rangle\\
		& = t\langle \Gamma, A_1,\Lambda, u_2,\ldots, u_n,\Gamma' \rangle [\Gamma, \blue{u_1},\Lambda, \Delta_2,\ldots, \Delta_n,\Gamma']
	\end{align*}
	\begin{proof}
		Induct on $n$: When $n = 2$, we have
		\begin{align*}
			&t[\Gamma,\blue{u_1},\Lambda,A_2,\Gamma']\langle \Gamma,\Delta_1,\Lambda, u_2,\Gamma' \rangle\\
			& = \text{By def}\\
			& t[\Gamma,\blue{u_1},\Lambda, A_2,\Gamma'][\Gamma,\Delta_1,\Lambda,u_2,\Gamma']\\
			& =\text{By centrality of $\blue{u_1}$}\\
			& t[\Gamma,A_1,\Lambda,u_2,\Gamma'][\Gamma,\blue{u_1},\Lambda,\Delta_2,\Gamma']\\
			& = \text{By def}\\
			& t\langle \Gamma,A_1,\Lambda,u_2,\Gamma' \rangle[\Gamma,\blue{u_1},\Lambda,\Delta_2,\Gamma']
		\end{align*} 
		Assume true for some $n$. Then, we obtain an equality:
		\begin{align*}
			&	t[\Gamma,\blue{u_1},\Lambda,A_2,\ldots, A_n,A_{n+1},\Gamma']\langle \Gamma,\Delta_1,\Lambda,u_2,\ldots, u_n,A_{n+1},\Gamma' \rangle\\
			& = t[\Gamma,A_1,\Lambda,u_2,\ldots, u_n,A_{n+1},\Gamma'][\Gamma, \blue{u_1},\Lambda,\Delta_2,\ldots, \Delta_n,A_{n+1},\Gamma']
		\end{align*}
		Consequently,
		\begin{align*}
			&t[\Gamma, \blue{u_1},\Lambda,A_2,\ldots, A_{n+1},\Gamma']\langle \Gamma,\Delta_1,\Lambda,u_2,\ldots, u_{n+1},\Gamma' \rangle\\
			& = \{\text{Substitute $r =	t[\Gamma, \blue{u_1},\Lambda,A_2,\ldots, A_{n+1},\Gamma'] $}\}\\
			& r \langle \Gamma,\Delta_1,\Lambda,u_2,\ldots, u_{n+1},\Gamma' \rangle\\
			& = \{\text{By def}\}\\
			&	r\langle \Gamma,\Delta_1,\Lambda,u_2,\ldots, u_n,A_{n+1},\Gamma' \rangle[\Gamma,\Delta_1,\Lambda, \Delta_2,\ldots, \Delta_n, u_{n+1},\Gamma']
		\end{align*}
		Simplifying the first half of this term:
		\begin{align*}
			&r\langle \Gamma,\Delta_1,\Lambda,u_2,\ldots, u_n,A_{n+1},\Gamma' \rangle\\
			& = \{\text{Substituting the value of $r$}\}\\
			&  t[\Gamma, \blue{u_1},\Lambda,A_2,\ldots, A_{n+1},\Gamma'] \langle \Gamma,\Delta_1,\Lambda,u_2,\ldots, u_n,A_{n+1},\Gamma' \rangle\\
			& =\{\text{Using the IH}\}\\
			& t\langle \Gamma,A_1,\Lambda,u_2,\ldots, u_n,A_{n+1},\Gamma' \rangle[\Gamma,\blue{u_1},\Lambda, \Delta_2,\ldots, ,\Delta_n, A_{n+1},\Gamma']\\
			& = \{\text{Substitute $s = t\langle \Gamma,A_1,\Lambda,u_2,\ldots, u_n,A_{n+1},\Gamma' \rangle$}\}\\
			&  s[\Gamma,\blue{u_1},\Lambda, \Delta_2,\ldots, \Delta_n,A_{n+1},\Gamma']
		\end{align*}
		Then, 
		\begin{align*}
			&	r\langle \Gamma,\Delta_1,\Lambda,u_2,\ldots, u_n,A_{n+1},\Gamma' \rangle[\Gamma,\Delta_1,\Lambda, \Delta_2,\ldots, \Delta_n, u_{n+1},\Gamma']\\
			& =  \{\text{Substituting according to the previous equality}\}\\
			&s[\Gamma,\blue{u_1},\Lambda, \Delta_2,\ldots, ,\Delta_n,A_{n+1},\Gamma'][\Gamma,\Delta_1,\Lambda, \Delta_2,\ldots, \Delta_n, u_{n+1},\Gamma']\\
			& = \{\text{Using the fact that $\blue{u_1}$ is central}\}\\
			&  s[\Gamma,A_1,\Lambda,\Delta_2,\ldots,\Delta_n, u_{n+1},\Gamma'][\Gamma, \blue{u_1},\Lambda,\Delta_2,\ldots, \Delta_n,\Delta_{n+1},\Gamma']
		\end{align*}
		However, 
		\begin{align*}
			&	s[\Gamma,A_1,\Lambda,\Delta_2,\ldots, \Delta_n,u_{n+1},\Gamma']\\
			& = \{\text{Substituting the value of $s$}\}\\
			&t\langle \Gamma,A_1,\Lambda,u_2,\ldots, u_n,A_{n+1},\Gamma' \rangle[\Gamma,A_1,\Lambda,\Delta_2,\ldots, \Delta_n,u_{n+1},\Gamma']\\
			& =\{\text{By def}\}\\
			& t\langle \Gamma,A_1,\Lambda,u_2,\ldots, u_{n+1},\Gamma' \rangle
		\end{align*}
		Putting all these together, 
		\begin{align*}
			&t[\Gamma, \blue{u_1},\Lambda,A_2,\ldots, A_{n+1},\Gamma']\langle \Gamma,\Delta_1,\Lambda,u_2,\ldots, u_{n+1},\Gamma' \rangle\\
			& = 	r\langle \Gamma,\Delta_1,\Lambda,u_2,\ldots, u_n,A_{n+1},\Gamma' \rangle[\Gamma,\Delta_1,\Lambda, \Delta_2,\ldots, \Delta_n, u_{n+1},\Gamma']\\
			& = s[\Gamma,A_1,\Lambda,\Delta_2,\ldots,\Delta_n, u_{n+1},\Gamma'][\Gamma, \blue{u_1},\Lambda,\Delta_2,\ldots, \Delta_n,\Delta_{n+1},\Gamma']\\
			& = t\langle \Gamma,A_1,\Lambda,u_2,\ldots, u_{n+1},\Gamma' \rangle[\Gamma, \blue{u_1},\Lambda,\Delta_2,\ldots, \Delta_n,\Delta_{n+1},\Gamma']
		\end{align*}
		as required. 
	\end{proof}

\end{lemma}

This fact can be generalized to the case where the first $j$ terms are all central. 
\begin{lemma}\label{cent-left}
	Let $t: \Gamma,A_1,\ldots, A_j,\Lambda, A_{j+1},\ldots, A_n,\Gamma'$, and $u_i: \Delta_i\to A_i$  be multimaps for $1\leq i\leq n$. Assume that $\blue{u_1},\ldots, \blue{u_j}$ are all central. Then, 
	\begin{align*}
		&t\langle \Gamma, \blue{u_1},\ldots, \blue{u_j},\Lambda,A_{j+1},\ldots, A_n,\Gamma' \rangle\langle \Gamma, \Delta_1,\ldots, \Delta_j, \Lambda, u_{j+1},\ldots, u_n,\Gamma'  \rangle 
		\\& = t\langle \Gamma, A_1,\ldots,A_j,\Lambda, u_{j+1},\ldots, u_n,\Gamma' \rangle\langle \Gamma, \blue{u_1},\ldots, \blue{u_j},\Lambda, \Delta_{j+1},\ldots, \Delta_n,\Gamma' \rangle
	\end{align*}
\end{lemma}
\begin{proof}
	Induct on $j$. The case $j = 1$ was handled in the previous lemma. Assume true for some $j$. Then, 
	\begin{align*}
		&	t\langle \Gamma,\blue{u_1},\ldots, \blue{u_j},A_{j+1},\Lambda,A_{j+2},\ldots, A_n,\Gamma' \rangle \\ 
		&	\qquad\qquad\qquad\qquad\qquad\qquad\langle \Gamma,\Delta_1,\ldots, \Delta_j,A_{j+1},\Lambda,u_{j+2},\ldots, u_n,\Gamma' \rangle\\
		& = t\langle \Gamma,A_1,\ldots, A_j,A_{j+1},\Lambda, u_{j+2},\ldots, u_n,\Gamma' \rangle\\
		&\qquad\qquad\qquad\qquad\qquad\qquad \langle \Gamma, \blue{u_1},\ldots, \blue{u_j}, A_{j+1},\Lambda, \Delta_{j+2},\ldots, \Delta_n, \Gamma' \rangle 
	\end{align*}
	Assume $u_i$ are central for $i\leq j+1$. Then, 
	\begin{align*}
		&	t\langle \Gamma,\blue{u_1},\ldots \blue{u_j},\blue{u_{j+1}},\Lambda,A_{j+2},\ldots, A_n,\Gamma' \rangle \\
		& = \{\text{By def}\}\\
		& t\langle \Gamma,\blue{u_1},\ldots, \blue{u_j},A_{j+1},\Lambda,A_{j+2},\ldots, A_n,\Gamma' \rangle\\
		&\qquad\qquad\qquad\qquad\qquad\qquad
		[\Gamma,\Delta_1,\ldots, \Delta_j,\blue{u_{j+1}},\Lambda,A_{j+2}, \ldots, A_n,\Gamma']\\
		& =\{\text{Setting $r = t\langle \Gamma,\blue{u_1},\ldots, \blue{u_j},A_{j+1},\Lambda,A_{j+2},\ldots, A_n,\Gamma' \rangle$}\}\\
		&  r	[\Gamma,\Delta_1,\ldots, \Delta_j,\blue{u_{j+1}},\Lambda,A_{j+2}, \ldots, A_n,\Gamma']
	\end{align*}
	Then, using the previous lemma, 
	\begin{align*}
		&r	[\Gamma,\Delta_1,\ldots, \Delta_j,\blue{u_{j+1}},\Lambda,A_{j+2}, \ldots, A_n,\Gamma']\\
		&\qquad\qquad\qquad\qquad\qquad\qquad\langle \Gamma, \Delta_1,\ldots, \Delta_j, \Delta_{j+1}, \Lambda,u_{j+2},\ldots, u_n,\Gamma' \rangle\\
		& = r\langle \Gamma, \Delta_1,\ldots, \Delta_j, A_{j+1},\Lambda,u_{j+2},\ldots, u_n,\Gamma' \rangle\\
		&	\qquad\qquad\qquad\qquad\qquad\qquad		
		[\Gamma, \Delta_1,\ldots, \Delta_j, \blue{u_{j+1}},\Lambda, \Delta_{j+2},\ldots, \Delta_n,\Gamma']
	\end{align*}
	Also consider the following chain:
	\begin{align*}
		&r\langle \Gamma, \Delta_1,\ldots, \Delta_j, A_{j+1},\Lambda,u_{j+2},\ldots, u_n,\Gamma' \rangle\\
		& = \{\text{Substituting the value of $r$}\}\\
		& t\langle \Gamma,\blue{u_1},\ldots, \blue{u_j},A_{j+1},\Lambda,A_{j+2},\ldots, A_n,\Gamma' \rangle\\
		& 	\qquad\qquad\qquad\qquad\qquad\qquad\langle \Gamma, \Delta_1,\ldots, \Delta_j, A_{j+1},\Lambda,u_{j+2},\ldots, u_n,\Gamma' \rangle\\
		& = \{\text{Using IH}\}\\
		&  t\langle \Gamma,A_1,\ldots, A_j,A_{j+1},\Lambda, u_{j+2},\ldots, u_n,\Gamma' \rangle\\
		&\qquad\qquad\qquad\qquad\qquad\qquad \langle \Gamma, \blue{u_1},\ldots, \blue{u_j}, A_{j+1},\Lambda, \Delta_{j+2},\ldots, \Delta_n, \Gamma' \rangle \\
		& = \{\text{Setting $s =t\langle \Gamma,A_1,\ldots, A_j,A_{j+1},\Lambda, u_{j+2},\ldots, u_n,\Gamma' \rangle $}\}\\
		&  s\langle \Gamma, \blue{u_1},\ldots, \blue{u_j}, A_{j+1},\Lambda, \Delta_{j+2},\ldots, \Delta_n, \Gamma' \rangle 
	\end{align*}
	Then, 
	\begin{align*}
		&r\langle \Gamma, \Delta_1,\ldots, \Delta_j, A_{j+1},\Lambda,u_{j+2},\ldots, u_n,\Gamma' \rangle\\
		&	\qquad\qquad\qquad\qquad\qquad\qquad		
		[\Gamma, \Delta_1,\ldots, \Delta_j, \blue{u_{j+1}}, \Delta_{j+2},\ldots, \Delta_n,\Gamma']\\
		& = \{\text{From the previous chain of equalities}\}\\
		& s\langle \Gamma, \blue{u_1},\ldots, \blue{u_j}, A_{j+1},\Lambda, \Delta_{j+2},\ldots, \Delta_n, \Gamma' \rangle\\ 
		&	\qquad\qquad\qquad\qquad\qquad\qquad		
		[\Gamma, \Delta_1,\ldots, \Delta_j, \blue{u_{j+1}}, \Delta_{j+2},\ldots, \Delta_n,\Gamma']\\
		& = \{\text{By def}\}\\ 
		&  s\langle \Gamma,\blue{u_1},\ldots, \blue{u_{j+1}},\Lambda,\Delta_{j+2},\ldots, \Delta_n,\Gamma' \rangle
	\end{align*}
	
	Putting everything together:
	\begin{align*}
		&	t\langle \Gamma,\blue{u_1},\ldots \blue{u_j},\blue{u_{j+1}},\Lambda,A_{j+2},\ldots, A_n,\Gamma' \rangle\\
		&	\qquad\qquad\qquad\qquad\qquad\qquad	\langle \Gamma,\Delta_1,\ldots, \Delta_{j+1},\Lambda, u_{j+2},\ldots, u_n,\Gamma' \rangle\\
		& = \{\text{Using the first equality}\}\\
		& 	r	[\Gamma,\Delta_1,\ldots, \Delta_j,\blue{u_{j+1}},\Lambda,A_{j+2}, \ldots, A_n,\Gamma']\\
		&	\qquad\qquad\qquad\qquad\qquad\qquad	\langle \Gamma,\Delta_1,\ldots, \Delta_{j+1},\Lambda, u_{j+2},\ldots, u_n,\Gamma' \rangle\\
		& = \{\text{Using previous lemma}\}\\
		&  r\langle \Gamma, \Delta_1,\ldots, \Delta_j, A_{j+1},\Lambda,u_{j+2},\ldots, u_n,\Gamma' \rangle\\
		&	\qquad\qquad\qquad\qquad\qquad\qquad		
		[\Gamma, \Delta_1,\ldots, \Delta_j, \blue{u_{j+1}},\Lambda, \Delta_{j+2},\ldots, \Delta_n,\Gamma']\\
		& = \{\text{Using the chain of equalities above}\}\\
		&  s\langle \Gamma,\blue{u_1},\ldots, \blue{u_{j+1}},\Lambda,\Delta_{j+2},\ldots, \Delta_n,\Gamma' \rangle\\
		& =\{\text{Using the value of $s$}\}\\
		&  t\langle \Gamma,A_1,\ldots, A_j,A_{j+1},\Lambda, u_{j+2},\ldots, u_n,\Gamma' \rangle\\
		& \qquad\qquad\qquad\qquad\qquad\qquad\langle \Gamma,\blue{u_1},\ldots, \blue{u_{j+1}},\Lambda,\Delta_{j+2},\ldots, \Delta_n,\Gamma' \rangle
	\end{align*}
	as required, and this completes the induction.	
\end{proof}
Basically the same argument yields a similar lemma, where the last $j-n$ terms are central.
\begin{lemma}\label{cent-right}
	Let $t: \Gamma,A_1,\ldots, A_j,\Lambda, A_{j+1},\ldots, A_n,\Gamma'$, and $u_i: \Delta_i\to A_i$  be multimaps. Assume that $\blue{u_j},\ldots, \blue{u_n}$ are all central. Then, 
	\begin{align*}
		&t\langle \Gamma, {u_1},\ldots, {u_j},\Lambda,A_{j+1},\ldots, A_n,\Gamma' \rangle\langle \Gamma, \Delta_1,\ldots, \Delta_j, \Lambda, \blue{u_{j+1}},\ldots, \blue{u_n},\Gamma'  \rangle 
		\\& = t\langle \Gamma, A_1,\ldots,A_j,\Lambda, \blue{u_{j+1}},\ldots, \blue{u_n},\Gamma' \rangle\langle \Gamma, {u_1},\ldots, {u_j},\Lambda, \Delta_{j+1},\ldots, \Delta_n,\Gamma' \rangle
	\end{align*}
\end{lemma}

With these facts, we can prove the following well-known fact that a multicategory defined with the circle-$i$ method gives rise to a multicategory defined with a simultaneous substitution operation.

\begin{lemma}\label{multicat-assoc}
	Let $t: \Gamma,A_1,\ldots, A_n,\Gamma' \to B$, $\blue{u_i}: A^1_i,\ldots, A_i^{k_i} \to A_i$ and $\blue{v^j_i}: \Delta^j_i \to A^j_i$ be multimaps. Assume that all the $\blue{u_i}$ and all the $\blue{v^j_i}$ are central. Then, we have an equality: 
	\begin{align*}
		&	t\langle \Gamma,\blue{u_1}\langle \blue{v^1_1},\ldots, \blue{v^{k_1}_1} \rangle,\ldots, \blue{u_n}\langle \blue{v^1_n},\ldots, \blue{v^{k_n}_n} \rangle,\Gamma' \rangle\\
		& = t\langle \Gamma,\blue{u_1},\ldots, \blue{u_n},\Gamma' \rangle\langle \Gamma,\blue{v^1_1},\ldots, \blue{v^{k_1}_1},\ldots, \blue{v^1_n},\ldots, \blue{v^{k_n}_n},\Gamma'  \rangle
	\end{align*}
\end{lemma}
\begin{proof}
	We induct on $n$. The case for $n=1$ is easily deduced from Lemma \ref{isub-assoc}. Assume this is true for some $n$, so in particular, 
	\begin{align*}
		&	t\langle \Gamma,\blue{u_1}\langle \blue{v^1_1},\ldots, \blue{v^{k_1}_1} \rangle,\ldots, \blue{u_n}\langle \blue{v^1_n},\ldots, \blue{v^{k_n}_n} \rangle,A_{n+1},\Gamma'  \rangle \\
		& = t\langle \Gamma, \blue{u_1},\ldots, \blue{u_n},A_{n+1},\Gamma' \rangle\langle \Gamma, \blue{v^1_1},\ldots, \blue{v^{k_1}_1},\ldots, \blue{v^1_n},\ldots, \blue{v^{k_n}_n},A_{n+1},\Gamma' \rangle
	\end{align*}
	Set ${\theta_i}: = \blue{u_i}\langle \blue{v^1_i},\ldots, \blue{v_i^{k_i}} \rangle$. Then, 
	\begin{align*}
		&t\langle \Gamma, \theta_1,\ldots, \theta_n, \blue{u_{n+1}}\langle \blue{v^1_{n+1}},\ldots, \blue{v^{k_{n+1}}_{n+1}} \rangle,\Gamma' \rangle\\
		& = \{\text{Using Splitting}\}\\
		&t\langle \Gamma, \theta_1,\ldots, \theta_n, A_{n+1},\Gamma' \rangle[\Gamma, \Delta^1_1,\ldots, \Delta^{k_n}_n, \blue{u_{n+1}}\langle \blue{v^1_{n+1}},\ldots, \blue{v^{k_{n+1}}_{n+1}} \rangle,\Gamma']\\
		& = \{\text{Setting $r = t\langle \Gamma, \theta_1,\ldots, \theta_n, A_{n+1},\Gamma' \rangle$}\}\\
		& r[\Gamma, \Delta^1_1,\ldots, \Delta^{k_n}_n, \blue{u_{n+1}}\langle \blue{v^1_{n+1}},\ldots, \blue{v^{k_{n+1}}_{n+1}} \rangle,\Gamma']\\
		& = \{\text{Using Lemma \ref{isub-assoc}}\}\\
		& r[\Gamma, \Delta^1_1,\ldots, \Delta^{k_n}_n, \blue{u_{n+1}},\Gamma']\langle \Gamma,\Delta^1_1,\ldots, \Delta^{k_n}_n, \blue{v^1_{n+1}},\ldots, \blue{v^{k_{n+1}}_{n+1}},\Gamma' \rangle		
	\end{align*}
	Expanding $r$, we get
	\begin{align*}
		&	r[\Gamma, \Delta^1_1,\ldots, \Delta^{k_n}_n, \blue{u_{n+1}},\Gamma']\\
		& = \{\text{Use the value of $r$}\}\\
		&  t\langle \Gamma, \theta_1,\ldots, \theta_n, A_{n+1},\Gamma' \rangle[\Gamma, \Delta^1_1,\ldots, \Delta^{k_n}_n, \blue{u_{n+1}},\Gamma']\\
		& = \{\text{Using the IH}\}\\
		& t\langle \Gamma, \blue{u_1},\ldots, \blue{u_n},A_{n+1},\Gamma' \rangle\langle \Gamma, \blue{v^1_1},\ldots, \blue{v^{k_n}_n},A_{n+1},\Gamma' \rangle[\Gamma, \Delta^1_1,\ldots, \Delta^{k_n}_n, \blue{u_{n+1}},\Gamma']\\
		& =\{\text{Using Lemma \ref{cent-left}}\}\\
		& t\langle \Gamma, \blue{u_1},\ldots, \blue{u_n},A_{n+1},\Gamma' \rangle[\Gamma, A^1_1,\ldots, A^{k_n}_n, \blue{u_{n+1}},\Gamma']\langle \Gamma, \blue{v^1_1},\ldots, \blue{v^{k_n}_n}, A^1_{n+1},\ldots, A^{k_{n+1}}_{n+1},\Gamma' \rangle\\
		& =\{\text{By def}\}\\
		& t\langle \Gamma,\blue{u_1},\ldots,\blue{u_{n+1}},\Gamma' \rangle\langle \Gamma, \blue{v^1_1},\ldots, \blue{v^{k_n}_n}, A^1_{n+1},\ldots, A^{k_{n+1}}_{n+1},\Gamma' \rangle\\
		& = \{\text{Setting $s = t\langle \Gamma,\blue{u_1},\ldots,\blue{u_{n+1}},\Gamma' \rangle$}\}\\
		& s\langle \Gamma, \blue{v^1_1},\ldots, \blue{v^{k_n}_n}, A^1_{n+1},\ldots, A^{k_{n+1}}_{n+1},\Gamma' \rangle
	\end{align*}
	
	Then, 
	\begin{align*}
		& t\langle \Gamma, \theta_1,\ldots, \theta_{n+1},\Gamma' \rangle\\
		& =\{\text{Using the previous two chains of equalities}\}\\
		&  s\langle \Gamma, \blue{v^1_1},\ldots, \blue{v^{k_n}_n}, A^1_{n+1},\ldots, A^{k_{n+1}}_{n+1},\Gamma' \rangle\langle \Gamma,\Delta^1_1,\ldots, \Delta^{k_n}_n, \blue{v^1_{n+1}},\ldots, \blue{v^{k_{n+1}}_{n+1}},\Gamma' \rangle	\\
		& = \{\text{Using Splitting}\}\\
		& s\langle \Gamma, \blue{v^1_1},\ldots, \blue{v^{k_n}_n}, \blue{v^1_{n+1}},\ldots, \blue{v^{k_{n+1}}_{n+1}},\Gamma' \rangle\\
		& =\{\text{Resubstitute the value of $s$}\}\\
		& t\langle \Gamma,\blue{u_1},\ldots,\blue{u_{n+1}},\Gamma' \rangle\langle \Gamma, \blue{v^1_1},\ldots, \blue{v^{k_n}_n}, \blue{v^1_{n+1}},\ldots, \blue{v^{k_{n+1}}_{n+1}},\Gamma' \rangle
	\end{align*}
	which completes the proof for $n+1$.
\end{proof}

Note that the fact that $u_{n+1}$ and $\{v^j_{n+1}\}$ were central did not figure in the proof of the inductive step, and this assumption was only required to make the induction work. So, a stronger version of  the Lemma holds with the stronger assumption that only  $u_1,\ldots, u_{n-1}$ and $v^1_1,\ldots, v^{k_{n-1}}_{n-1}$ are central. We generalize this further to include cases where $u_p$ and $v^1_p,\ldots, v^{k_p}_p$ are not necessarily central, and this will give us the operations $esub_j$. 

\begin{lemma}\label{freyd-assoc}
	Let $t: \Gamma,A_1,\ldots, A_n,\Gamma' \to B$, $u_i: A^1_i,\ldots, A_i^{k_i} \to A_i$ and $v^j_i: \Delta^j_i \to A^j_i$ be multimaps. Choose some $ p \in \{1,\ldots, n\}$ and  and assume that $u_i$ is central for all $i\neq p$ and $v^j_i$ is central except when $i = p$. Then, the following equality still holds:
	\begin{align*}
		&	t\langle \Gamma,\blue{u_1}\langle \blue{v^1_1},\ldots, \blue{v^{k_1}_1} \rangle,\ldots, u_p\langle v^1_p,\ldots, v^{k_p}_p \rangle,\ldots, \blue{u_n}\langle \blue{v^1_n},\ldots, \blue{v^{k_n}_n} \rangle,\Gamma' \rangle\\
		& = t\langle \Gamma,\blue{u_1},\ldots,u_p,\ldots \blue{u_n},\Gamma' \rangle\langle \Gamma',\blue{v^1_1},\ldots, \blue{v^{k_1}_1},\ldots, v^1_p,\ldots, v^{k_p}_p,\ldots, \blue{v^1_n}\ldots, \blue{v^{k_n}_n},\Gamma'  \rangle
	\end{align*}
\end{lemma}
\begin{proof}
	Set $\theta_i:  = u_i\langle v^1_i,\ldots, v_i^{k_i} \rangle$. Then, we compute a series of equalities: The first splits the following term. 
	\begin{equation}\label{t-val}
		\begin{aligned}
			&t\langle \Gamma,\theta_1,\ldots,\theta_p,\ldots \theta_n,\Gamma' \rangle \\
			& = \{\text{Using Splitting}\}\\
			&  t\langle \Gamma, \theta_1,\ldots, \theta_{p}, A_{p+1},\ldots,A_n,\Gamma' \rangle\langle \Gamma, \Delta^1_1,\ldots, \Delta^{k_p}_p, \theta_{p+1},\ldots, \theta_{n},\Gamma' \rangle\\
			& =\{\text{Setting $r =  t\langle \Gamma, \theta_1,\ldots, \theta_{p}, A_{p+1},\ldots,A_n,\Gamma' \rangle$}\}\\
			& r\langle \Gamma, \Delta^1_1,\ldots, \Delta^{k_p}_p, \blue{u_{p+1}}\langle \blue{v^1_{p+1}},\ldots, \blue{v^{k_{p+1}}_{p+}} \rangle,\ldots, \blue{u_{n}}\langle \blue{v^1_n},\ldots, \blue{v^{k_n}_n} \rangle,\Gamma' \rangle\\
			& =\{\text{Use Lemma \ref{multicat-assoc} since $ \blue{u_{p+1}},\ldots, \blue{u_n}, \blue{v^1_{p+1}},\ldots, \blue{v^{k_n}_n}$ are all central}\}\\
			& r\langle \Gamma,\Delta^1_1,\ldots, \Delta^{k_p}_p, \blue{u_{p+1}},\ldots, \blue{u_n},\Gamma' \rangle\langle \Gamma,\Delta^1_1,\ldots, \Delta^{k_p}_p, \blue{v^1_{p+1}},\ldots, \blue{v^{k_n}_n},\Gamma' \rangle
		\end{aligned}
	\end{equation}
	
	We compute the value of the term above piece by piece. The term $r$ is unwrapped:
	\begin{equation}\label{r-val}
		\begin{aligned}
			& r\\
			& = \{\text{Using the value of $r$}\}\\
			& t \langle \Gamma,\theta_1,\ldots, \theta_p, A_{p+1},\ldots, A_n,\Gamma' \rangle\\
			& = \{\text{By def}\}\\
			& t\langle \Gamma,\theta_1,\ldots, \theta_{p-1},A_p,\ldots, A_n,\Gamma' \rangle[\Gamma,\Delta^1_1,\ldots, \Delta^{k_{p-1}}_{p-1}, \theta_p,A_{p+1},\ldots, A_n,\Gamma']\\
			& =  \{\text{Setting $s = t\langle \Gamma,\theta_1,\ldots, \theta_{p-1},A_p,\ldots, A_n,\Gamma' \rangle$}\}\\
			& s		[\Gamma,\Delta^1_1,\ldots, \Delta^{k_{p-1}}_{p-1}, u_p\langle v^1_p,\ldots, v^{k_p}_{p} \rangle,A_{p+1},\ldots, A_n,\Gamma']\\
			& = \{\text{Using Lemma \ref{isub-assoc}}\}\\
			& s[\Gamma,\Delta^1_1,\ldots, \Delta^{k_{p-1}}_{p-1}, u_p, A_{p+1},\ldots, A_n,\Gamma']\langle \Gamma,\Delta^1_1,\ldots, \Delta^{k_{p-1}}_{p-1}, v^1_p,\ldots, v^{k_p}_p, A_{p+1},\ldots, A_n,\Gamma' \rangle\\
		\end{aligned}
	\end{equation}

	The term $s$ is unwrapped:
	\begin{equation}\label{s-val}
		\begin{aligned}
			& s \\
			& = \{\text{By def}\}\\
			& t\langle \Gamma,\blue{u_1}\langle \blue{v^1_1},\ldots, \blue{v^{k_1}_1} \rangle,\ldots, \blue{u_{p-1}}\langle \blue{v^1_{p-1}},\ldots, \blue{v^{k_{p-1}}_{p-1}} \rangle,A_p,\ldots, A_n,\Gamma' \rangle\\
			& =\{\text{Applying Lemma \ref{multicat-assoc} with the fact that $\blue{u_1},\ldots, \blue{u_p}, \blue{v^1_1},\ldots, \blue{v^{k_{p-1}}_{p-1}}$}\}\\
			& t\langle \Gamma,\blue{u_1},\ldots, \blue{u_{p-1}},A_p,\ldots, A_n,\Gamma' \rangle\langle \Gamma, \blue{v^1_1},\ldots, \blue{v^{k_{p-1}}_{p-1}}, A_p,\ldots, A_n,\Gamma' \rangle\\
			& = \{\text{Setting $a  =t\langle \Gamma,\blue{u_1},\ldots, \blue{u_{p-1}},A_p,\ldots, A_n,\Gamma' \rangle$}\}\\
			& a\langle \Gamma, \blue{v^1_1},\ldots, \blue{v^{k_{p-1}}_{p-1}}, A_p,\ldots, A_n,\Gamma' \rangle
		\end{aligned}
	\end{equation}
	
	The first half of $r$ is computed; 
	\begin{equation}
		\begin{aligned}
			&s[\Gamma,\Delta^1_1,\ldots, \Delta^{k_{p-1}}_{p-1}, u_p, A_{p+1},\ldots, A_n,\Gamma']\\
			& = \{\text{Using \ref{s-val}}\}\\
			& a\langle \Gamma, \blue{v^1_1},\ldots, \blue{v^{k_{p-1}}_{p-1}}, A_p,\ldots, A_n,\Gamma' \rangle[\Gamma,\Delta^1_1,\ldots, \Delta^{k_{p-1}}_{p-1}, u_p, A_{p+1},\ldots, A_n,\Gamma']\\
			& = \{\text{Using \ref{cent-left} and the fact that $\blue{v^1_1},\ldots, \blue{v^{k_{p-1}}_{p-1}}$ are central}\}\\
			&  a[\Gamma, A^1_1,\ldots, A^{k_{p-1}}_{p-1}, u_p,A_{p+1},\ldots, A_n,\Gamma']\langle \Gamma, \blue{v^1_1},\ldots, \blue{v^{k_{p-1}}_{p-1}}, A^1_p,\ldots, A^{k_p}_p, A_{p+1},\ldots, A_n,\Gamma' \rangle\\
		\end{aligned}
	\end{equation}
	
	The first half of the previous result is computed:	
	\begin{equation}\label{a-val}
		\begin{aligned}
			& a[\Gamma, A^1_1,\ldots, A^{k_{p-1}}_{p-1}, u_p,A_{p+1},\ldots, A_n,\Gamma']\\
			& = \{\text{Substituting the value of $a$}\}\\
			&t\langle \Gamma,\blue{u_1},\ldots, \blue{u_{p-1}},A_p,\ldots, A_n,\Gamma' \rangle[\Gamma, A^1_1,\ldots, A^{k_{p-1}}_{p-1}, u_p,A_{p+1},\ldots, A_n,\Gamma']\\
			& = \{\text{By def}\}\\
			& t\langle \Gamma,\blue{u_1},\ldots,\blue{u_{p-1}}, u_p,A_{p+1},\ldots, A_n,\Gamma' \rangle
		\end{aligned}
	\end{equation}
	
	We go back to computing the first half of $r$:
	\begin{equation}\label{r-val2}\\
		\begin{aligned}
			&s[\Gamma,\Delta^1_1,\ldots, \Delta^{k_{p-1}}_{p-1}, u_p, A_{p+1},\ldots, A_n,\Gamma']\\
			& = \{\text{Using \ref{s-val}}\}\\
			&  a[\Gamma, A^1_1,\ldots, A^{k_{p-1}}_{p-1}, u_p,A_{p+1},\ldots, A_n,\Gamma']\langle \Gamma, \blue{v^1_1},\ldots, \blue{v^{k_{p-1}}_{p-1}}, A^1_p,\ldots, A^{k_p}_p, A_{p+1},\ldots, A_n,\Gamma' \rangle\\
			& = \{\text{Using \ref{a-val}}\}\\
			& t\langle \Gamma,\blue{u_1},\ldots,\blue{u_{p-1}}, u_p,A_{p+1},\ldots, A_n,\Gamma' \rangle\langle \Gamma, \blue{v^1_1},\ldots, \blue{v^{k_{p-1}}_{p-1}}, A^1_p,\ldots, A^{k_p}_p, A_{p+1},\ldots, A_n,\Gamma' \rangle\\
			& = \{\text{Setting $b = t\langle \Gamma,\blue{u_1},\ldots, \blue{u_{p-1}},u_p,A_{p+1},\ldots, A_n,\Gamma' \rangle$}\}\\
			& b\langle \Gamma, \blue{v^1_1},\ldots, \blue{v^{k_{p-1}}_{p-1}}, A^1_p,\ldots, A^{k_p}_p, A_{p+1},\ldots, A_n,\Gamma' \rangle
		\end{aligned}
	\end{equation}
	
	Using all the information above, we arrive at an expression for $r$:
	\begin{equation}\label{r-val3}
		\begin{aligned}
			& r \\
			& = \{\text{Using \ref{r-val}}\}\\
			&   s[\Gamma,\Delta^1_1,\ldots, \Delta^{k_{p-1}}_{p-1}, u_p, A_{p+1},\ldots, A_n,\Gamma']\langle \Gamma,\Delta^1_1,\ldots, \Delta^{k_{p-1}}_{p-1}, v^1_p,\ldots, v^{k_p}_p, A_{p+1},\ldots, A_n,\Gamma' \rangle\\
			& = \{\text{Using \ref{r-val2}}\}\\
			& b\langle \Gamma, \blue{v^1_1},\ldots, \blue{v^{k_{p-1}}_{p-1}}, A^1_p,\ldots, A^{k_p}_p, A_{p+1},\ldots, A_n,\Gamma' \rangle\langle \Gamma,\Delta^1_1,\ldots, \Delta^{k_{p-1}}_{p-1}, v^1_p,\ldots, v^{k_p}_p, A_{p+1},\ldots, A_n,\Gamma' \rangle\\
			& = \{\text{Using Splitting}\}\\
			& b\langle \Gamma, \blue{v^1_1},\ldots, \blue{v^{k_{p-1}}_{p-1}}, v^1_p,\ldots, v^{k_p}_p, A_{p+1},\ldots, A_n,\Gamma' \rangle
		\end{aligned}
	\end{equation}

	The first half of the main term is computed
	\begin{equation}\label{r-val4}
		\begin{aligned}
			& r \langle \Gamma,\Delta^1_1,\ldots, \Delta^{k_p}_p, \blue{u_{p+1}},\ldots, \blue{u_n},\Gamma' \rangle\\
			& = \{\text{Using \ref{r-val3}}\}\\
			&  b\langle \Gamma, \blue{v^1_1},\ldots, \blue{v^{k_{p-1}}_{p-1}}, v^1_p,\ldots, v^{k_p}_p, A_{p+1},\ldots, A_n,\Gamma' \rangle\langle \Gamma,\Delta^1_1,\ldots, \Delta^{k_p}_p, \blue{u_{p+1}},\ldots, \blue{u_n},\Gamma' \rangle\\
			& = \{\text{Using Lemma \ref{cent-right} and the fact that $\blue{u_{p+1}},\ldots, \blue{u_n}$ are central} \}\\
			& b\langle \Gamma, A^1_1,\ldots, A^{k_p}_p, \blue{u_{p+1}},\ldots, \blue{u_n},\Gamma' \rangle\\
			&\qquad\qquad\qquad\langle \Gamma, \blue{v^1_1},\ldots, \blue{v^{k_{p-1}}_{p-1}}, v^1_p,\ldots, v^{k_p}_p, A^1_{p+1},\ldots, A^{k_{p+1}}_{p+1},\ldots, A^1_n,\ldots, A^{k_n}_n,\Gamma' \rangle
		\end{aligned}
	\end{equation}

	\begin{equation}\label{b-val}
		\begin{aligned}
			&b\langle \Gamma, A^1_1,\ldots, A^{k_p}_p, \blue{u_{p+1}},\ldots, \blue{u_n},\Gamma' \rangle\\
			& = \{\text{Substituting the value of $b$}\}\\
			& t\langle \Gamma,\blue{u_1},\ldots, \blue{u_{p-1}},u_p,A_{p+1},\ldots, A_n,\Gamma' \rangle\langle \Gamma, A^1_1,\ldots, A^{k_p}_p, \blue{u_{p+1}},\ldots, \blue{u_n},\Gamma' \rangle\\
			& = \{\text{Using Splitting}\}\\
			& t\langle \Gamma,\blue{u_1},\ldots, u_p,\ldots, \blue{u_n},\Gamma' \rangle\\
			& = \{\text{Set $c = t\langle \Gamma,\blue{u_1},\ldots, u_p,\ldots, \blue{u_n},\Gamma' \rangle$}\}\\
			& c
		\end{aligned}
	\end{equation}
	
	\begin{equation}\label{r-subval}
		\begin{aligned}
			&r \langle \Gamma,\Delta^1_1,\ldots, \Delta^{k_p}_p, \blue{u_{p+1}},\ldots, \blue{u_n},\Gamma' \rangle\langle \Gamma,\Delta^1_1,\ldots, \Delta^{k_p}_p, \blue{v^1_{p+1}},\ldots, \blue{v^{k_n}_n},\Gamma' \rangle \\
			&=\{\text{Using \ref{b-val} and \ref{r-val4}}\}\\
			&  c\langle \Gamma, \blue{v^1_1},\ldots, v^{k_p}_p, A^1_{p+1},\ldots, A^{k_{p+1}}_{p+1},\ldots, A^1_n,\ldots, A^{k_n}_n,\Gamma' \rangle\langle \Gamma,\Delta^1_1,\ldots, \Delta^{k_p}_p, \blue{v^1_{p+1}},\ldots, \blue{v^{k_n}_n},\Gamma' \rangle\\
			& = \{\text{Using Splitting}\}\\
			&  c\langle \Gamma,\blue{v^1_1},\ldots, v^{k_p}_p,\blue{v^1_{p+1}},\ldots, \blue{v^{k_n}_n},\Gamma' \rangle
		\end{aligned}
	\end{equation}
	In summary, putting everything together, we obtain:
	\begin{align*}
		&t\langle \Gamma,\theta_1,\ldots, \theta_n,\Gamma' \rangle\\
		& = \{\text{Using }\ref{t-val}\}\\
		&  r \langle \Gamma,\Delta^1_1,\ldots, \Delta^{k_p}_p, \blue{u_{p+1}},\ldots, \blue{u_n},\Gamma' \rangle\langle \Gamma,\Delta^1_1,\ldots, \Delta^{k_p}_p, \blue{v^1_{p+1}},\ldots, v^{k_n}_n,\Gamma' \rangle\\
		& = \{\text{Using \ref{r-subval}}\}\\
		& c\langle \Gamma,\blue{v^1_1},\ldots, v^{k_p}_p,\blue{v^1_{p+1}},\ldots, \blue{v^{k_n}_n},\Gamma' \rangle\\
		& = \{\text{Substituting the value of $c$}\}\\
		& =t\langle \Gamma,\blue{u_1},\ldots, u_p,\ldots, \blue{u_n},\Gamma' \rangle\langle \Gamma,\blue{v^1_1},\ldots,\blue{v^{k_1}_1},\ldots, v^1_{p},\ldots v^{k_p}_p,\blue{v^1_{p+1}},\ldots, \blue{v^{k_n}_n},\Gamma' \rangle
	\end{align*}
	as required.
\end{proof}

The proof above seems rather dense. The idea is this: break the initial term into 3 parts. The first part contains the first $p-1$ terms, the middle part consists of the $p$th term and the third part contains the last $n-p$ terms. Lemma \ref{multicat-assoc} can be applied on the first $p-1$ terms and the last $n-p$ terms. The middle term is also split using Lemma \ref{isub-assoc}. By centrality, everything is then shifted around to obtain the required expression.

Now, we are ready to define the the operations $psub$ and $esub_j$ in an effectful multicategory. 
\begin{definition}\label{def}
	Let $J: \bb{C}_0\to \bb{C}_1$ be an effectful category. Keeping with the convention, we denote terms in $\bb{C}_0$ in blue, and terms in $\bb{C}_1$ in red. 
	Then, define the operations
	\begin{align*}
		psub: \Mhom{\bb{C}_0}{A_1,\ldots, A_n}{B}\times \prod_{i=1}^{n} \Mhom{\bb{C}_0}{\Delta_i}{A_i}&\to \Mhom{\bb{C}_0}{\Delta_1,\ldots, \Delta_n}{B}\\
		\blue{t}, \blue{u_1},\ldots, \blue{u_n} &\mapsto  \blue{t}\langle \blue{u_1},\ldots, \blue{u_n}\rangle\\
	\end{align*}
	\begin{align*}
		esub_j: \Mhom{\bb{C}_1}{A_1,\ldots, A_n}{B}\times \prod_{i=1}^{n} \Mhom{\bb{C}_{\delta_{i,j}}}{\Delta_i}{A_i}&\to \Mhom{\bb{C}_0}{\Delta_1,\ldots, \Delta_n}{B}\\
		\red{t}, \blue{u_1},\ldots, \red{u_j},\ldots, \blue{u_n} & \mapsto \red{t}\{ \blue{u_1},\ldots, \red{u_j},\ldots, \blue{u_n}\} \\&:= \red{t}\langle J(\blue{u_1}),\ldots, \red{u_j},\ldots, J(\blue{u_n})\rangle  
	\end{align*}
\end{definition}

\begin{lemma}
	There is a functor
	\[ 
	G: \category{EffMultiCat} \to ([\to,\category{Set},\times, \fun])-\category{MultiCat}
	\]
	which is an inverse to the functor $F$ from Lemma \ref{functor-f}.
\end{lemma}
\begin{proof}
	We set 
	\begin{enumerate}
		\item $\ob{G\bb{C}} : = \ob{\bb{C}}$
		\item For $\Gamma,A$, set
		\begin{align*}
			\Mhom{G\bb{C}}{\Gamma}{A}_0 &: = \Mhom{\bb{C}_0}{\Gamma}{A}\\
			\Mhom{G\bb{C}}{\Gamma}{A}_1 &: = \Mhom{\bb{C}_1}{\Gamma}{A}\\
			\Mhom{G\bb{C}}{\Gamma}{A}_\diamond &: = J_{\Gamma;A}
		\end{align*}
		\item Identities are the same, i.e, $\blue{id_A}  = id_A^{\bb{C}_0}$.
		\item $psub$ and $esub_j$ are as in Definition \ref{def}
	\end{enumerate}
	
	Verifying that the axioms in Lemma \ref{unwrapped} is straightforward, and use facts like Lemma \ref{multicat-assoc}, \ref{freyd-assoc}, \ref{isub-pres}. If $(f_0, f_1): (\bb{C}_0\to \bb{C}_1)\to (\bb{D}_0\to \bb{D}_1)$ is a morphism of effectful multicategories, define $Gf: G\bb{C}\to G\bb{D}$ to be the obvious $([\to,\category{Set}],\times,\fun)$-multicategory morphism. 
	It is clear why the constructions $F$ and $ G$ are inverse - $FG\bb{C}$ and $GF\bb{C}$ return the same structures, and the operations defined are also recovered.
\end{proof}

\section{Algebras for Duoidally Enriched Multicategories}

We follow our general principle for defining the category of algebras for a multi-ary structure. For this, we require a notion of 2-morphism for $(\mcal{V},\otimes,\fun)$-multicategories. As it turns out, this a straightforward enriched generalization of the notion of a transformation of multicategory morphisms.

\begin{definition}
	Let $f,g: \bb{C}\to \bb{D}$ be  morphisms of $(\mcal{V},\otimes,\fun)$-multicategories. A transformation $\eta: f\Rightarrow g$ is a family of arrows, $\{\eta_A: I\to \Mhom{\bb{D}}{fA}{gA}\}_{A\in \ob{\bb{C}}}$ such that the following diagram commutes for all $A_1,\ldots, A_n,B$: 
	\[\begin{tikzcd}[cramped]
		{\Mhom{\bb{C}}{A_1,\ldots, A_n}{B}\otimes (\fun_{i=1}^n I)} && {\Mhom{\bb{D}}{gA_1,\ldots, gA_n}{gB}\otimes (\fun_{i=1}^n \Mhom{\bb{D}}{fA_i}{gA_i})} \\
		\\
		{\Mhom{\bb{C}}{A_1,\ldots, A_n}{B}} && {\Mhom{\bb{D}}{fA_1,\ldots, fA_n}{gB}} \\
		\\
		{I\otimes \Mhom{\bb{C}}{A_1,\ldots, A_n}{B}} && {\Mhom{\bb{D}}{fB}{gB}\otimes \Mhom{\bb{D}}{fA_1,\ldots, fA_n}{fB}}
		\arrow["coh", from=3-1, to=1-1]
		\arrow["coh"', from=3-1, to=5-1]
		\arrow["{g\otimes (\fun_{i=1}^n \eta_{A_i})}", from=1-1, to=1-3]
		\arrow["comp", from=1-3, to=3-3]
		\arrow["comp"', from=5-3, to=3-3]
		\arrow["{\eta_B\otimes f}"', from=5-1, to=5-3]
	\end{tikzcd}\]
\end{definition}

\begin{proposition}
	Given $\mcal{V}$-multicategories $\bb{C}, \bb{D}$, the following data determines a category, $\Hom{(\mcal{V},\otimes, \fun)\text{-}\category{Mult}}{\bb{C}}{\bb{D}}$
	\begin{enumerate}
		\item Objects: $\mcal{V}$-multicategory morphisms, $f: \bb{C}\to \bb{D}$
		\item Morphisms: Transformations, $\eta: f\Rightarrow g$. 
		\item Identities: For $f: \bb{C}\to \bb{D}$, the identity transformation has as components identity morphisms, $id_{fA}$.
		\item Composition: Given $\eta: f\Rightarrow g$ and $\epsilon: g\Rightarrow h$, define their composite, $\epsilon\ast \eta$ to have as components:
		\[
		(\epsilon\ast \eta)_A: I\xrightarrow{\cong} I\otimes I \xrightarrow{\epsilon_A\otimes\eta_A} \Mhom{\bb{D}}{gA}{hA}\otimes \Mhom{\bb{D}}{fA}{gA} \xrightarrow{comp} \Mhom{\bb{D}}{fA}{hA}
		\] 
	\end{enumerate}
\end{proposition}
\begin{proof}
As with the other multi-ary structures we have dealt with so far, this involves checking that the transformations are closed under composition. This is done by straightforward  diagram chasing involving coherence, associativity and definitions.
\end{proof}

We recover the notion of a transformation of multicategory morphisms in the case $(\mcal{V},\otimes, \fun) = (\category{Set},\times, \times)$. 
We also obtain a notion of a transformation between $([\to,\category{Set}],\times, \fun)$-multicategories, allowing us to define a notion of algebra for them.

\begin{definition}
	For an effectful multicategory $\bb{C}$, define its category of algebras in another effectful multicategory $\bb{D}$  to be the category $\Hom{([\to,\category{Set}],\times,\fun)\text{-}\category{MultiCat}}{\bb{C}}{\bb{D}}$.
\end{definition}

Unfortunately, there does not seem to be a canonical example of $\bb{D}$ in which  models can be considered in, and we leave this exploration to future work.

\chapter{Conclusions and Further Vistas}

This dissertation has explored the idea of extending the concepts of universal algebra to structures in effectful computation. Following a lengthy exploration of how concepts in universal algebra are applied in pure computation through the theory of clones and multicategories, we extract a general principle that allows one to define the category of algebras for any multi-ary structure.  With this principle, we sought out definitions of 2-morphisms for premulticategories, and for effectful multicategories. 

To do this, we construct the funny tensor on $[\to,\mcal{V}]$, which then provided the motivation for our novel notion of a duoidally enriched multicategory. These structures generalize both multicategories and effectful multicategories at the same time, and hence provide a link between pure and effectful computation. We use this link to transport the definition of a transformation from theory of multicategories to the theory of effectful multicategories. 

However, there are many more avenues to be explored. Our general principle for defining a category of algebras for a multi-ary structure, we propose the requirement that there be a canonical example of the structure in which models can be considered. In pure computation, this canonical structure is provided by \category{Set}. However, the properties of \category{Set} do not make it very well-suited for studying effectful computation. An alternative might be to study models in the effectful multicategory of stateful functions \cite{staton-levy}. However, as pointed out in Chapter 4, we are required to pick out a set of states beforehand, which leads to the question of canonicity. A reasonable solution might be to study models into the effectful multicategory of stateful functions with a countably infinite set of states, i.e, $S = \mathbb{N}$.

There are other very interesting questions that might be worthwhile to explore, and we list them out.

\subsection*{Adding products to an effectful language}

It was observed in Chapter 4 that it was significantly harder to construct left adjoints to the functors from premonoidal categories/effectful categories to premulticategories/effectful multicategories. The reason for this is that freely adding products to an effectful language is not straightforward. This is the case since a pair $(a,b)$ cannot be added into the language without specifying its order of evaluation. In the pure setting, $n$-tuples $(a_1,\ldots, a_n)$ were added without thought on the order of evaluation. However, for effectful languages we have to add $n$-tuples $(a_1,\ldots, a_n)$, each with $n!$ orders of evaluation, and then quotient out by a suitable relation to take care of centrality. Doing this abstractly, with hom-sets will then tell us how to freely add products to effectful languages. 

The functor from monoidal categories to multicategories is monadic, which means that product types on a linear language  are an algebraic structure on the language itself. Whether  product types on effectful languages  are an algebraic structure remains to be seen, and verification of this has to start by examining the functors from premonoidal categories/effectful categories to premulticategories/effectful multicategories.

\subsection*{Free constructions}

A multigraph is a multi-ary version of a graph, i.e, it consists of a set of objects and for every sequence of objects, $A_1,\ldots, A_n,B$, a set of arrows, $\Mhom{\bb{C}}{A_1,\ldots, A_n}{B}$ is given. These form a category \category{MultiGraph}. Evidently, there is a forgetful functor from \category{MultiCat} to \category{MultiGraph}, and Leinster shows that this functor has a left adjoint. 

We can enriched multigraphs to obtain a category $\mcal{V}$-\category{MultiGraph}, and we obtain a forgetful functor, $U: (\mcal{V},\otimes, \square)\text{-}\category{MultiCat} \to \mcal{V}\text{-}\category{MultiGraph}$. Whether this functor is monadic, and has a left adjoint remains to be seen.

	\subsection*{The substitution product}
	
	Single sorted clones and multicategories, which are typically called clones and operads can be seen as monoids in a certain category equipped with the \textit{substitution product}, $\odot$ \cite{coend}. Whether one-object $(\mcal{V},\otimes, \square)$-multicategories also arise in a similar fashion is still an open question.

	\subsection*{$T$-Multicategories}
	
	The notion of a $T$-multicategory generalizes the notion of a multicategory in a different manner, as a monad in a suitable bicategory of spans \cite{leinster}. Ordinary multicategories and categories can all be viewed as instances of this general framework.  The simultaneous substitution operation of a multicategory enables this sort of generalization, and this  prevented premulticategories from being viewed similarly. Since $(\mcal{V},\otimes, \square)$-multicategories have something that looks like a simultaneous substitution operation, the question of whether these serve as examples of $T$-multicategories is an interesting one.

\appendix
 \begin{appendix}

\chapter{Proofs From Section \ref{6.2.2}}
\label{appendix1}

We give proofs for  Lemmas \ref{unwrap-1} and \ref{unwrap-2} that are used to obtain a description of a $([\to,\category{Set}],\times,\fun)$-multicategory.

\subsection*{Proof  of Lemma \ref{unwrap-1}}

\begin{proof}
	To give a map, $a\times (\fun_{i=1}^n b^i)\to c$ in $[\to,\category{Set}]$ is to give a commutative square,
	\begin{equation}\label{comp-comm}
		\begin{tikzcd}
			{a_0\times (\prod_{i=1}^n b_0^i)} &&& {a_1\times (\fun_{i=1}^n b^i)_1} \\
			{c_0} &&& {c_1}
			\arrow["{f_0}", from=1-1, to=2-1]
			\arrow["{a_\diamond \times (\fun_{i=1}^n b^i)_\diamond}", from=1-1, to=1-4]
			\arrow["{f_1}", from=1-4, to=2-4]
			\arrow["{c_\diamond}", from=2-1, to=2-4]
		\end{tikzcd}
	\end{equation}
	However, using the pushout characterization of the $n$-ary funny tensor, and the fact that $a_1\times - $ preserves colimits, the following is a pushout diagram, and hence, to give a map, $f_1: a_1\times (\fun_{i=1}^n b^i)_1 \to c_1$ is to give maps, $f_{1,j}: a_1\times (\fun_{i=1}^n b^i_{\delta_{i,j}})$ for all $j \in \{1,\ldots, n\}$, such that $f_{1,j} =  f_1\circ (1\times k_j) $.
	\[\begin{tikzcd}
		& {a_1\times (\prod_{i=1}^n b^i_{\delta_{i,1}})} \\
		\\
		{a_1\times (\prod_{i=1}^n b^i_0)} & \vdots & {a_1\times (\fun_{i=1}^n b^i)_1} \\
		\\
		& {a_1\times (\prod_{i=1}^n b^i_{\delta_{i,n}})}
		\arrow["{1\times (\prod_{i=1}^n 1| b^1_1)}"{description}, from=3-1, to=1-2]
		\arrow["{1\times (\prod_{i=1}^n 1| b^n_1)}"{description}, from=3-1, to=5-2]
		\arrow["{1\times k_1}"{description}, from=1-2, to=3-3]
		\arrow["{1\times k_n}"{description}, from=5-2, to=3-3]
	\end{tikzcd}\]
	Further, by the pushout characterization of the $n$-ary tensor, we know that 
	\[
	(\fun_{i=1}^n b^i)_\diamond = k_1\circ ( (\prod_{i=1}^n 1| b^j_1))
	\]
	Hence, to make diagram \ref{comp-comm} commute, we require 
	\begin{align*}
		f_{1}\circ (a_\diamond \times (\fun_{i=1}^n b^i)_\diamond) & = c_\diamond \circ f_0\\
		f_1\circ (1\times (\fun_{i=1}^n b^i)_\diamond)\circ (a_\diamond \times 1) & = c_\diamond \circ f_0\\
		f_1\circ (1\times k_1)\circ (1\times (\prod_{i=1}^n 1| b^j_1))\circ (a_\diamond \times 1) & = c_\diamond \circ f_0\\
		f_{1,j}\circ (1\times (\prod_{i=1}^n 1| b^j_1))\circ (a_\diamond \times 1) & = c_\diamond \circ f_0\\
		f_{1,j}\circ (a_\diamond \times(\prod_{i=1}^n 1| b^j_1) ) & = c_\diamond \circ f_0
	\end{align*}
\end{proof}

\subsection*{Proof of Lemma \ref{unwrap-2}}

\begin{proof}
	We use the abbreviation: 
	\[
	c^i : = \fun_{j=1}^{k_i} c^{i,j}
	\] 
	To  give $f^i, g, x, y$ is equivalent to giving $f^i_0, f^i_1, g_0, g_1,x_0,x_1,y_0,y_1$ such that the following squares are commutative:
	Now, to satisfy the conditions, we need the respective left and right composites of the following two squares to be equal: 

	\[\begin{tikzcd}
		{\substack{a_0\times  (\prod_{i=1}^n (b^i_0\times (\prod_{j=1}^{k_i}c^{i,j}_0)))}} &&& {\substack{a_1\times  (\fun_{i=1}^n (b^i\times (\fun_{j=1}^{k_i} c^{i,j})))_0}} \\
		{\substack{a_0\times (\prod_{i=1}^n e^i_0)}} &&& {\substack{a_1\times (\fun_{i=1}^n e^i)_1}} \\
		{\substack{q_0}} &&& {\substack{q_1}} \\
		\\
		{\substack{a_0\times  (\prod_{i=1}^n (b^i_0\times (\prod_{j=1}^{k_i}c^{i,j}_0)))}} &&& {\substack{a_1\times  (\fun_{i=1}^n (b^i\times (\fun_{j=1}^{k_i} c^{i,j})))_0}} \\
		{\substack{(a_0\times (\prod_{i=1}^n b^i_0))\times (\prod_{i=1}^n \prod_{j=1}^{k_i} c^{i,j}_0)}} &&& {\substack{(a_1 \times (\fun_{i=1}^n b^i)_1)\times (\fun_{i=1}^n \fun_{j=1}^{k_i}c^{i,j})_1}} \\
		{\substack{d_0\times (\prod_{i=1}^n \prod_{j=1}^{k_i} c^{i,j}_0)}} &&& {\substack{d_1\times (\fun_{i=1}^n \fun_{j=1}^{k_i} c^{i,j})_1}} \\
		{\substack{q_0}} &&& {\substack{q_1}}
		\arrow["{\substack{coh}}", from=5-1, to=6-1]
		\arrow["{\substack{x_0\times 1}}", from=6-1, to=7-1]
		\arrow["{y_0}", from=7-1, to=8-1]
		\arrow["{\substack{a_\diamond \times (\fun_{i=1}^n (b^i\times (\fun_{j=1}^{k_i} c^{i,j})))_\diamond}}", from=5-1, to=5-4]
		\arrow["{\substack{a_\diamond \times (\fun_{i=1}^n b^i)_\diamond \\\times (\fun_{i=1}^n \fun_{j=1}^{k_i}c^{i,j})_\diamond}}", from=6-1, to=6-4]
		\arrow["{\substack{coh}}", from=5-4, to=6-4]
		\arrow["{\substack{d_\diamond \times (\fun_{i=1}^n \fun_{j=1}^{k_i}c^{i,j})_\diamond}}", from=7-1, to=7-4]
		\arrow["{\substack{q_\diamond}}", from=8-1, to=8-4]
		\arrow["{\substack{x_1\times 1}}", from=6-4, to=7-4]
		\arrow["{\substack{y_1}}", from=7-4, to=8-4]
		\arrow["{\substack{a_\diamond \times (\fun_{i=1}^n (b^i\times (\fun_{j=1}^{k_i} c^{i,j})))_\diamond}}", from=1-1, to=1-4]
		\arrow["{\substack{1\times (\prod_{i=1}^n f^i_0)}}", from=1-1, to=2-1]
		\arrow["{\substack{g_0}}", from=2-1, to=3-1]
		\arrow["{\substack{q_\diamond}}", from=3-1, to=3-4]
		\arrow["{\substack{g_1}}", from=2-4, to=3-4]
		\arrow["{\substack{1\times (\fun_{i=1}^n f^i)_1}}", from=1-4, to=2-4]
		\arrow["{\substack{a_\diamond \times (\fun_{i=1}^n e^i)_\diamond}}", from=2-1, to=2-4]
	\end{tikzcd}\]
	To ask for the square in \ref{main-square} to commute is equivalent to asking that the composites on the left and the right composites of the top diagram are respectively equal to the composites on the left and the right of the bottom diagram. For the composites on the left to be equal is to make the diagram \ref{first-square} commute. We will show that to for the composites on the right  to be equal is equivalent to \ref{second-square} commuting for all $l,m$.
	
	Denote by $\phi$ and $\psi$ the composites, $g_1\circ (1\times (\fun_{i=1}^n f^i)_1)$ and $y_1\circ (x_1\times 1)\circ coh$ respectively.

	To give the required maps, $f^i_1,g^1,x_1,y_1$ is equivalent to giving maps, $f^i_{1,m}, g_{1,l}, x_{1,l}, y_{1,l,m}$ as indicated by the following diagrams. In all the diagrams below, the square is indeed a pushout square, and this can be seen by the pushout characterization of the $n$-ary funny tensor along with the fact that $a\times -$ or $-\times a$ preserves colimits:
	\[\begin{tikzcd}
		{b^i_1\times \prod_{j=1}^{k_i} c^{i,j}_0 } && {b^i_1\times \prod_{j=1}^{k_i} c^{i,j}_{\delta_{j,k_i}}} \\
		& {\reflectbox{$\ddots$}} \\
		{b^i_1\times \prod_{j=1}^{k_i} c^{i,j}_{\delta_{j,1}}} && {b^i_1\times c^i_1} \\
		&&& {e^i_1}
		\arrow["{1\times \iota^i_{k_i}}"{description}, from=1-3, to=3-3]
		\arrow["{1\times \iota^i_1}"{description}, from=3-1, to=3-3]
		\arrow["{1\times \prod_{j=1}^{k_i}1| c^{i,k_i}_\diamond}"{description}, from=1-1, to=1-3]
		\arrow["{1\times \prod_{j=1}^{k_i}1| c^{i,1}_\diamond}"{description}, from=1-1, to=3-1]
		\arrow["{f^i_{1,1}}"{description}, curve={height=12pt}, from=3-1, to=4-4]
		\arrow["{f^i_{1,k_i}}"{description}, curve={height=-12pt}, from=1-3, to=4-4]
		\arrow["{f_1^i}"{description}, dashed, from=3-3, to=4-4]
	\end{tikzcd}\]
	
	\[\begin{tikzcd}
		{a_1\times \prod_{i=1}^n e^i_0} && {a_1\times \prod_{i=1}^n e^i_{\delta_{i,n}}} \\
		& {\reflectbox{$\ddots$}} \\
		{a_1\times \prod_{i=1}^n e^i_{\delta_{i,1}}} && {a_1\times (\fun_{i=1}^n e^i)_1} \\
		&&& {q_1}
		\arrow["{1\times \prod_{i=1}^n 1| e^n_\diamond}"{description}, from=1-1, to=1-3]
		\arrow["{1\times \prod_{i=1}^n 1| e^1_\diamond}"{description}, from=1-1, to=3-1]
		\arrow["{1\times u_1}"{description}, from=3-1, to=3-3]
		\arrow["{1\times u_n}"{description}, from=1-3, to=3-3]
		\arrow["{g_1}"{description}, dashed, from=3-3, to=4-4]
		\arrow["{g_{1,1}}"{description}, curve={height=12pt}, from=3-1, to=4-4]
		\arrow["{g_{1,n}}"{description}, curve={height=-18pt}, from=1-3, to=4-4]
	\end{tikzcd}\]
	
	\[\begin{tikzcd}
		{a_1\times \prod_{i=1}^n b^i_0} && {a_1\times \prod_{i=1}^n b^i_{\delta_{i,n}}} \\
		& {\reflectbox{$\ddots$}} \\
		{a_1\times \prod_{i=1}^n b^i_{\delta_{i,1}}} && {a_1\times (\fun_{i=1}^n b)_1} \\
		&&& {d_1}
		\arrow["{1\times \prod_{i=1}^n 1| b^n_\diamond}"{description}, from=1-1, to=1-3]
		\arrow["{1\times \prod_{i=1}^n 1| b^1_\diamond}"{description}, from=1-1, to=3-1]
		\arrow["{1\times v_1}"{description}, from=3-1, to=3-3]
		\arrow["{1\times v_n}"{description}, from=1-3, to=3-3]
		\arrow["{x_1}"{description}, from=3-3, to=4-4]
		\arrow["{x_{1,1}}"{description}, curve={height=12pt}, from=3-1, to=4-4]
		\arrow["{x_{1,n}}"{description}, curve={height=-12pt}, from=1-3, to=4-4]
	\end{tikzcd}\]
	
	\[\begin{tikzcd}
		{d_1\times \prod_{i=1}^n c^i_0} && {d_1\times \prod_{i=1}^n c^i_{\delta_{i,n}}} \\
		& {\reflectbox{$\ddots$}} \\
		{d_1\times \prod_{i=1}^n c^i_{\delta_{i,1}}} && {d_1\times (\fun_{i=1}^n c^i)_1} \\
		&&& {q_1}
		\arrow["{1\times \prod_{i=1}^n 1| c^n_\diamond}"{description}, from=1-1, to=1-3]
		\arrow["{1\times \prod_{i=1}^n 1| c^1_\diamond}"{description}, from=1-1, to=3-1]
		\arrow["{1\times t_1}"{description}, from=3-1, to=3-3]
		\arrow["{1\times t_n}"{description}, from=1-3, to=3-3]
		\arrow["{y_1}"{description}, from=3-3, to=4-4]
		\arrow["{y_{1,1}}"{description}, curve={height=12pt}, from=3-1, to=4-4]
		\arrow["{y_{1,n}}"{description}, curve={height=-12pt}, from=1-3, to=4-4]
	\end{tikzcd}\]
	
	\[\begin{tikzcd}[column sep=small]
		{d_1\times \prod_{i=1}^n \prod_{j=1}^{k_i} c^{i,j}_0} && {d_1\times \prod_{i=1}^n \prod_{j=1}^{k_i} c^{i,j}_{\delta_{i,l}\cdot\delta_{j,k_l}}} \\
		& {\reflectbox{$\ddots$}} \\
		{d_1\times \prod_{i=1}^n \prod_{j=1}^{k_i} c^{i,j}_{\delta_{i,l}\cdot\delta_{j,1}}} && {d_1\times \prod_{i=1}^n c^i_{\delta_{i,l}}} \\
		&&& {q_1}
		\arrow["{\substack{1\times (\prod_{i=1} 1 |\\ (1\times \prod_{j=1}^{k_l} 1|c^{l,k_l}_\diamond))}}", from=1-1, to=1-3]
		\arrow["{1\times (\prod_{i=1} 1 | (1\times \prod_{j=1}^{k_l} 1|c^{l,1}_\diamond))}"{description}, from=1-1, to=3-1]
		\arrow["{1\times \prod_{i=1}^n 1|q^l_1}"{description}, from=3-1, to=3-3]
		\arrow["{1\times \prod_{i=1}^n 1|q^l_{k_i}}"{description}, from=1-3, to=3-3]
		\arrow["{y_{1,l,1}}"{description}, curve={height=18pt}, from=3-1, to=4-4]
		\arrow["{y_{1,l,k_l}}"{description}, curve={height=-18pt}, from=1-3, to=4-4]
		\arrow["{y_{1,l}}"{description}, dashed, from=3-3, to=4-4]
	\end{tikzcd}\]

	Now note that $\phi$ and $\psi$ are maps out of a pushout square, so these determine and are determined by $\phi_l, \psi_l$ that make the following diagram commute: 
	\[\begin{tikzcd}[column sep=scriptsize]
		& {\substack{a_1\times (\prod_{i=1}^n (b^i_{\delta_{i,1}}\times c^i_{\delta_{i,1}}))}} \\
		{\substack{a_1\times (\prod_{i=1}^n (b^i_0\times c^i_0))}} & \vdots & {\substack{a_1\times (\fun_{i=1}^n b^i\times c^i)_1}} && {\substack{q_1}} \\
		& {\substack{a_1\times (\prod_{i=1}^n (b^i_{\delta_{i,n}}\times c^i_{\delta_{i,n}}))}}
		\arrow[from=2-1, to=1-2]
		\arrow[from=2-1, to=3-2]
		\arrow["{1\times r_1}"{description}, from=1-2, to=2-3]
		\arrow["{1\times r_n}"{description}, from=3-2, to=2-3]
		\arrow["{\substack{\phi_1}}"{description}, curve={height=-24pt}, from=1-2, to=2-5]
		\arrow["{\substack{\phi_n}}"{description}, curve={height=24pt}, from=3-2, to=2-5]
		\arrow["{\substack{\phi}}"{description}, curve={height=-6pt}, from=2-3, to=2-5]
		\arrow["{\substack{\psi_1}}"{description}, curve={height=-12pt}, from=1-2, to=2-5]
		\arrow["{\substack{\psi}}"{description}, curve={height=6pt}, from=2-3, to=2-5]
		\arrow["{\substack{\psi_n}}"{description}, curve={height=12pt}, from=3-2, to=2-5]
	\end{tikzcd}\]
	As a consequence of uniqueness, $\phi =\psi$ if and only if $\phi_l = \psi_l$ for all $l\in \{1,\ldots, n\}$. 
	The following cube can be verified to be commutative with all the information above. 

	
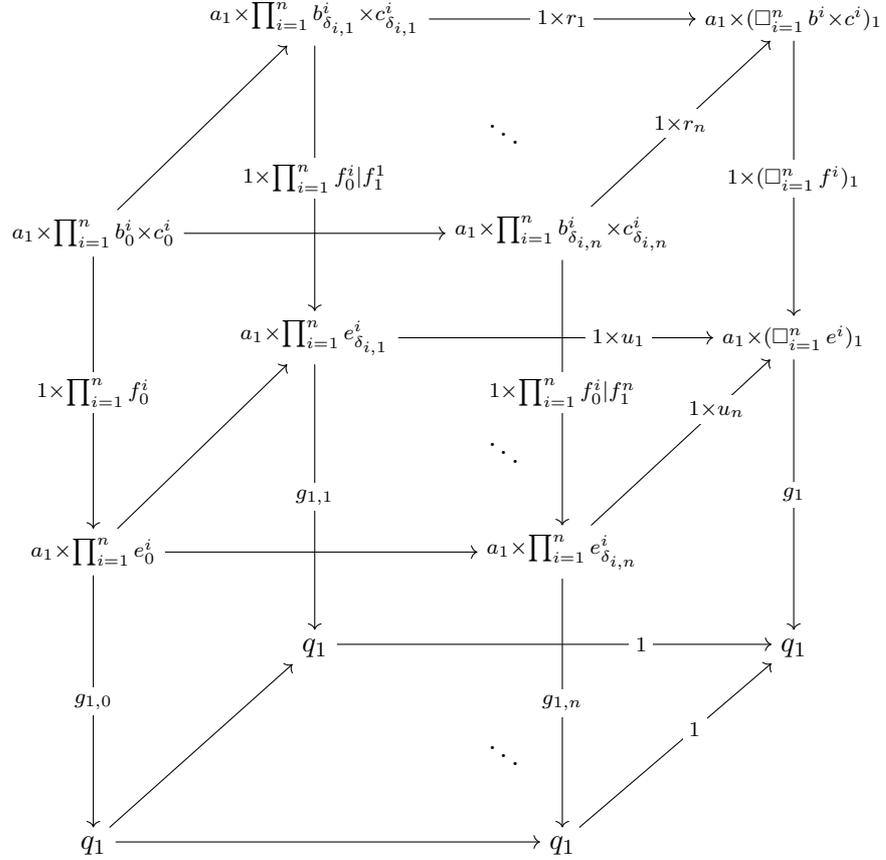
\[\begin{tikzcd}[cramped,column sep=tiny]
		& {\substack{a_1\times\prod_{i=1}^n b^i_{\delta_{i,1}} \times c^i_{\delta_{i,1}}}} && {\substack{a_1\times(\fun_{i=1}^n b^i \times c^i)_1}} \\
		&& \ddots\qquad\qquad \\
		{\substack{a_1\times \prod_{i=1}^n b^i_0 \times c^i_0 }} && {\substack{a_1\times\prod_{i=1}^n b^i_{\delta_{i,n}} \times c^i_{\delta_{i,n}}}} \\
		& {\substack{a_1\times\prod_{i=1}^n e^i_{\delta_{i,1}}}} && {\substack{a_1\times(\fun_{i=1}^n e^i)_1}} \\
		&& \ddots\qquad\qquad \\
		{\substack{a_1\times\prod_{i=1}^n e^i_0}} && {\substack{a_1\times\prod_{i=1}^n e^i_{\delta_{i,n}}}} \\
		& {q_1} && {q_1} \\
		&& \ddots\qquad\qquad \\
		{q_1} && {q_1}
		\arrow[from=3-1, to=1-2]
		\arrow[from=3-1, to=3-3]
		\arrow["{\substack{1\times r_1}}"{description}, from=1-2, to=1-4]
		\arrow["{\substack{1\times r_n}}"{description}, from=3-3, to=1-4]
		\arrow[from=6-1, to=4-2]
		\arrow["{1\times u_1}"{description, pos=0.7}, from=4-2, to=4-4]
		\arrow[from=6-1, to=6-3]
		\arrow["{1\times u_n}"{description, pos=0.7}, from=6-3, to=4-4]
		\arrow["{\substack{1\times (\fun_{i=1}^n f^i)_1}}"{description}, from=1-4, to=4-4]
		\arrow["{\substack{1\times\prod_{i=1}^n  f_0^i | f^1_1}}"{description}, from=1-2, to=4-2]
		\arrow["{\substack{1\times \prod_{i=1}^n  f_0^i | f^n_1}}"{description}, from=3-3, to=6-3]
		\arrow["{\substack{1\times \prod_{i=1}^n f^i_0}}"{description}, from=3-1, to=6-1]
		\arrow["{g_{1,0}}"{description}, from=6-1, to=9-1]
		\arrow[from=9-1, to=7-2]
		\arrow["1"{description, pos=0.7}, from=7-2, to=7-4]
		\arrow[from=9-1, to=9-3]
		\arrow["1"{description, pos=0.6}, from=9-3, to=7-4]
		\arrow["{g_{1,1}}"{description}, from=4-2, to=7-2]
		\arrow["{g_{1,n}}"{description}, from=6-3, to=9-3]
		\arrow["{g_1}"{description}, from=4-4, to=7-4]
	\end{tikzcd}\]
	Hence, 
	\begin{align*}
		\phi_l &= \phi\circ (1\times r_l)\\
		& = g_1\circ (1\times (\fun_{i=1}^n f^i)_1)\circ (1\times r_n)\\
		& = g_{1,n}\circ (1\times \prod_{i=1}^{n}f^i_0| f^l_1)
	\end{align*}
	
	Similarly, the following cube can also be seen to be commutative with all the information presented above:

	
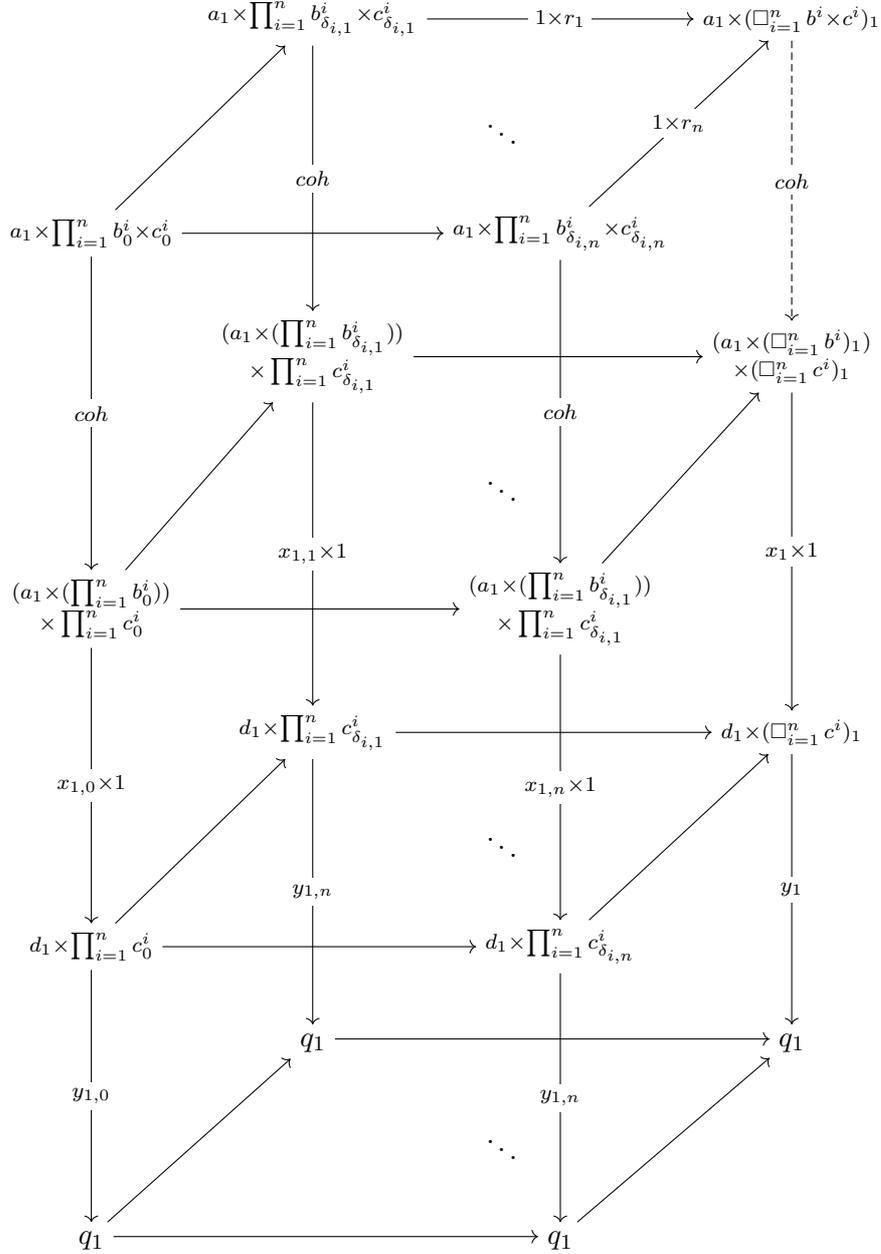
\[\begin{tikzcd}[cramped, column sep=tiny]
		& {\substack{a_1\times\prod_{i=1}^n b^i_{\delta_{i,1}} \times c^i_{\delta_{i,1}}}} && {\substack{a_1\times(\fun_{i=1}^n b^i \times c^i)_1}} \\
		&& \ddots\qquad\qquad \\
		{\substack{a_1\times \prod_{i=1}^n b^i_0 \times c^i_0 }} && {\substack{a_1\times\prod_{i=1}^n b^i_{\delta_{i,n}} \times c^i_{\delta_{i,n}}}} \\
		& {\substack{(a_1\times (\prod_{i=1}^n b^i_{\delta_{i,1}}))\\\times \prod_{i=1}^n c^i_{\delta_{i,1}}}} && {\substack{ (a_1\times (\fun_{i=1}^n b^i)_1)\\\times (\fun_{i=1}^n c^i)_1}} \\
		&& \ddots\qquad\qquad \\
		{\substack{(a_1\times (\prod_{i=1}^n b^i_0))\\\times \prod_{i=1}^n c^i_0}} && {\substack{(a_1\times (\prod_{i=1}^n b^i_{\delta_{i,1}}))\\\times \prod_{i=1}^n c^i_{\delta_{i,1}}}} \\
		& {\substack{d_1\times \prod_{i=1}^n c^i_{\delta_{i,1}}}} && {\substack{d_1\times (\fun_{i=1}^n c^i)_1}} \\
		&& \ddots\qquad\qquad \\
		{\substack{d_1\times \prod_{i=1}^n c^i_{0}}} && {\substack{d_1\times \prod_{i=1}^n c^i_{\delta_{i,n}}}} \\
		& {q_1} && {q_1} \\
		&& \ddots\qquad\qquad \\
		{q_1} && {q_1}
		\arrow[from=3-1, to=3-3]
		\arrow["{1\times r_1}"{description}, from=1-2, to=1-4]
		\arrow["{1\times r_n}"{description}, from=3-3, to=1-4]
		\arrow[from=3-1, to=1-2]
		\arrow["coh"{description}, from=3-1, to=6-1]
		\arrow[from=6-1, to=4-2]
		\arrow[from=4-2, to=4-4]
		\arrow[from=6-1, to=6-3]
		\arrow[from=6-3, to=4-4]
		\arrow["coh"{description}, from=3-3, to=6-3]
		\arrow["coh"{description}, dashed, from=1-4, to=4-4]
		\arrow["coh"{description}, from=1-2, to=4-2]
		\arrow["{x_{1,0}\times 1}"{description}, from=6-1, to=9-1]
		\arrow[from=9-1, to=7-2]
		\arrow[from=7-2, to=7-4]
		\arrow[from=9-3, to=7-4]
		\arrow["{x_{1,1}\times 1}"{description}, from=4-2, to=7-2]
		\arrow["{x_1\times1}"{description}, from=4-4, to=7-4]
		\arrow["{x_{1,n}\times 1}"{description}, from=6-3, to=9-3]
		\arrow[from=9-1, to=9-3]
		\arrow["{y_{1,0}}"{description}, from=9-1, to=12-1]
		\arrow["{y_{1,n}}"{description}, from=7-2, to=10-2]
		\arrow[from=12-1, to=10-2]
		\arrow[from=12-1, to=12-3]
		\arrow[from=10-2, to=10-4]
		\arrow[from=12-3, to=10-4]
		\arrow["{y_{1,n}}"{description}, from=9-3, to=12-3]
		\arrow["{y_1}"{description}, from=7-4, to=10-4]
	\end{tikzcd}\]
	Therefore,
	\begin{align*}
		\psi_l & = \psi\circ (1\times r_l)\\
		& = y_1\circ (x_1\times 1)\circ coh \circ (1\times r_l)\\
		& = y_{1,n}\circ (x_{1,n}\times 1)\circ coh 
	\end{align*}

	Now note that $\phi_l$ and $\psi_l$ are themselves maps out of a pushout square! So, these maps determine and are determined by the maps, $\phi_{l,m}, \psi_{l,m}$ which make the diagram below commute: 
	
	
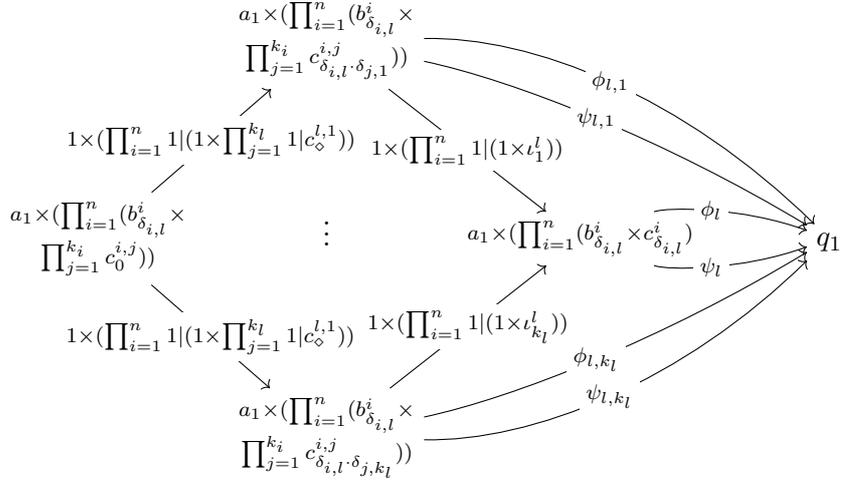
\[\begin{tikzcd}[cramped,column sep=small]
		& {\substack{a_1\times (\prod_{i=1}^n (b^i_{\delta_{i,l}}\times \\\prod_{j=1}^{k_i} c^{i,j}_{\delta_{i,l}\cdot\delta_{j,1}})) }} \\
		\\
		{\substack{a_1\times (\prod_{i=1}^n (b^i_{\delta_{i,l}}\times \\\prod_{j=1}^{k_i} c^{i,j}_0)) }} & \vdots & {\substack{a_1\times (\prod_{i=1}^n (b^i_{\delta_{i,l}}\times c^i_{\delta_{i,l}}) }} &&& {q_1} \\
		\\
		& {\substack{a_1\times (\prod_{i=1}^n (b^i_{\delta_{i,l}}\times\\ \prod_{j=1}^{k_i} c^{i,j}_{\delta_{i,l}\cdot\delta_{j,k_l}})) }}
		\arrow["{\substack{1\times (\prod_{i=1}^n 1 | (1\times \prod_{j=1}^{k_l} 1|c^{l,1}_\diamond))}}"{description}, from=3-1, to=1-2]
		\arrow["{\substack{1\times (\prod_{i=1}^n 1 | (1\times \prod_{j=1}^{k_l} 1|c^{l,1}_\diamond))}}"{description}, from=3-1, to=5-2]
		\arrow["{\substack{1\times  (\prod_{i=1}^n 1| (1\times\iota^l_{1}))}}"{description}, from=1-2, to=3-3]
		\arrow["{\substack{1\times  (\prod_{i=1}^n 1| (1\times\iota^{l}_{k_l}))}}"{description}, from=5-2, to=3-3]
		\arrow["{\phi_l}"{description}, curve={height=-12pt}, from=3-3, to=3-6]
		\arrow["{\phi_{l,1}}"{description}, curve={height=-30pt}, from=1-2, to=3-6]
		\arrow["{\psi_{l,{k_l}}}"{description}, curve={height=30pt}, from=5-2, to=3-6]
		\arrow["{\psi_l}"{description}, curve={height=12pt}, from=3-3, to=3-6]
		\arrow["{\phi_{l,k_l}}"{description}, curve={height=12pt}, from=5-2, to=3-6]
		\arrow["{\psi_{l,1}}"{description}, curve={height=-12pt}, from=1-2, to=3-6]
	\end{tikzcd}\]
	
	As before, $\phi_l = \psi_l$ if and only if $\psi_{l,m} = \phi_{l,m}$ for all $m\in\{1,\ldots, k_l\}$. We use an analogous argument as above. Consider the commutative cube: 

	
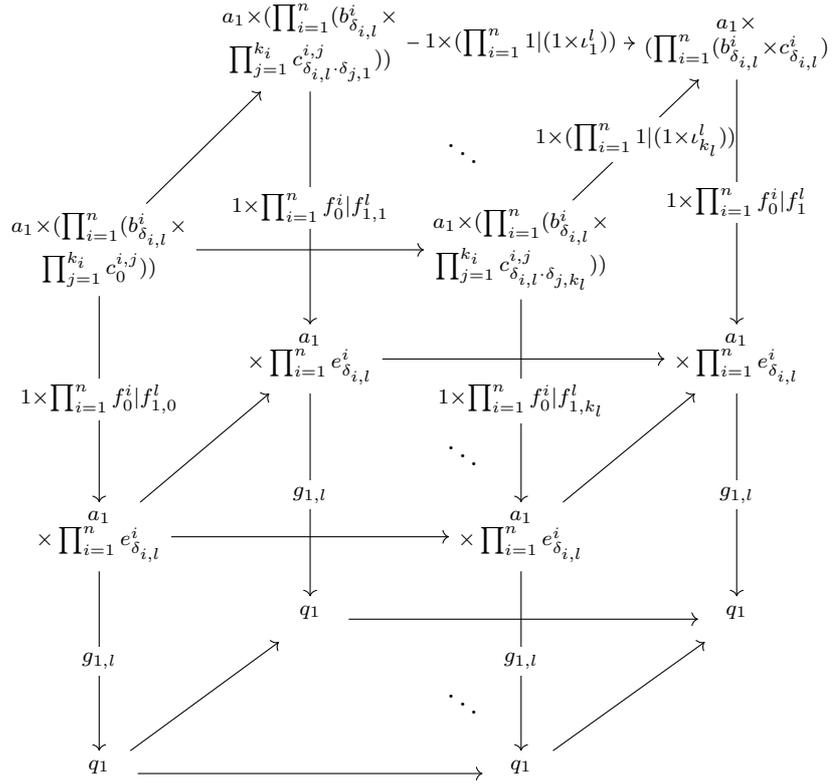
\[\begin{tikzcd}[cramped,column sep=tiny, row sep=small]
		& {\substack{a_1\times (\prod_{i=1}^n (b^i_{\delta_{i,l}}\times \\\prod_{j=1}^{k_i} c^{i,j}_{\delta_{i,l}\cdot\delta_{j,1}})) }} && {\substack{a_1\times\\ (\prod_{i=1}^n (b^i_{\delta_{i,l}}\times c^i_{\delta_{i,l}}) }} \\
		&& \ddots\qquad\qquad \\
		{\substack{a_1\times (\prod_{i=1}^n (b^i_{\delta_{i,l}}\times \\\prod_{j=1}^{k_i} c^{i,j}_0)) }} && {\substack{a_1\times (\prod_{i=1}^n (b^i_{\delta_{i,l}}\times\\ \prod_{j=1}^{k_i} c^{i,j}_{\delta_{i,l}\cdot\delta_{j,k_l}})) }} \\
		& {\substack{a_1\\\times \prod_{i=1}^n e^i_{\delta_{i,l}}}} && {\substack{a_1\\\times \prod_{i=1}^n e^i_{\delta_{i,l}}}} \\
		&& \ddots\qquad\qquad \\
		{\substack{a_1\\\times \prod_{i=1}^n e^i_{\delta_{i,l}}}} && {\substack{a_1\\\times \prod_{i=1}^n e^i_{\delta_{i,l}}}} \\
		& {\substack{q_1\\\qquad}} && {\substack{q_1\\\qquad}} \\
		&& \ddots\qquad\qquad \\
		{\substack{q_1\\\qquad}} && {\substack{q_1\\\qquad}}
		\arrow[from=3-1, to=1-2]
		\arrow[from=3-1, to=3-3]
		\arrow["{\substack{1\times  (\prod_{i=1}^n 1| (1\times\iota^l_{1}))}}"{description}, from=1-2, to=1-4]
		\arrow["{\substack{1\times  (\prod_{i=1}^n 1| (1\times\iota^{l}_{k_l}))}}"{description}, from=3-3, to=1-4]
		\arrow["{\substack{1\times \prod_{i=1}^n f_0^i|f^l_{1,0}}}"{description}, from=3-1, to=6-1]
		\arrow[from=6-1, to=4-2]
		\arrow[from=6-1, to=6-3]
		\arrow[from=4-2, to=4-4]
		\arrow[from=6-3, to=4-4]
		\arrow["{\substack{1\times \prod_{i=1}^n f_0^i|f^l_{1,k_l}}}"{description}, from=3-3, to=6-3]
		\arrow["{\substack{1\times \prod_{i=1}^n f_0^i|f^l_{1,1}}}"{description}, from=1-2, to=4-2]
		\arrow["{\substack{1\times \prod_{i=1}^n f^i_0| f^l_1}}"{description}, from=1-4, to=4-4]
		\arrow[from=9-1, to=7-2]
		\arrow[from=9-1, to=9-3]
		\arrow[from=7-2, to=7-4]
		\arrow[from=9-3, to=7-4]
		\arrow["{\substack{g_{1,l}}}"{description}, from=6-1, to=9-1]
		\arrow["{\substack{g_{1,l}}}"{description}, from=4-2, to=7-2]
		\arrow["{\substack{g_{1,l}}}"{description}, from=6-3, to=9-3]
		\arrow["{\substack{g_{1,l}}}"{description}, from=4-4, to=7-4]
	\end{tikzcd}\]
	
	Hence, 
	
	\begin{align*}
		\phi_{l,m} & = \phi_l \circ (1\times  (\prod_{i=1}^n 1| (1\times\iota^l_{m})))\\
		& = g_{1,l}\circ (1\times \prod_{i=1}^{n}f^i_0| f^l_1)\circ 1\times  (\prod_{i=1}^n 1| (1\times\iota^l_{m})))\\
		& = g_{1,l}\circ (1\times \prod_{i=1}^{n} f^i_0| f^l_{1,k_m})
	\end{align*}
	
	Similarly, consider the cube:
	\[\begin{tikzcd}[cramped,column sep=tiny,row sep=small]
		& {\substack{a_1\times (\prod_{i=1}^n (b^i_{\delta_{i,l}}\times \\\prod_{j=1}^{k_i} c^{i,j}_{\delta_{i,l}\cdot\delta_{j,1}})) }} && {\substack{a_1\times\\ (\prod_{i=1}^n (b^i_{\delta_{i,l}}\times c^i_{\delta_{i,l}}) }} \\
		&& \ddots\qquad\qquad \\
		{\substack{a_1\times (\prod_{i=1}^n (b^i_{\delta_{i,l}}\times \\\prod_{j=1}^{k_i} c^{i,j}_0)) }} && {\substack{a_1\times (\prod_{i=1}^n (b^i_{\delta_{i,l}}\times\\ \prod_{j=1}^{k_i} c^{i,j}_{\delta_{i,l}\cdot\delta_{j,k_l}})) }} \\
		& {\substack{ (a_1\times \prod_{i=1}^n b^i_{\delta_{i,l}}) \\\times(\prod_{i=1}^n\prod_{j=1}^{k_i} c^{i,j}_{\delta_{i,l}\cdot\delta_{j,1}})  }} && {\substack{(a_1\times \prod_{i=1}^n b^i_{\delta_{i,l}})\\\times\prod_{i=1}^n c^i_{\delta_{i,l}}}} \\
		&& \ddots\qquad\qquad \\
		{\substack{ (a_1\times \prod_{i=1}^n b^i_{\delta_{i,l}}) \\\times(\prod_{i=1}^n\prod_{j=1}^{k_i} c^{i,j}_0)  }} && {\substack{ (a_1\times \prod_{i=1}^n b^i_{\delta_{i,l}}) \\\times(\prod_{i=1}^n\prod_{j=1}^{k_i} c^{i,j}_{\delta_{i,l}\cdot\delta_{j,k_l}})  }} \\
		& {\substack{d_1\times\\(\prod_{i=1}^n\prod_{j=1}^{k_i} c^{i,j}_{\delta_{i,l}\cdot\delta_{j,1}})}} && {\substack{d_1\times \prod_{i=1}^nc^i_{\delta_{i,l}}}} \\
		&& \ddots\qquad\qquad \\
		{\substack{d_1\times\\(\prod_{i=1}^n\prod_{j=1}^{k_i} c^{i,j}_0)}} && {\substack{d_1\times\\(\prod_{i=1}^n\prod_{j=1}^{k_i} c^{i,j}_{\delta_{i,l}\cdot\delta_{j,k_l}})}} \\
		& {q_1} && {q_1} \\
		&& \ddots\qquad\qquad \\
		{q_1} && {q_1}
		\arrow[from=3-1, to=1-2]
		\arrow[from=3-1, to=3-3]
		\arrow["{\substack{1\times  (\prod_{i=1}^n 1| (1\times\iota^l_{1}))}}"{description}, from=1-2, to=1-4]
		\arrow["{\substack{1\times  (\prod_{i=1}^n 1| (1\times\iota^{l}_{k_l}))}}"{description}, from=3-3, to=1-4]
		\arrow["coh"{description}, from=3-1, to=6-1]
		\arrow[from=6-1, to=4-2]
		\arrow[from=6-1, to=6-3]
		\arrow[from=4-2, to=4-4]
		\arrow[from=6-3, to=4-4]
		\arrow["coh"{description}, from=1-2, to=4-2]
		\arrow["coh"{description}, from=3-3, to=6-3]
		\arrow["coh"{description}, from=1-4, to=4-4]
		\arrow["{\substack{x_{1,l}\times1}}"{description}, from=6-3, to=9-3]
		\arrow[from=9-1, to=7-2]
		\arrow[from=9-1, to=9-3]
		\arrow[from=9-3, to=7-4]
		\arrow["{\substack{x_{1,l}\times1}}"{description}, from=6-1, to=9-1]
		\arrow[from=7-2, to=7-4]
		\arrow["{\substack{x_{1,l}\times1}}"{description}, from=4-2, to=7-2]
		\arrow["{\substack{x_{1,l}\times1}}"{description}, from=4-4, to=7-4]
		\arrow[from=10-2, to=10-4]
		\arrow[from=12-3, to=10-4]
		\arrow["{\substack{y_{1,l}}}"{description}, from=7-4, to=10-4]
		\arrow["{\substack{y_{1,l,k_l}}}"{description}, from=9-3, to=12-3]
		\arrow["{\substack{y_{1,l,1}}}"{description}, from=7-2, to=10-2]
		\arrow[from=9-1, to=12-1]
		\arrow[from=12-1, to=10-2]
		\arrow[from=12-1, to=12-3]
	\end{tikzcd}\]
	Then, we obtain the equality,
	\begin{align*}
		\psi_{l,m} & = \psi_l \circ (1\times  (\prod_{i=1}^n 1| (1\times\iota^l_{m})))\\
		& = y_{1,l}\circ (x_{1,l}\times 1)\circ coh\circ (1\times  (\prod_{i=1}^n 1| (1\times\iota^l_{m})))\\
		& = y_{1,l,k_l}\circ (x_{1,l}\times 1)\circ coh
	\end{align*}
	
	Consequently, $\psi_{l,m} = \phi_{l,m}$ if and only if the following diagram commutes: 
	\[\begin{tikzcd}[cramped]
		{\substack{a_1\times (\prod_{i=1}^n (b^i_{\delta_{i,l}}\times \prod_{j=1}^{k_i}c^{i,j}_{\delta_{i,l}\cdot\delta_{j,m}}))}} && {\substack{a_1\times \prod_{i=1}^n e^i_{\delta_{i,l}}}} \\
		{\substack{(a_1\times \prod_{i=1}^n b^i_{\delta_{i,l}})\times \prod_{i=1}^n\prod_{j=1}^{k_i} c^{i,j}_{\delta_{i,l}\cdot\delta_{j,m}}}} \\
		{\substack{d_1\times \prod_{i=1}^n\prod_{j=1}^{k_i}c^{i,j}_{\delta_{i,l}\cdot\delta_{j,m}}}} && {\substack{q_1}}
		\arrow["{\substack{1\times \prod_{i=1}^n f^i_0 | f^l_{1,m}}}", from=1-1, to=1-3]
		\arrow["{\substack{coh}}", from=1-1, to=2-1]
		\arrow["{\substack{x_{1,l}\times 1}}", from=2-1, to=3-1]
		\arrow["{\substack{y_{1,l,m}}}", from=3-1, to=3-3]
		\arrow["{\substack{g_{1,l}}}", from=1-3, to=3-3]
	\end{tikzcd}\]
	If this commutes for all $m$, we obtain $\phi_l = \psi_l$, and if this commutes for all $l$, we then obtain $\phi = \psi$, as required. 
\end{proof}
\end{appendix}

\addcontentsline{toc}{chapter}{Bibliography}
\bibliography{refs}           

\begin{thebibliography}{}

\end{thebibliography}


\begin{thebibliography}{10}

\bibitem{vitale}
Jir{\'i} Ad{\'a}mek, Jiř{\'i} Rosick{\'y}, and Enrico~M. Vitale.
\newblock {\em Algebraic Theories: A Categorical Introduction to General
  Algebra}.
\newblock Cambridge University Press, 2010.

\bibitem{2-monoidal}
Marcelo Aguiar and Swapneel Mahajan.
\newblock {\em Monoidal functors, species and Hopf algebras}.
\newblock American Mathematical Society, 2010.

\bibitem{awodey}
S.~Awodey.
\newblock {\em Category Theory}.
\newblock Number~52 in Oxford Logic Guides. Oxford University Press, 2nd
  edition, 2010.

\bibitem{cohn}
P.~M. Cohn.
\newblock {\em Universal Algebra}, volume~6 of {\em Mathematics and its
  applications}.
\newblock Springer Netherlands, 1981.

\bibitem{alg-catkelly}
Frederick Foltz, C.~Lair, and G.~M. Kelly.
\newblock Algebraic categories with few monoidal biclosed structures or none.
\newblock {\em Journal of Pure and Applied Algebra}, 17:171--177, 1980.

\bibitem{act}
Brendan Fong and David~I. Spivak.
\newblock {\em An invitation to applied category theory: Seven sketches in
  Compositionality}.
\newblock Cambridge University Press, 2019.

\bibitem{gould}
M.~Gould.
\newblock {\em Coherence for Categorified Operadic Theories}.
\newblock PhD thesis, University of Glasgow, 2008.

\bibitem{harper}
Robert Harper.
\newblock The holy trinity, Mar 2011.

\bibitem{hefford-roman}
James Hefford and Mario Román.
\newblock Optics for premonoidal categories.
\newblock {\em Electronic Proceedings in Theoretical Computer Science},
  397:152–171, December 2023.

\bibitem{hermida}
Claudio Hermida.
\newblock Representable multicategories.
\newblock {\em Advances in Mathematics}, 151(2):164--225, 2000.

\bibitem{hyland2012}
Martin Hyland.
\newblock Classical lambda calculus in modern dress.
\newblock {\em Mathematical Structures in Computer Science}, 27:762 -- 781,
  2012.

\bibitem{kelly}
G.~M. Kelly.
\newblock {\em Basic Concepts of Enriched Category Theory}.
\newblock Number~64 in London {M}athematical {S}ociety {L}ecture {N}ote
  {S}eries. Cambridge University Press, 1982.

\bibitem{nlab-clone}
G.M. Kelly and A.J. Power.
\newblock Adjunctions whose counits are coequalizers, and presentations of
  finitary enriched monads.
\newblock {\em Journal of Pure and Applied Algebra}, 89(1):163--179, 1993.

\bibitem{ls86}
J.~Lambek and P.J. Scott.
\newblock {\em Introduction to Higher Order Categorical Logic}.
\newblock Cambridge University Press, 1986.

\bibitem{ded-2}
Joachim Lambek.
\newblock Deductive systems and categories ii. standard constructions and
  closed categories.
\newblock In {\em Mathematical systems theory}, 1969.

\bibitem{leinster}
T.~Leinster.
\newblock {\em Higher operads, higher categories}.
\newblock Number 298 in London Mathematical Society Lecture Note Series.
  Cambridge University Press, 2004.

\bibitem{modelling-env}
Paul~Blain Levy, John Power, and Hayo Thielecke.
\newblock Modelling environments in call-by-value programming languages.
\newblock {\em Inf. Comput.}, 185:182--210, 2003.

\bibitem{coend}
Fosco Loregian.
\newblock {\em (Co)End Calculus}.
\newblock Cambridge University Press, 2021.

\bibitem{categories-work}
Saunders Mac~Lane.
\newblock {\em Categories for the working mathematician}.
\newblock Springer, 2010.

\bibitem{manes}
E.~G. Manes.
\newblock {\em Algebraic theories}.
\newblock Springer, 1976.

\bibitem{moggi91}
Eugenio Moggi.
\newblock Notions of computation and monads.
\newblock {\em Information and Computation}, 93(1):55–92, 1991.

\bibitem{premon-alg}
John Power.
\newblock Premonoidal categories as categories with algebraic structure.
\newblock {\em Theoretical Computer Science}, 278(1):303--321, 2002.
\newblock Mathematical Foundations of Programming Semantics 1996.

\bibitem{notions-of-computation}
John Power and Edmund~P. Robinson.
\newblock Premonoidal categories and notions of computation.
\newblock {\em Mathematical Structures in Computer Science}, 7:453 -- 468,
  1997.

\bibitem{roman}
Mario Román.
\newblock Promonads and string diagrams for effectful categories.
\newblock In Jade Master and Martha Lewis, editors, {\em {\rm Proceedings Fifth
  International Conference on} Applied Category Theory, {\rm Glasgow, United
  Kingdom, 18-22 July 2022}}, volume 380 of {\em Electronic Proceedings in
  Theoretical Computer Science}, pages 344--361. Open Publishing Association,
  2023.

\bibitem{psaville}
P.~Saville.
\newblock {\em Cartesian closed bicategories: type theory and coherence}.
\newblock PhD thesis, University of Cambridge, 2019.

\bibitem{staton-levy}
Sam Staton and Paul~Blain Levy.
\newblock Universal properties of impure programming languages.
\newblock In {\em ACM-SIGACT Symposium on Principles of Programming Languages},
  2013.

\bibitem{street}
Ross Street.
\newblock Elementary cosmoi i.
\newblock In Gregory~M. Kelly, editor, {\em Category Seminar}, pages 134--180,
  Berlin, Heidelberg, 1974. Springer Berlin Heidelberg.

\bibitem{wolff}
Harvey Wolff.
\newblock V-cat and v-graph.
\newblock {\em Journal of Pure and Applied Algebra}, 4(2):123--135, 1974.

\end{thebibliography}
\bibliographystyle{plain}  

\end{document}